\documentclass[aps,prx,twocolumn,superscriptaddress,nofootinbib,longbibliography]{revtex4-2}

\usepackage{graphicx} 
\usepackage{amsmath}
\usepackage{amsfonts}
\usepackage{mathrsfs}
\usepackage{amsthm}
\usepackage{bm}
\usepackage{enumitem}
\usepackage{float}
\makeatletter
\let\newfloat\newfloat@ltx
\makeatother
\usepackage{algorithm}
\usepackage{algpseudocode}
\usepackage{xcolor}
\usepackage{hyperref}
\hypersetup{
    colorlinks=true,
    linkcolor=blue,
    citecolor=blue,
    urlcolor=blue
}
\usepackage{subcaption}
\usepackage{ragged2e}

\providecommand{\calH}{\ensuremath{\mathcal{H}}}

\providecommand{\calM}{\ensuremath{\mathcal{M}}}
\providecommand{\calN}{\ensuremath{\mathcal{N}}}
\providecommand{\calO}{\ensuremath{\mathcal{O}}}
\providecommand{\calP}{\ensuremath{\mathcal{P}}}

\providecommand{\calS}{\ensuremath{\mathcal{S}}}

\providecommand{\sfA}{\ensuremath{\mathsf{A}}}

\providecommand{\sfS}{\ensuremath{\mathsf{S}}}

\providecommand{\bbZ}{\ensuremath{\mathbb{Z}}}

\newcommand{\ket}[1]{\left\lvert #1 \right\rangle}
\newcommand{\bra}[1]{\left\langle #1 \right\rvert}

\newcommand{\ketbra}[2]{\left\lvert #1 \right\rangle \! \left\langle #2 \right\rvert}
\newcommand{\ketbrat}[1]{\left\lvert #1 \right\rangle \! \left\langle #1 \right\rvert}
\newcommand{\norm}[1]{\left\lVert#1\right\rVert}
\newcommand{\abs}[1]{\left\lvert#1\right\rvert}
\newcommand{\set}[1]{\left\{ #1\right\}}
\newcommand{\Sym}{\mathrm{Sym}}
\newcommand{\mel}[3]{\left\langle #1 \middle| #2 \middle| #3 \right\rangle}
\DeclareMathOperator{\Tr}{Tr}
\newcommand{\trace}[2]{\Tr_{#1}\left( #2 \right)}
\DeclareMathOperator*{\E}{{\mathbb{E}}}
\newcommand{\Exp}[2]{\E_{#1}\left[ #2 \right]}

\newcommand{\add}{\mathrm{add}}
\newcommand{\rel}{\mathrm{rel}}

\newtheorem{theorem}{Theorem}
\newtheorem{definition}[theorem]{Definition}

\newtheorem{lemma}[theorem]{Lemma}

\newtheorem{corollary}[theorem]{Corollary}

\providecommand{\M}{\ensuremath{\mathfrak{M}}}

\usepackage{tikz}

\definecolor{outline}{HTML}{566779}
\definecolor{ink}{HTML}{1B2738}
\definecolor{palegreen}{HTML}{DEEEE6}
\definecolor{palepink}{HTML}{F3E5E1}

\newcommand{\tableau}[5]{%
  \begin{tikzpicture}[
    x=8mm, y=8mm,
    draw=outline, text=ink,
    line width=0.6pt, font=\Large
  ]
    \fill[#3] (0,0) rectangle (#4,1);
    \fill[#5] (0,-1) rectangle (1,0);

    \draw (0,-1) -- (1,-1) -- (1,0) -- (4,0)
          -- (4,1) -- (0,1) -- cycle;

    \draw (0,0) -- (1,0);
    \foreach \x in {1,2,3}
      \draw (\x,0) -- (\x,1);

    \foreach \entry [count=\i from 0] in {#1}
      \node at ({\i+0.5},0.5) {\entry};
    \node at (0.5,-0.5) {#2};
  \end{tikzpicture}%
}

\begin{document}
\title{Predicting properties of Scrooge ensembles with high accuracy and low sample complexity}

\author{Yue Wu}
\thanks{These authors contributed equally.}
  \affiliation{Chen-Ning Yang Institute for Advanced Study, Tsinghua University, Beijing, 100084, China}
  \affiliation{Beijing Key Laboratory of Cold Atom Quantum Computation, Tsinghua University, Beijing 100084, China}

  \author{Yuzhi Tong}
  \thanks{These authors contributed equally.}
  \affiliation{Chen-Ning Yang Institute for Advanced Study, Tsinghua University, Beijing, 100084, China}
  \affiliation{Beijing Key Laboratory of Cold Atom Quantum Computation, Tsinghua University, Beijing 100084, China}

  \author{Helen Propson}
  \thanks{These authors contributed equally.}
  \affiliation{Institute of Physics, \'Ecole Polytechnique F\'ed\'erale de Lausanne (EPFL), CH-1015 Lausanne, Switzerland}

  \author{Yihui Quek}
  \thanks{
  \href{yihuiquek3.14@gmail.com}{{yihuiquek3.14@gmail.com}}}
  \affiliation{School of Computer and Communication Sciences, \'Ecole Polytechnique F\'ed\'erale de Lausanne (EPFL), CH-1015 Lausanne, Switzerland}
  \affiliation{Institute of Physics, \'Ecole Polytechnique F\'ed\'erale de Lausanne (EPFL), CH-1015 Lausanne, Switzerland}

  \author{Liang Mao}
  \thanks{
  \href{liangmao.physics@gmail.com}{{liangmao.physics@gmail.com}}}
  \affiliation{Chen-Ning Yang Institute for Advanced Study, Tsinghua University, Beijing, 100084, China}
  
\date{\today}
\begin{abstract}
   The Scrooge ensemble associated with a density matrix $\rho$ is the maximally random ensemble whose average density matrix is $\rho$. 
   It has attracted growing interest for its applications to quantum many-body dynamics and quantum information. 
   The intrinsic statistical properties of the Scrooge ensembles are encoded in their higher-order moments, yet they are nontrivial to analyze and prepare for they feature a non-polynomial Haar integrand, blocking direct use of Haar integration machinery. Existing physical mechanisms to sample multiple copies of states from the Scrooge ensembles, meanwhile, typically require costly postselection, rendering them inefficient for large system size. 
In this work, we develop efficient and accurate methods for analyzing and implementing Scrooge moments. We first introduce a polynomial approximation to the Scrooge $k$-th moments whose error is tunable down to a threshold that improves from inverse polynomial to exponential in $1/||\rho||_\infty$ over previous results. Building on our approximation, we then give the first efficient quantum algorithm that estimates observable expectation values using $\widetilde{\mathcal{O}}(k^3/\epsilon^5)$ copies of $\rho$ for any target error $\epsilon$ above this threshold.
   We further provide a construction of a block encoding of the Scrooge $k$-th moment when the preparation circuit of $\rho$ is known. Of independent interest could be a key subroutine in our algorithms, which efficiently implements projection onto the $k$-fold symmetric subspace for $\rho^{\otimes k}$ using $\mathcal{O}(k^3/\epsilon^2)$ samples of $\rho$, avoiding the naive $k!$ postselection cost of such a projection. These results provide new tools for both the theoretical analysis and practical prediction of properties of Scrooge ensembles, with potential applications to quantum many-body dynamics and quantum information processing.
\end{abstract}
\maketitle

\section{Introduction}

Haar-random states are a canonical model of quantum randomness in many-body physics and quantum information. They model states generated by chaotic many-body dynamics~\cite{deutsch1991quantum,fisher2023random,srednicki1994chaos,rigol2008thermalization,nahum2017quantum,roberts2017chaos,cotler2017chaos,nakata2017efficient,DAlessio2016QuantumChaosEigenstateThermalization,cotler_emergent_2023,shaw2025experimental,ho2022exact} and are conjectured to capture aspects of black holes and quantum gravity~\cite{page1993average, page1994black,hayden2007black,sekino2008fast}. They also underpin quantum device benchmarking~\cite{emerson2005scalable, knill2008randomized, dankert2009exact,mark2023benchmarking}, quantum supremacy demonstrations~\cite{Arute2019QuantumSupremacy,morvan2024phase,google2025observation}, quantum learning~\cite{Huang2020PredictingManyPropertiesQuantum, elben2023randomized, huang2025certifying}, and quantum cryptography~\cite{ ji2018pseudorandom,ananth2022cryptography,kretschmer2023quantum}. Yet Haar randomness is completely structureless, limiting its ability to describe quantum systems subject to nontrivial physical constraints.

Scrooge ensembles provide a structured, physically motivated generalization of Haar random states: the Scrooge ensemble of $\rho$ is the "maximally random" ensemble of pure states whose density matrix is $\rho$~\cite{JRW94}. Specifically,
\begin{align}\label{eq:main:scrooge}
    \mathrm{Sc}(\rho):=\set{\ket{\psi}=\frac{\sqrt{\rho}\ket{\phi}}{\norm{\sqrt{\rho}\ket{\phi}}},\nu(\psi)=D\bra{\phi}\rho\ket{\phi}\nu_\mathrm{H}(\phi)},
\end{align}
where $\nu(\psi)$ is the probability measure, $\nu_\mathrm{H}$ is the Haar measure, and $D$ is the Hilbert-space dimension. Among all pure-state ensembles with first moment $\rho$, the Scrooge ensemble minimizes accessible information~\cite{JRW94}; the same distribution appears in quantum statistical mechanics as the Gaussian adjusted projected measure~\cite{goldstein2006distribution,goldstein2016universal}. More recently, Scrooge ensembles have emerged as a universal description of quantum deep thermalization and Hilbert-space ergodicity~\cite{mark2024maximum,mcginley2025scroogeensemblemanybodyquantum,chang2025deep,mok2026nature, mark2026deep, yan2026characterizing,mark2026deep,liu2026observation,manna2025projected,chang2025deep,varikuti2024unraveling,liu2024deep,bejan2025matchgate,lucas2023generalized,mark2026deep,liu2026emergence,liu2026conditional,wu2026exact,pawlik2026restricted}, with applications to shadow tomography~\cite{tran2023measuring, innocenti2023shadow, mcginley2023shadow}, quantum simulator benchmarking~\cite{mark2023benchmarking, Choi2023PreparingRandom, mcginley2025scroogeensemblemanybodyquantum}, and approximate quantum error correction~\cite{Venkatesh2026ErrorCorrection}.

Together, these developments highlight the growing importance of Scrooge ensembles in quantum many-body physics and quantum information science. 
This motivates the basic computational task of accessing Scrooge ensembles, namely, estimating the properties of their higher moment operators, quantities that encode the intrinsic statistical properties.
Yet, an algorithmic way to estimate such properties remains elusive. 
Known mechanisms for the emergence of Scrooge ensembles in many-body systems fall broadly into two classes, neither of which can be readily exploited for the sampling task: projected ensembles obtained by measuring part of an entangled quantum state~\cite{chang2025deep, mark2024maximum, mok2026nature, mark2026deep}, and temporal ensembles formed by sampling states along a system's time evolution~\cite{mark2024maximum, mcginley2025scroogeensemblemanybodyquantum,mok2026nature}. For projected ensembles, preparing multiple copies of the same sampled state generally incurs a substantial postselection cost that recently proved to be inevitable~\cite{du2026no}. 
Temporal ensembles avoid this postselection overhead, but accurately realizing the desired distribution can instead require quantum evolution over extremely long timescales, set by the inverse minimum energy gap. 

In this work, we systematically address this gap by presenting the first efficient, high-accuracy quantum algorithm for estimating observable expectation values of Scrooge $k$-th moments using copies of the unknown state $\rho$. 
The starting point for our investigation is a better approximation of the Scrooge $k$-th moment operator. 
While Haar moment operators have a pristine analytical formulation, their Scrooge counterparts take a less familiar form:
\begin{align}
    \label{eq-main:scrooge-moment0}
    \chi_{\mathrm{Sc}}^{(k)}(\rho)=\Exp{\psi\sim\mathrm{Haar}}{\frac{D\left(\sqrt{\rho}\ketbrat{\psi}\sqrt{\rho}\right)^{\otimes k}}{\bra{\psi}\rho\ket{\psi}^{k-1}}}.
\end{align}
The denominator of Eq. \ref{eq-main:scrooge-moment0} directly blocks the use of Haar-averaging machinery to compute $k$-fold averages over the Scrooge ensemble. 
This  motivates us to come up with good approximations of Eq. \ref{eq-main:scrooge-moment0} that could be more amenable to analytical treatment. Ref. \cite{mcginley2025scroogeensemblemanybodyquantum} was the first to tackle this question for general $k$, putting forth an approximation that amounted to setting the denominator of Eq. \ref{eq-main:scrooge-moment0} to 1. This resulted in relative error scaling roughly as $\norm{\rho}_{\infty}^{1/2}$, where $\norm{\rho}_{\infty}$ is the largest eigenvalue of $\rho$. This approximation therefore becomes accurate only when $\rho$ is sufficiently mixed.

Our first contribution is a polynomial approximation to the Scrooge $k$-th moment operator that achieves an exponential improvement in accuracy.
Our method replaces the inverse denominator of Eq.~\ref{eq-main:scrooge-moment0} by a tunable polynomial approximation. 
By increasing the expansion order, the approximation rapidly becomes more accurate, ultimately reaching an error that can be exponentially small in $1/\norm{\rho}_{\infty}$. This substantially extends the regime of states for which Scrooge moments can be accurately approximated. We complement these guarantees with numerical results for physically motivated Gibbs states, where the observed errors are often considerably smaller than our theoretical upper bounds.

Building on this approximation method, we construct an efficient, Schur transform-based quantum algorithm~\cite{harrow2005applications, burchardt2025high, bacon2006efficient} for estimating observable expectation values of the Scrooge $k$-th moments. 
The number of samples of $\rho$ required is $\widetilde{\calO}(k^3/\epsilon^5)$ for any error $\epsilon$ above the prescribed threshold. A key ingredient of our algorithm is an efficient procedure for projecting $\rho^{\otimes k}$ onto the fully symmetric subspace, requiring only $\mathcal{O}(k^3/\epsilon^2)$ copies of $\rho$. 
This avoids the postselection cost of a naive projection-and-postselection method, which incurs a tiny success probability about $1/k!$.
Furthermore,
when a preparation circuit for a purification of $\rho$ is additionally available, we construct a block encoding of the Scrooge $k$-th moment operator, bringing Scrooge moments into the standard block-encoding framework for quantum algorithms. Together, these results provide efficient tools for accessing Scrooge moments under two natural input models, with potential applications to quantum many-body dynamics and quantum information processing.


\textbf{Organization of this paper ---} In Section~\ref{sec:approximation}, we analyze the Scrooge $k$-th moment in detail and develop a polynomial approximation method to address the denominator in its defining expression. Using the concentration properties of Haar-random states, we derive bounds on both the additive and relative approximation errors. We also numerically evaluate the method's performance on Gibbs states of physical interest. A crucial step in translating this method into a quantum algorithm is to implement projection onto the $k$-fold symmetric subspace for $\rho^{\otimes k}$. In Section~\ref{sec:sym-projection}, we introduce Schur–Weyl duality and use it to develop an efficient implementation of this projection. Section~\ref{sec:expectation-value} combines these results to construct an algorithm for estimating observable expectation values and presents its sample and circuit complexities. Section~\ref{sec:be} discusses how to extend this method to construct a block encoding of the Scrooge $k$-th moment. Finally, Section~\ref{sec:discussion} discusses potential applications and concludes this paper.

\section{Approximating the Scrooge $k$-th moment with high accuracy}\label{sec:approximation}

To estimate the observable expectation values of the Scrooge $k$-th moment $\chi_{\mathrm{Sc}}^{(k)}(\rho)$, we first analyze the moment operator in detail. From Eq.~\eqref{eq:main:scrooge}, we have
\begin{align}
\label{eq-main:scrooge-moment}
\chi_{\mathrm{Sc}}^{(k)}(\rho)=\Exp{\phi\sim\mathrm{Haar}}{\frac{D\left(\sqrt{\rho}\ketbrat{\phi}\sqrt{\rho}\right)^{\otimes k}}{\bra{\phi}\rho\ket{\phi}^{k-1}}}.
\end{align}
This expression presents an immediate obstacle to a more precise analysis of $\chi_{\mathrm{Sc}}^{(k)}$: the denominator complicates the Haar average, making it difficult to obtain a closed-form expression for $\chi_{\mathrm{Sc}}^{(k)}$ for arbitrary $\rho$. Indeed, evaluating this average requires knowledge of the full spectrum of $\rho$~\cite{mark2024maximum}. Previous work~\cite{mcginley2025scroogeensemblemanybodyquantum} derived an approximation formula with a fixed error bound of $\calO(k^2\norm{\rho}_{\infty}^{1/2})$. Below, we first develop a refined polynomial approximation method for $\chi_{\mathrm{Sc}}^{(k)}$ that enables access to the high-precision regime.
The approximation features a series-expansion formula with
 error  systematically controlled via the expansion order $\lambda$.  

\subsection{Handling the denominator}
We first outline our strategy to deal with the denominator of the Scrooge moment. For a fixed Hilbert space dimension $D$, we define
\begin{equation}\label{eq:rphi}
    r_\phi := D\bra{\phi}\rho\ket{\phi}.
\end{equation}
The Scrooge $k$-th moment operator Eq.~\eqref{eq-main:scrooge-moment} is
\begin{align}
    \label{eq-text:scrooge-moment}
    \chi_{\mathrm{Sc}}^{(k)}(\rho)=\Exp{\phi\sim\mathrm{Haar}}{\frac{\left(\sqrt{D\rho}\ketbrat{\phi}\sqrt{D\rho}\right)^{\otimes k}}{r_{\phi}^{k-1}}}.
\end{align}
The key idea is to replace $ 1/r_{\phi}^{k-1}$ with a polynomial involving only nonnegative powers of $r_{\phi}$. This would transform Eq.~\eqref{eq-text:scrooge-moment} into a  linear combination of moments of Haar-random states, which we are well-equipped to analyze using existing Haar integration machinery. 

To formulate this polynomial approximation, we ask {\em where} we must most accurately approximate the function, as $r_{\phi}$ takes values in the enormous range $[0,D]$. 
Observe that $\Exp{\phi\sim\mathrm{Haar}}{r_{\phi}} = 1$. Moreover, $r_\phi$ is expected to concentrate around $1$ due to the concentration properties of Haar-random states when $\rho$ has sufficiently low purity.
This suggests seeking an approximation to the function $f(r_{\phi}) = 1/r_{\phi}^{k-1}$ that is almost flat when $r_\phi$ is close to one. Taking $x=r^{k-1}$, we Taylor expand $1/x$ about $x=1$: 
$\frac{1}{x}= 1+(1-x)+(1-x)^2+\cdots$. Truncating the series after $\lambda$ terms and substituting the resulting approximation into Eq.~\eqref{eq-text:scrooge-moment} yields our approximation for the Scrooge $k$-th moment.

\begin{figure*}[t]
    \centering
        \includegraphics[width=\linewidth]{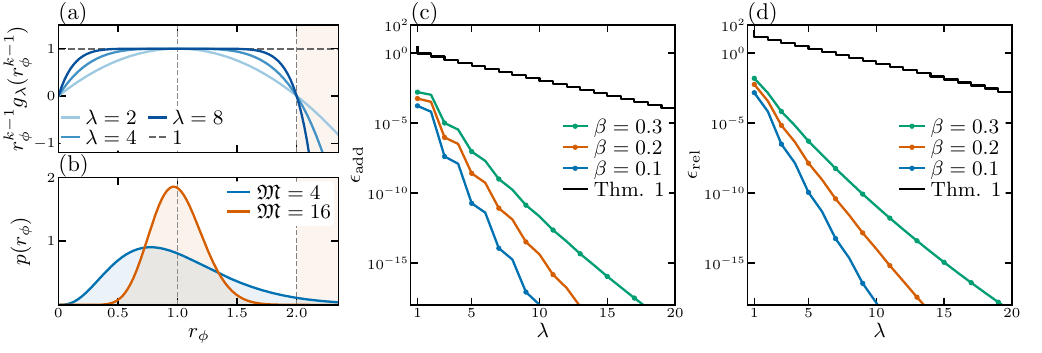}
    \captionsetup{justification=Justified,singlelinecheck=false}
    \caption{Numerical results                                                    of the polynomial approximation for $k=2$.
    (a) Comparison of $g_\lambda(r_\phi^{k-1})$ with $1/r_\phi^{k-1}$ for $\lambda=2,4,8$. The curves show $r_\phi^{k-1}g_\lambda(r_\phi^{k-1})$, and the reference line is the constant $1$.
    (b) Haar probability density of $r_\phi=D\bra{\phi}\rho\ket{\phi}$ for $\rho=\Pi_\M/\M$, where $\Pi_\M$ is a rank-$\M$ projector, with $D=64$ and $\M=4,16$.
    The shaded region $r_\phi>2$ marks the tail where increasing $\lambda$ amplifies the polynomial remainder.
    (c,d) Additive and relative errors, $\epsilon_x=\Delta_{2,\lambda}^{x}$, for Gibbs states of the periodic $N_{\rm site}=12$ transverse-field Ising chain at $J=h=1$ and $\beta=0.1,0.2,0.3$.
    Black staircases show the order prescription in Eq.~\eqref{eq:main-poly-lambda-choice}. We use $C_{\mathrm{add}}=e^{11/12}$ and $C_{\mathrm{rel}}=2^{33/8}e^{19/24}$ from Appendix~\ref{app:approximation}. Numerical values below the vertical-axis limit are omitted.}
    \label{fig:additive-error}
\end{figure*}

\subsection{The polynomial approximation method}
Formally, for an integer
$\lambda\geq 1$, we introduce the series
\begin{align}
    g_{\lambda}(x) &= \sum_{j=0}^{\lambda-1}(1-x)^j \nonumber
    = \frac{1-(1-x)^\lambda}{x}\\
    &=\sum_{i=1}^\lambda\binom{\lambda}{i}(-1)^{i-1} x^{i-1}
\end{align}
where the second equality follows from the geometric series sum. 
Then our $\lambda$-th order approximation is, with $r_{\phi}$ as defined in Eq.~\eqref{eq:rphi},
\begin{align}
    \chi_{\mathrm{PA}}^{(k)}(\lambda)
    &:=
    \Exp{\phi\sim\mathrm{Haar}}{
        \left(\sqrt{D\rho}\ketbra{\phi}{\phi}\sqrt{D\rho}\right)^{\otimes k}
        g_\lambda(r_{\phi}^{k-1})
    }\nonumber\\
    &=
    \sum_{i=1}^{\lambda}
    \binom{\lambda}{i}(-1)^{i-1}
    D^{(i-1)(k-1)} \times
    \nonumber\\
    &\qquad
    \Exp{\phi\sim\mathrm{Haar}}{
        \begin{aligned}
            &\left(\sqrt{D\rho}\ketbra{\phi}{\phi}\sqrt{D\rho}\right)^{\otimes k}\\
            &\quad\times
            \left(
                \Tr\left[\left(\sqrt{\rho}\ketbra{\phi}{\phi}\sqrt{\rho}\right)\right]
            \right)^{(i-1)(k-1)}
        \end{aligned}
    }\nonumber\\
&=    \sum_{i=1}^{\lambda}
    \binom{\lambda}{i}(-1)^{i-1}
    D^{k_i}\Tr_{\mathrm{aux}}\!\left[
    \rho^{\otimes k_i}\chi_{\mathrm{Haar}}^{(k_i)}
    \right],
    \label{eq:pm-decomposition}
\end{align}
where $k_i:=k+a_i=i(k-1)+1$, $a_i:=(i-1)(k-1)$, and
$\Tr_{\mathrm{aux}}$ traces out the $a_i$ auxiliary systems.

Each term in the above linear combination involves at most $\lambda(k-1)+1$ copies of $\rho$ and is written in terms of the product of those copies and a $k_i$-th moment operator of Haar-random states. Moreover, the first-order version ($\lambda=1$) of our approximation exactly recovers the previous method \cite{mcginley2025scroogeensemblemanybodyquantum}, up to a prefactor:
\begin{equation}
\chi_{\mathrm{PA}}^{(k)}(1)
=
\Exp{\phi\sim\mathrm{Haar}}{
\left(\sqrt{D\rho}\ketbra{\phi}{\phi}\sqrt{D\rho}\right)^{\otimes k}},
\end{equation}
Despite the simplicity of our method, this approximation ends up being the same as  the `polynomial modulation' technique of Ref. ~\cite{cotler_emergent_2023}, which was also developed to handle normalization denominators in projected-ensemble moments; 
our discussions above can be regarded as a clarified exposition of this technique.

We quantify the quality of our approximation using two complementary notions of error. The additive error is defined as
\begin{equation}
\Delta_{k,\lambda}^{\mathrm{add}}
:=
\norm{\chi_{\mathrm{Sc}}^{(k)}
-\chi_{\mathrm{PA}}^{(k)}(\lambda)}_1,
\end{equation}
which directly controls the error in the expectation value of any bounded observable. A stronger notion is the relative error, defined as the smallest
$\Delta_{k,\lambda}^{\mathrm{rel}}$ such that
\begin{equation}
(1-\Delta_{k,\lambda}^{\mathrm{rel}})
\chi_{\mathrm{Sc}}^{(k)}
\preceq
\chi_{\mathrm{PA}}^{(k)}(\lambda)
\preceq
(1+\Delta_{k,\lambda}^{\mathrm{rel}})
\chi_{\mathrm{Sc}}^{(k)},
\label{eq:main-relative-sandwich}
\end{equation}
where $A\preceq B$ means $B-A$ is positive semidefinite.

The fluctuation around $r_\phi=1$ can be expressed as $\operatorname{Var}_{\phi\sim\mathrm{Haar}}[r_\phi]=\frac{D\operatorname{Tr}(\rho^2)-1}{D+1}
    \approx \trace{}{\rho^2}$. So the error of our approximation also scales with the norms $\norm{\rho}_2$ or $\norm{\rho}_{\infty}$. For technical convenience, we introduce
    \begin{equation}
    \M := \left\lfloor \frac{1}{\norm{\rho}_\infty} \right\rfloor 
\end{equation}
to characterize the norm.

More precisely,
the truncation error of $g_\lambda$ gives the following bound on the additive error:
\begin{align}
\Delta_{k,\lambda}^{\mathrm{add}}
\leq
\Exp{\phi\sim\mathrm{Haar}}{
r_\phi\left|1-r_\phi^{k-1}\right|^\lambda
}.
\end{align}
For large $\M$, the quantity $r_\phi$ is sharply concentrated around unity, so that $\left|1-r_\phi^{k-1}\right|<1$ for the dominant part of the distribution. Increasing $\lambda$ therefore improves the polynomial approximation in this typical region, suppressing its contribution to the error. However, in the upper tail where $r_\phi^{k-1}>2$, one has $\left|1-r_\phi^{k-1}\right|>1$, and the same power instead amplifies the tail contribution. The resulting competition leads to a turnover of the error bound at sufficiently large $\lambda$. As $\M$ increases, the distribution of $r_\phi$ becomes more concentrated around unity, as shown in Fig.~\ref{fig:additive-error}(b), suppressing the upper tail and allowing higher expansion orders.

In Appendix~\ref{app:approximation}, we identify a range of $\lambda$ over which the error bounds decrease monotonically with increasing $\lambda$, 
\begin{equation}
    1\leq \lambda \leq 
    \left\lfloor
    \frac{\M}{12s_x(k-1)^2}
    \right\rfloor ,
\end{equation}
where $x\in\{\mathrm{add},\mathrm{rel}\}$ and $s_{\mathrm{add}}:=1,
s_{\mathrm{rel}}:=k$. For fixed $k$, the upper limit of this range scales as $\M$. Choosing the largest admissible order then yields an error bound that is exponentially small in $\M$. Consequently, for any target error $\epsilon_{\rm t}$ above a certain threshold, one can choose a sufficiently large admissible $\lambda$ to achieve the desired accuracy. An explicit sufficient choice of $\lambda$ is given below.

\begin{theorem}
[Additive and relative errors of the approximation method]
\label{thm-main:lambda-choice}
Let $k\geq2$ be an integer, $x\in\{\mathrm{add},\mathrm{rel}\}$ and $s_{\mathrm{add}}:=1,
s_{\mathrm{rel}}:=k$.
There exist universal constants $C_{\mathrm{add}},C_{\mathrm{rel}}>0$
such that, whenever $12s_x(k-1)^2\leq\M$, for every target
accuracy $\epsilon_{\rm t}>0$ satisfying
\begin{equation}
\epsilon_{\rm t}\geq
C_x
\exp\left[
-\frac{\M}{24s_x(k-1)^2}
\right],
\label{eq:main-poly-epsilon-range}
\end{equation}
the choice
\begin{equation}
\lambda=
\max\set{
\left\lfloor
2\log\left(C_x/\epsilon_{\rm t}\right)
\right\rfloor,
1}
\label{eq:main-poly-lambda-choice}
\end{equation}
guarantees $\Delta_{k,\lambda}^{x}\leq\epsilon_{\rm t}$.
\end{theorem}\noindent
This theorem combines the Theorems~\ref{thm:additive-error} and~\ref{thm:relative-error} in Appendix~\ref{app:approximation}.
Within this regime, for fixed $k$ and $\M$, Eq.~\eqref{eq:main-poly-lambda-choice} shows that the required $\lambda$ grows only logarithmically with $1/\epsilon_{\rm t}$ within the stated accuracy range. 

\subsection{Numerics}
To validate our theoretical predictions,
we numerically calculate the errors of the approximation for thermal states $\rho_\beta=e^{-\beta H}/Z_\beta$ of the periodic transverse-field Ising chain
\begin{equation}
    H=-J\sum_{j=1}^{N_{\rm site}}Z_jZ_{j+1}-h\sum_{j=1}^{N_{\rm site}}X_j ,
    \label{eq:numerics-ising}
\end{equation}
and compare them with our theoretical bounds.

The additive and relative errors for Scrooge second moment associated with $\rho_\beta$ are plotted against $\lambda$ in Fig.~\ref{fig:additive-error}(c,d).
The black curves show the error prescribed by Theorem~\ref{thm-main:lambda-choice}, and colored curves show the errors obtained from exact numerical calculation.
Here $\M=997,306,116$ for $\beta=0.1,0.2,0.3$, respectively. 
Both errors decrease rapidly over a range of orders and are smaller at higher temperatures (smaller $\beta$) for a fixed $\lambda$. At $\lambda=10$, the additive and relative errors are below $2.2\times10^{-12}$ and $1.2\times10^{-11}$, respectively, for all three temperatures. Thus, a modest expansion order already gives a high-accuracy approximation to the Scrooge second moment in these examples.
The numerical errors are  smaller than the theoretical predictions, indicating that the actual performance of our method could be substantially better than the theoretical bounds for those physical states.

\section{Efficient implementation of symmetric subspace projection}
\label{sec:sym-projection}

\begin{figure*}[t]
    \centering
    \includegraphics[width=\linewidth]{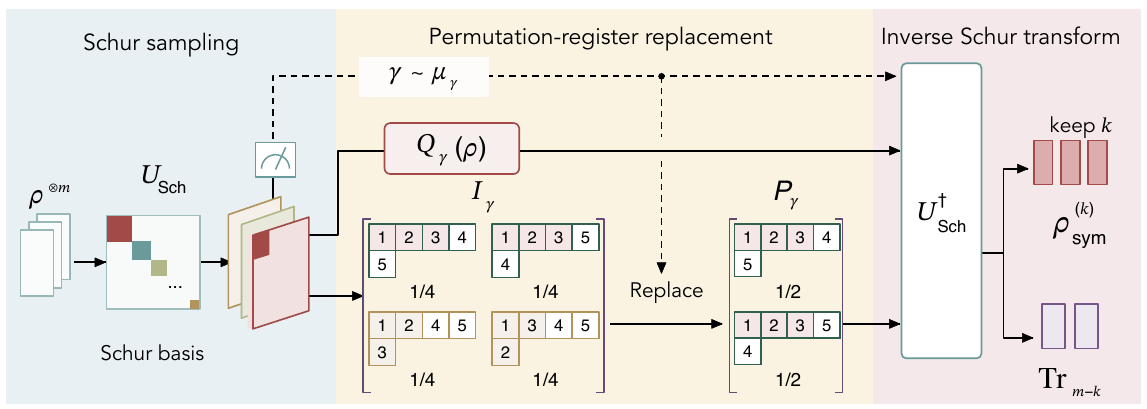}
    \captionsetup{justification=Justified,singlelinecheck=false}
    \caption{
    Schematic illustration of our implementation of symmetric subspace projection.
    The first step applies the Schur transform to $\rho^{\otimes m}$ and measures the sector label, obtaining $\gamma$ with probability $\mu_{\gamma}(\rho)$. Then the maximally mixed permutation register $I_{\gamma}/r_{\gamma}$ is replaced by $P_{\gamma}^{(k)
    }/\trace{}{P_{\gamma}^{(k)}}$, while the unitary group register $Q_{\gamma}(\rho)$ is preserved.  Here $P_\gamma^{(k)}$ selects tableaux whose first $k$ entries lie in the first row. The illustrated sector $\gamma=(4,1)$, with $m=5$ and $k=3$, replaces the uniform mixture over four tableaux with a uniform mixture over two tableaux. The retained label and both registers are supplied to the inverse Schur transform, after which the last $m-k$ systems are discarded. The remaining $k$ systems approximate the target $\rho_{\mathrm{sym}}^{(k)}$, with accuracy controlled by $m$ as stated in Theorem~\ref{Efficient implementation of symmetric subspace projection; informal}. }
    \label{fig:symmetric-preparation-main}
\end{figure*}

Building on this approximation method, we next translate it to a quantum algorithm to estimate expectation values of any observable $O_k$ supported on $k$-fold tensor-product Hilbert space. The key idea is to estimate $\trace{}{\chi_{\mathrm{PA}}^{(k)}(\lambda) O_k}$ as an approximation to the target quantity $\trace{}{\chi_{\mathrm{Sc}}^{(k)}(\rho) O_k}$.
 
Using Eq.~\eqref{eq:pm-decomposition}, we can write
\begin{equation}
\begin{aligned}
    \Tr\!\left[\chi_{\rm PA}^{(k)}(\lambda)O_k\right]
    &=
    \sum_{i=1}^{\lambda}
    \binom{\lambda}{i}(-1)^{i-1}\frac{D^{k_i}}{\binom{D+k_i-1}{k_i}}p_{\mathrm{sym}}^{(k_i)} \\ &\times
    \frac{\Tr\!\left[
        \rho^{\otimes k_i}
        \Pi_{\mathrm{sym}}^{(k_i)}
        \left(
            O_k\otimes I^{\otimes a_i}
        \right)
    \right]}{p_{\mathrm{sym}}^{(k_i)}},
\end{aligned}
\label{eq:main-pm-expectation-decomposition}
\end{equation}
where $\Pi_{\mathrm{sym}}^{(k_i)}$ is the projection onto $k_i$-fold symmetric subspace that arises from taking the Haar average. $p_{\mathrm{sym}}^{(k_i)}=\trace{}{\rho^{\otimes k_i}
        \Pi_{\mathrm{sym}}^{(k_i)}}$ is the normalization factor. Hence, our 
next step is to implement the state $\rho_{\mathrm{sym}}^{(k_i)}=\rho^{\otimes k_i}
        \Pi_{\mathrm{sym}}^{(k_i)}/p_{\mathrm{sym}}^{(k_i)}$.

To prepare
$\rho_{\mathrm{sym}}^{(k_i)}
=\rho^{\otimes k_i}\Pi_{\mathrm{sym}}^{(k_i)}/p_{\mathrm{sym}}^{(k_i)}$
from copies of $\rho$, a natural approach is to take $k_i$ copies
and measure the symmetric subspace projector
$\Pi_{\mathrm{sym}}^{(k_i)}$ using the Schur transform~\cite{harrow2005applications, keyl2001estimating, bacon2006efficient, harrow2013church}.
Conditioned on obtaining the symmetric outcome, the resulting
state is precisely $\rho_{\mathrm{sym}}^{(k_i)}$.
However, this outcome occurs with probability
\begin{align}
p_{\mathrm{sym}}^{(k_i)}
=\mathrm{Tr}\!\left[
\rho^{\otimes k_i}\Pi_{\mathrm{sym}}^{(k_i)}
\right].
\end{align}
For example, when $\rho=I/D$ and $k_i=k$, this probability is
\begin{align}
p_{\mathrm{sym}}^{(k_i)}
=\frac{\binom{D+k-1}{k}}{D^k}
=\frac{1}{k!}\prod_{j=0}^{k-1}\left(1+\frac{j}{D}\right)
=\frac{1+\mathcal{O}(1/D)}{k!}.
\end{align}
Thus, for a highly mixed state in this regime, the expected
number of attempts is approximately $k!$, making direct
post-selection costly as $k$ grows.

To avoid this post-selection cost, we draw ideas from Schur-Weyl duality.

\subsection{Schur-Weyl duality}\label{maintextsec:schur-transform}

The Schur-Weyl duality provides a basis in which the symmetries of a multi-copy Hilbert space are explicitly resolved~\cite{harrow2005applications, keyl2001estimating, bacon2006efficient, harrow2013church}. For $k$ copies of a $D$-dimensional Hilbert space $\mathbb{C}^D$, there are two natural commuting group actions: the collective action $U^{\otimes k}$ of the unitary group $\mathcal U_D$ and the permutation action of the symmetric group $\mathcal{S}_k$ on the $k$ copies. By Schur--Weyl duality, the Hilbert space can therefore be decomposed into sectors labeled by irreducible representations of these two groups,
\begin{align}
(\mathbb{C}^D)^{\otimes k}
\cong
\bigoplus_{\gamma \vdash k,\ \ell(\gamma)\le D}
\mathcal Q_\gamma^D\otimes \mathcal P_\gamma.
\end{align}
Here, $\gamma\vdash k$ denotes a partition of $k$, which labels the corresponding irreducible sector, and $\ell(\gamma)$ is the number of its nonzero parts. The spaces $\mathcal Q_\gamma^D$ and $\mathcal P_\gamma$ carry the irreducible representations of $\mathcal U_D$ and $\mathcal{S}_k$, respectively. 
  
The Schur transform $U_{\mathrm{Sch}}$ maps the computational basis of $(\mathbb{C}^D)^{\otimes k}$ to the Schur basis $\ket{\gamma,q_\gamma,p_\gamma}$, where $q_\gamma$ and $p_\gamma$ label basis states in $\mathcal Q_\gamma^D$ and $\mathcal P_\gamma$, respectively. The Schur transform is known to be efficiently implementable, with circuit depth $\operatorname{poly}\!\left(
k,
\log D,
\log\frac{1}{\epsilon_{\mathrm{Sch}}}
\right)$ to achieve an error $\epsilon_{\mathrm{Sch}}$~\cite{harrow2005applications, burchardt2025high, bacon2006efficient}.

\subsection{Efficient implementation of symmetric subspace projection}
To avoid the post-selection cost, our method exploits the structure
of the target state in the Schur basis. The key observation is that
within each representation sector, the required state of the
permutation register is known and can be prepared directly,
without post-selection.

Suppose we have $\rho^{\otimes m}$, with $m\geq k$, and consider
the target state
\begin{align}\label{eq:main:schur1}
\frac{\Pi_{\mathrm{sym}}^{m\rightarrow k}\rho^{\otimes m}}
{\trace{}{\Pi_{\mathrm{sym}}^{m\rightarrow k}\rho^{\otimes m}}}
=
\rho_{\mathrm{sym}}^{(k)}
\otimes \rho^{\otimes (m-k)},
\end{align}
where we define $\Pi_{\mathrm{sym}}^{k\rightarrow k'}=\Pi_{\mathrm{sym}}^{(k')}\otimes I^{\otimes(k-k')}$.
Tracing out the last $m-k$ registers therefore yields the
desired state $\rho_{\mathrm{sym}}^{(k)}$.

The input state has the following Schur decomposition
\begin{align}\label{eq:main:schur2}
\rho^{\otimes m}
\cong
\bigoplus_{\substack{\gamma\vdash m,\ \ell(\gamma)\leq D}}
\mu_\gamma(\rho)Q_\gamma(\rho)
\otimes \frac{I_\gamma}{r_\gamma},
\end{align}
where $Q_\gamma(\rho)$ is a normalized density matrix,
the non-negative weights $\mu_\gamma(\rho)$ sum to one,
and $r_\gamma$ is the dimension of $\mathcal{P}_\gamma$.
Meanwhile, the projector $\Pi_{\mathrm{sym}}^{k\rightarrow k'}$ also admits a Schur decomposition
\begin{align}\label{eq:main:schur3}
\Pi_{\mathrm{sym}}^{(k')}\otimes I^{\otimes(k-k')}
\cong
\bigoplus_{\substack{\gamma\vdash k, \ell(\gamma)\leq D}}
I_\gamma\otimes P_\gamma^{(k')},
\end{align}
where $P_\gamma^{(k')}$ denotes the representation of $\Pi_\mathrm{sym}^{(k')}$.
Note that $\Pi_{\mathrm{sym}}^{k\rightarrow k'}$ can be represented as a sum of permutation group elements, and hence its representation lies fully on the register $\calP_\gamma$.

Combining Eq.~\eqref{eq:main:schur1}, \eqref{eq:main:schur2}, and \eqref{eq:main:schur3}, we have
\begin{gather}
\rho_{\mathrm{sym}}^{(k)}\otimes\rho^{\otimes(m-k)}\notag\\
\cong
\bigoplus_{\substack{\gamma\vdash m,\ \ell(\gamma)\leq D}}
\frac{\mu_\gamma(\rho)a_\gamma^{(k)}}
{\sum_\gamma\mu_\gamma(\rho)a_\gamma^{(k)}}
Q_\gamma(\rho)\otimes
\frac{P_\gamma^{(k)}}{\trace{}{P_\gamma^{(k)}}},
\label{eq:main:schur4}
\end{gather}
where
$
    a_\gamma^{(k)}
=\frac{\trace{}{P_\gamma^{(k)}}}{r_\gamma}
$
is the probability of obtaining the symmetric subspace
conditioned on the sector $\gamma$. On the other hand, we can prepare an approximation of
$\rho_{\mathrm{sym}}^{(k)}\otimes\rho^{\otimes(m-k)}$
by replacing $\frac{I_\gamma}{r_\gamma}$ with
$\frac{P_\gamma^{(k)}}{\trace{}{P_\gamma^{(k)}}}$
conditioned on the Schur label $\gamma$ whenever
$\trace{}{P_\gamma^{(k)}}>0$; otherwise, we prepare an arbitrary fixed
state in the permutation register and set $a_\gamma^{(k)}=0$. This operation preserves $Q_\gamma(\rho)$ and leaves the information of $\rho$ with no loss. The final state is
\begin{align}\label{eq:main:schur5}
    \bigoplus_{\substack{\gamma\vdash m,\ \ell(\gamma)\leq D}}
\mu_\gamma(\rho)Q_\gamma(\rho)\otimes
\frac{P_\gamma^{(k)}}{\trace{}{P_\gamma^{(k)}}}.
\end{align}
The entire procedure is illustrated in Fig.~\ref{fig:symmetric-preparation-main}.

Bounding the difference between coefficients of Eq.~\eqref{eq:main:schur4} and \eqref{eq:main:schur5}, we demonstrate the efficiency of this method.
\begin{theorem}
[Efficient implementation of symmetric subspace projection; informal]
\label{Efficient implementation of symmetric subspace projection; informal}
    For $n\in\bbZ^+$, $1\leq k<m$,
    and $0<\epsilon<1$, there exists an efficient quantum channel $\mathcal M_{m,k}:
\mathcal L\!\left((\mathbb C^{2^n})^{\otimes m}\right)
\longrightarrow
\mathcal L\!\left((\mathbb C^{2^n})^{\otimes k}\right)$ such that for every density operator $\rho$ on $\calH=\mathbb{C}^{2^n}$,
    \begin{align}
        \norm{\calM_{m,k}(\rho^{\otimes m})
        -\rho_{\mathrm{sym}}^{(k)}
        }_1\leq \epsilon
    \end{align}
    holds
    as long as 
    \begin{align}
        m\geq\frac{5k^2(k-1)}{\epsilon^2}.
    \end{align}
\end{theorem}\noindent
See Theorem~\ref{thm:efficient-projection} in Appendix~\ref{app:projection} for a formal statement and proof.

\section{Observable expectation value estimation}\label{sec:expectation-value}

Combining all the results above, we show that the expectation value of an observable on the Scrooge $k$-th moment,
$\Tr\!\left[\chi_{\rm Sc}^{(k)} O_k\right]$,
can be estimated up to a prescribed error using the polynomial approximation. Our method estimates $\Tr\!\left[\chi_{\rm PA}^{(k)} O_k\right]$  as a good approximation for $\Tr\!\left[\chi_{\rm Sc}^{(k)} O_k\right]$. Our algorithm works by first estimating the expectation value of each term and then summing them up.

For each $1\leq i\leq\lambda$, denote
\begin{equation}
o_i
:=
\Tr\!\left[
\rho_{\mathrm{sym}}^{(k_i)}
\left(
O_k\otimes I^{\otimes a_i}
\right)
\right].
\end{equation}
Instead of estimating $p_{\mathrm{sym}}^{(k_i)}$ and $o_i$
separately, we directly estimate their product using samples obtained
from all Schur sectors. Set $m_i:=k_i^2(k_i-1)$.

The estimator is constructed by the following procedure.
\begin{enumerate}
    \item Prepare $\rho^{\otimes m_i}$ and apply the Schur transform
    $U_{\mathrm{Sch}}^{(m_i)}$. Measuring the irrep register produces a
    partition $\gamma\vdash m_i$ with probability $\mu_\gamma(\rho)$.

    \item Conditioned on $\gamma$,
    replace the permutation register by $P_\gamma^{(k_i)}/\trace{}{P_\gamma^{(k_i)}}$
    whenever $\trace{}{P_\gamma^{(k_i)}}>0$. If
    $\trace{}{P_\gamma^{(k_i)}}=0$, prepare an arbitrary fixed state in
    the permutation register and set $a_\gamma^{(k_i)}=0$.

    \item Apply the inverse Schur transform and measure
    $\left(O_k\otimes I^{\otimes a_i}\right)
        \otimes I^{\otimes(m_i-k_i)}
        =O_k\otimes I^{\otimes(m_i-k)}$. In repetition $(r,\ell)$, denote the
    measured Schur label and measurement outcome by
    $\gamma_{i,r,\ell}$ and $W_{i,r,\ell}$, respectively, and define
    the corresponding estimator sample by
    \begin{equation}
        X_{i,r,\ell}
        :=
        a_{\gamma_{i,r,\ell}}^{(k_i)}W_{i,r,\ell}.
    \end{equation}

    \item Repeat the above procedure independently and use the
    median-of-means estimator. Specifically, divide the samples into
    $R_i$ groups, each containing $L_i$ samples. For the $r$-th group,
    define the sample mean
    \begin{equation}
        Y_{i,r}
        :=
        \frac{1}{L_i}\sum_{\ell=1}^{L_i}X_{i,r,\ell},
    \end{equation}
    and define
    \begin{equation}
        \widehat{p_{\mathrm{sym}}^{(k_i)}o_i}
        :=
        \operatorname{median}\left\{
        Y_{i,1},\ldots,Y_{i,R_i}
        \right\}.
    \end{equation}
\end{enumerate}

Algorithm~\ref{alg:symmetric-sector-estimation} summarizes the
estimation procedure.

\begin{algorithm}[t!]
\caption{Estimation of the $i$-th symmetric-sector contribution}
\label{alg:symmetric-sector-estimation}
\begin{algorithmic}
\For{$r=1,\ldots,R_i$}
    \For{$\ell=1,\ldots,L_i$}
        \State Prepare $\rho^{\otimes m_i}$ and apply
        $U_{\mathrm{Sch}}^{(m_i)}$.
        \State Measure the irrep register and denote the outcome by
        $\gamma_{i,r,\ell}$.
        \State Compute $a_{\gamma_{i,r,\ell}}^{(k_i)}$.
        \If{$\Tr P_{\gamma_{i,r,\ell}}^{(k_i)}>0$}
            \State Replace the permutation register by
            $P_{\gamma_{i,r,\ell}}^{(k_i)}
            /\Tr P_{\gamma_{i,r,\ell}}^{(k_i)}$.
            \State Apply $U_{\mathrm{Sch}}^{(m_i)\dagger}$.
            \State Measure
            $O_k\otimes I^{\otimes(m_i-k)}$
            and denote the outcome by $W_{i,r,\ell}$.
            \State Set
            $X_{i,r,\ell}
            =a_{\gamma_{i,r,\ell}}^{(k_i)}W_{i,r,\ell}$.
        \Else
            \State Set $X_{i,r,\ell}=0$.
        \EndIf
    \EndFor
    \State Set
    $Y_{i,r}
    =L_i^{-1}\sum_{\ell=1}^{L_i}X_{i,r,\ell}$.
\EndFor
\State \Return
$\widehat{p_{\mathrm{sym}}^{(k_i)}o_i}
=\operatorname{median}
\left(Y_{i,1},\ldots,Y_{i,R_i}\right)$.
\end{algorithmic}
\end{algorithm}

Each sample satisfies
\begin{equation}
\mathbb{E}[X_{i,r,\ell}]
=
p_{\mathrm{sym}}^{(k_i)}o_i.
\end{equation}
Hence, the algorithm directly estimates
$p_{\mathrm{sym}}^{(k_i)}o_i$ without separately estimating its two
factors. Combining the estimators for $1\leq i\leq\lambda$ yields an
estimator for the target observable expectation value.

\begin{theorem}
[Quantum algorithm for observable expectation evaluation for Scrooge moments; informal]
\label{Sample and circuit complexities of operator expectation value estimation algorithm-maintext}
Let $\rho$ be an $n$-qubit density matrix, $k$ be an integer satisfying $k\geq 2$ and  $12(k-1)^2\leq \M$. 
There exists a quantum algorithm that, for any efficiently measurable Hermitian $k$-copy operator $O_k$ with $\norm{O_k}_\infty\leq 1$ and any $0<\epsilon_\mathrm{t}<1$ satisfying 
$ 
\epsilon_{\mathrm t} 
\geq 
2C_{\add} 
\exp\left[ 
-\M/\left(24(k-1)^2\right) 
\right]
$ 
, estimates $\trace{}{\chi_\mathrm{Sc}^{(k)}(\rho)O_k}$ up to additive error $\epsilon_\mathrm{t}$.
The algorithm requires at most
$\widetilde{\calO}(k^3/\epsilon_{\rm t}^5)$
total number of copies of $\rho$, and a sufficient total circuit complexity is $\operatorname{poly}\!\left(
n,k,1/\epsilon_{\mathrm t}
\right).$
\end{theorem}\noindent
See Theorem~\ref{Sample and circuit complexities of operator expectation value estimation algorithm} in Appendix~\ref{app:expectation-value} for a formal statement and proof.

\section{Block encoding}\label{sec:be}

In this section, we discuss how to extend our method to construct a block encoding of the Scrooge $k$-th moment.
Block encoding provides a way to coherently access a generally nonunitary operator by embedding it into a larger unitary~\cite{low2019hamiltonian,chakraborty2019power,gilyen2019quantum}. Once such a block encoding is available, the encoded operator can be further processed using quantum algorithmic techniques such as quantum singular value transformation (QSVT)~
\cite{low2017optimal,gilyen2019quantum,martyn2021grand}. 

To construct a block encoding of the Scrooge $k$-th moment, we first construct a block encoding of its polynomial approximation. By choosing the approximation order appropriately, this yields a block encoding that approximates the Scrooge $k$-th moment to the prescribed accuracy. A detailed and rigorous derivation of the construction and complexity analysis is provided in Appendix~\ref{ap:Construction of block encoding}.

An $(\alpha,s,\epsilon)$ block encoding of an $n$-qubit operator $A$
is an $(n+s)$-qubit unitary $U$ satisfying
\begin{equation}
    \bigl\|A-\alpha(\bra{0}^{\otimes s}\otimes I)
U(\ket{0}^{\otimes s}\otimes I)\bigr\|_\infty\leq\epsilon.
\end{equation}

Suppose now that \(A\) is a density operator and that we have coherent access to a unitary \(U_A\) preparing its purification $\ket{\Psi_A}_{\sf R,\sf S}$, where $\sf R$ is an $s$-qubit
ancilla register, $\sf S$ is the $n$-qubit system register, and
$\Tr_{\sf R}[\ket{\Psi_A}\bra{\Psi_A}]=A$. Then
\begin{equation}
\left(U_A^\dagger\otimes I_{\sf S'}\right)
\left(I_{\sf R}\otimes\operatorname{SWAP}_{\sf S,\sf S'}\right)
\left(U_A\otimes I_{\sf S'}\right)
\end{equation}
is an exact $(1,s+n,0)$ block-encoding unitary of $A$ acting on the
additional $n$-qubit register $\sf S'$.

The polynomial approximation can be written as
\begin{align}
\chi_{\mathrm{PA}}^{(k)}(\lambda)
=
\sum_{i=1}^{\lambda}
\binom{\lambda}{i}(-1)^{i-1}
c_i p_{\mathrm{sym}}^{(k_i)}\sigma_i,
\label{eq:main-block-direct-decomposition}
\end{align}
where
$\sigma_i:=\Tr_{\mathrm{aux}}\!\left[\rho_{\mathrm{sym}}^{(k_i)}\right]$ and $c_i=D^{k_i}/\binom{D+k_i-1}{k_i}$.
For each $i$, we first estimate
$p_{\mathrm{sym}}^{(k_i)}$ using efficient implementation of symmetric subspace
projection (see Appendix~\ref{Corollary: precise estimation of p} for details) and denote the resulting estimate by
$\widehat p_{\mathrm{sym}}^{(k_i)}$. Next, we use a purified implementation of the channel
in Theorem~\ref{Efficient implementation of symmetric subspace projection; informal} to prepare a state approximating
$\rho_{\mathrm{sym}}^{(k_i)}$. By treating the last $a_i$ output registers as ancillas, the same circuit prepares a purification of a
state approximating $\sigma_i$ on the first $k$ registers. The
purification-based construction introduced above then yields an
approximate block encoding of $\sigma_i$. 

We combine the individual block encodings using the standard
linear-combination-of-unitaries (LCU) technique~\cite{gilyen2019quantum, childs2012hamiltonian,childs2017quantum}. In our construction,
the LCU coefficients and the resulting normalization factor are
\begin{equation}
\widehat y_i
:=
\binom{\lambda}{i}(-1)^{i-1}
c_i\widehat p_{\mathrm{sym}}^{(k_i)},
\end{equation}
\begin{equation}
\alpha
:=
\sum_{i=1}^{\lambda}|\widehat y_i|
=
\sum_{i=1}^{\lambda}
\binom{\lambda}{i}
c_i\widehat p_{\mathrm{sym}}^{(k_i)}.
\end{equation}
The LCU construction requires an additional
$\lceil\log_2\lambda\rceil$ ancilla qubits for the index register.

Taking these considerations together, and accounting for the errors 
introduced by the polynomial approximation and the approximate 
Schur transform, the normalization factor $\alpha$, the number of ancilla 
qubits, and the circuit complexity required to achieve a target error 
are bounded by the following theorem.

\begin{theorem}
[Complexity of the block encoding of the Scrooge $k$-th moment; informal]
\label{thm-main:block-encoding-PA}
Let $\rho$ be an $n$-qubit density matrix, $k$ be an integer satisfying $k\geq 2$ and  $12(k-1)^2\leq \M$. If we have query access to an efficient purification
circuit $U_\rho$ for $\rho$ and $U_\rho^\dagger$, then
there exists a quantum algorithm that, for any   $0<\epsilon_\mathrm{t}<1$ satisfying
$
\epsilon_{\mathrm t}
\geq
2C_{\add}
\exp\left[
-\M/\left(24(k-1)^2\right)
\right]
$
, constructs an
$(\alpha,a_{\mathrm{BE}},\epsilon_{\mathrm t})$ block encoding of
$\chi_{\mathrm{Sc}}^{(k)}(\rho)$, where
$
\alpha
=
\mathcal O\left[
\left(
2C_{\add}/\epsilon_{\mathrm t}
\right)^2
\right]$,
$
a_{\mathrm{BE}}
=
\operatorname{poly}\left(
n,k,1/\epsilon_{\mathrm t}
\right)$. The total number of queries to
$U_\rho$ and $U_\rho^\dagger$ is
$
\operatorname{poly}\left(
k,1/\epsilon_{\mathrm t}\right)$,
and the total circuit complexity is
$
\operatorname{poly}\left(
n,k,1/\epsilon_{\mathrm t}
\right)
$.
\end{theorem}
See Theorem~\ref{thm:block-encoding-PA} in Appendix~\ref{ap:Construction of block encoding} for details.

\section{Discussion}\label{sec:discussion}

In this work, we develop methods for efficiently predicting the properties of Scrooge ensembles from copies of an unknown state $\rho$. We first introduce an analytical approximation to the higher-order moments of these ensembles. Combining this approximation with techniques from Schur--Weyl duality, we construct quantum algorithms for estimating observable expectation values with respect to Scrooge $k$-th moments. When a preparation circuit for $\rho$ is available, we also construct quantum circuits that implement block encodings of these moments.

Our approximation is motivated by the concentration properties of Haar-random states and takes the full density matrix $\rho$ as input. It offers a tunable error whose minimum achievable value is exponentially small in $1/\|\rho\|_\infty$, yielding an exponential improvement over previous results~\cite{mcginley2025scroogeensemblemanybodyquantum}. For fixed $k$, the approximation achieves negligible error when $\|\rho\|_\infty$ decreases sufficiently rapidly with system size. We expect this approach to facilitate the theoretical analysis of Scrooge ensembles at high precision. 
We expect potential applications in, e.g.,
 refined characterizations of the information-theoretic quantities studied in Refs.~\cite{mark2024maximum,mcginley2025scroogeensemblemanybodyquantum,mok2026nature} and more precise benchmarking protocols for quantum simulators~\cite{mark2023benchmarking,Choi2023PreparingRandom}.

Building on this approximation and Schur--Weyl duality, our quantum algorithm estimates observable expectation values with respect to Scrooge $k$-th moments to accuracy $\epsilon$ using $\widetilde{\mathcal{O}}(k^3/\epsilon^5)$ copies of $\rho$. This provides a way to predict properties of Scrooge ensembles without directly sampling from them. For example, it may enable the study of deeply thermalized ensembles using only copies of their associated density matrices, which can usually be efficiently prepared~\cite{chang2025deep, mark2024maximum, mok2026nature, mark2026deep, yan2026characterizing}.

Access to a preparation circuit for $\rho$ further allows us to construct block encodings of the Scrooge $k$-th moments. Previous work has shown that reweighting classical-shadow estimators using Scrooge second moments can improve the sample complexity of shadow tomography~\cite{tran2023measuring,innocenti2023shadow,mcginley2023shadow}. Our block-encoding construction may therefore provide a route to more efficient shadow-tomography protocols, although concrete implementations and their complexity guarantees remain to be established.

A key subroutine in our algorithms efficiently implements the projection of $\rho^{\otimes k}$ onto the $k$-fold symmetric subspace. To the best of our knowledge, this method has not been reported previously. We expect it to serve as a useful building block for other quantum algorithms, and similar techniques may be adapted to implement projections onto other irreducible representation spaces.

Several questions remain open. From a practical perspective, it would be valuable to identify further applications of Scrooge ensembles in quantum information processing. As mentioned, existing work demonstrates that when the knowledge of input states is partially known, the properties of Scrooge ensembles can be used to improve certain tasks like shadow tomography~\cite{tran2023measuring,innocenti2023shadow,mcginley2023shadow}. Extending this idea beyond shadow tomography may potentially improve other tasks, and our approximation methods and algorithms may help make such applications accessible.

From a complexity-theoretic perspective, an interesting question is to extend our guarantees to high-purity states. Although the minimum achievable approximation error is exponentially small in $1/\|\rho\|_\infty$, this guarantee becomes weak when $\|\rho\|_\infty$ is large. Whether this limitation can be overcome to obtain efficient algorithms for high-purity states, or even arbitrary density matrices, remains unclear. Finally, several bounds used in our complexity analysis are loose, suggesting that the sample complexity scaling of our algorithms may be improved. Determining the optimal sample complexity for this problem and developing algorithms that attain it are natural directions for future work.

\section{Acknowledgements}

We thank Thomas Schuster, Hui Zhai, Wen Wei Ho, Chang Liu, Yongtao Zhan, and Wai-Keong Mok for insightful discussions. We are especially grateful to Daniel Mark for pointing out the applications of Scrooge ensembles in shadow tomography.

\textbf{AI disclosure} --- GPT-5.6 Sol and -6 Astra (OpenAI) help generate candidate proofs for a subset of results in Appendix~\ref{app:projection} and intermediate proof steps for several lemmas and theorems in Appendix~\ref{app:expectation-value}, \ref{ap:Construction of block encoding} from targeted statements. All proofs involving AI use are verified, refined,
and reformulated by the authors. The same models are used to generate the codes for Fig.~\ref{fig:additive-error}. The authors also acknowledge the use of AI tools to improve the writing of the manuscript.
The authors take full responsibility for all content.

\bibliography{info_references}

@article{ho2022exact,
  author = {Ho, Wen Wei and Choi, Soonwon},
  title = {Exact emergent quantum state designs from quantum chaotic dynamics},
  journal = {Phys. Rev. Lett.},
  volume = {128},
  pages = {060601},
  year = {2022},
  doi = {10.1103/physrevlett.128.060601},
}

@article{shaw2025experimental,
  author = {Shaw, Adam L and Mark, Daniel K and Choi, Joonhee and Finkelstein, Ran and Scholl, Pascal and Choi, Soonwon and Endres, Manuel},
  title = {Experimental signatures of {Hilbert}-space ergodicity: Universal bitstring distributions and applications in noise learning},
  journal = {Phys. Rev. X},
  volume = {15},
  pages = {031001},
  year = {2025},
  doi = {10.1103/h6xy-zpx4},
}

@article{emerson2005scalable,
  author = {Emerson, Joseph and Alicki, Robert and {\.Z}yczkowski, Karol},
  title = {Scalable noise estimation with random unitary operators},
  journal = {J. Opt. B: Quantum Semiclassical Opt.},
  volume = {7},
  pages = {S347--S352},
  year = {2005},
  doi = {10.1088/1464-4266/7/10/021},
}

@article{knill2008randomized,
  author = {Knill, Emanuel and Leibfried, Dietrich and Reichle, Rolf and Britton, Joe and Blakestad, R Brad and Jost, John D and Langer, Chris and Ozeri, Roee and Seidelin, Signe and Wineland, David J},
  title = {Randomized benchmarking of quantum gates},
  journal = {Phys. Rev. A},
  volume = {77},
  pages = {012307},
  year = {2008},
  doi = {10.1103/physreva.77.012307},
}

@inproceedings{gilyen2019quantum,
  author = {Gily{\'e}n, Andr{\'a}s and Su, Yuan and Low, Guang Hao and Wiebe, Nathan},
  title = {Quantum singular value transformation and beyond: exponential improvements for quantum matrix arithmetics},
  booktitle = {Proceedings of the 51st Annual ACM SIGACT Symposium on Theory of Computing},
  pages = {193--204},
  publisher = {ACM},
  year = {2019},
  doi = {10.1145/3313276.3316366},
}

@article{innocenti2023shadow,
  author = {Innocenti, Luca and Lorenzo, Salvatore and Palmisano, Ivan and Albarelli, Francesco and Ferraro, Alessandro and Paternostro, Mauro and Palma, G Massimo},
  title = {Shadow tomography on general measurement frames},
  journal = {PRX Quantum},
  volume = {4},
  pages = {040328},
  year = {2023},
  doi = {10.1103/prxquantum.4.040328},
}

@article{mark2024maximum,
  author = {Mark, Daniel K and Surace, Federica and Elben, Andreas and Shaw, Adam L and Choi, Joonhee and Refael, Gil and Endres, Manuel and Choi, Soonwon},
  title = {Maximum entropy principle in deep thermalization and in {Hilbert}-space ergodicity},
  journal = {Phys. Rev. X},
  volume = {14},
  pages = {041051},
  year = {2024},
  doi = {10.1103/physrevx.14.041051},
}

@article{mark2023benchmarking,
  author = {Mark, Daniel K and Choi, Joonhee and Shaw, Adam L and Endres, Manuel and Choi, Soonwon},
  title = {Benchmarking quantum simulators using ergodic quantum dynamics},
  journal = {Phys. Rev. Lett.},
  volume = {131},
  pages = {110601},
  year = {2023},
  doi = {10.1103/physrevlett.131.110601},
}

@article{JRW94,
  author = {Jozsa, Richard and Robb, Daniel and Wootters, William K.},
  title = {Lower bound for accessible information in quantum mechanics},
  journal = {Phys. Rev. A},
  volume = {49},
  pages = {668--677},
  year = {1994},
  doi = {10.1103/physreva.49.668},
}

@misc{harrow2005applications,
  author = {Harrow, Aram W},
  title = {Applications of coherent classical communication and the {Schur} transform to quantum information theory},
  year = {2005},
  eprint = {quant-ph/0512255},
  archivePrefix = {arXiv},
  primaryClass = {quant-ph},
}

@article{sykora1974quantum,
  author = {S{\'y}kora, Stanislav},
  title = {Quantum theory and the {Bayesian} inference problems},
  journal = {J. Stat. Phys.},
  volume = {11},
  pages = {17--27},
  year = {1974},
  doi = {10.1007/bf01019475},
}

@article{cotler_emergent_2023,
  author = {Cotler, Jordan S. and Mark, Daniel K. and Huang, Hsin-Yuan and Hern{\'a}ndez, Felipe and Choi, Joonhee and Shaw, Adam L. and Endres, Manuel and Choi, Soonwon},
  title = {Emergent quantum state designs from individual many-body wave functions},
  journal = {PRX Quantum},
  volume = {4},
  pages = {010311},
  year = {2023},
  doi = {10.1103/prxquantum.4.010311},
}

@article{deutsch1991quantum,
  author = {Deutsch, Josh M},
  title = {Quantum statistical mechanics in a closed system},
  journal = {Phys. Rev. A},
  volume = {43},
  pages = {2046--2049},
  year = {1991},
  doi = {10.1103/physreva.43.2046},
}

@article{fisher2023random,
  author = {Fisher, Matthew P. A. and Khemani, Vedika and Nahum, Adam and Vijay, Sagar},
  title = {Random quantum circuits},
  journal = {Annu. Rev. Condens. Matter Phys.},
  volume = {14},
  pages = {335--379},
  year = {2023},
  doi = {10.1146/annurev-conmatphys-031720-030658},
}

@article{srednicki1994chaos,
  author = {Srednicki, Mark},
  title = {Chaos and quantum thermalization},
  journal = {Phys. Rev. E},
  volume = {50},
  pages = {888--901},
  year = {1994},
  doi = {10.1103/physreve.50.888},
}

@article{rigol2008thermalization,
  author = {Rigol, Marcos and Dunjko, Vanja and Olshanii, Maxim},
  title = {Thermalization and its mechanism for generic isolated quantum systems},
  journal = {Nature},
  volume = {452},
  pages = {854--858},
  year = {2008},
  doi = {10.1038/nature06838},
}

@article{nahum2017quantum,
  author = {Nahum, Adam and Ruhman, Jonathan and Vijay, Sagar and Haah, Jeongwan},
  title = {Quantum entanglement growth under random unitary dynamics},
  journal = {Phys. Rev. X},
  volume = {7},
  pages = {031016},
  year = {2017},
  doi = {10.1103/physrevx.7.031016},
}

@article{roberts2017chaos,
  author = {Roberts, Daniel A and Yoshida, Beni},
  title = {Chaos and complexity by design},
  journal = {J. High Energy Phys.},
  volume = {2017},
  number = {4},
  pages = {121},
  year = {2017},
  doi = {10.1007/jhep04(2017)121},
}

@article{cotler2017chaos,
  author = {Cotler, Jordan and Hunter-Jones, Nicholas and Liu, Junyu and Yoshida, Beni},
  title = {Chaos, complexity, and random matrices},
  journal = {J. High Energy Phys.},
  volume = {2017},
  number = {11},
  pages = {048},
  year = {2017},
  doi = {10.1007/jhep11(2017)048},
}

@article{nakata2017efficient,
  author = {Nakata, Yoshifumi and Hirche, Christoph and Koashi, Masato and Winter, Andreas},
  title = {Efficient quantum pseudorandomness with nearly time-independent {Hamiltonian} dynamics},
  journal = {Phys. Rev. X},
  volume = {7},
  pages = {021006},
  year = {2017},
  doi = {10.1103/physrevx.7.021006},
}

@article{hayden2007black,
  author = {Hayden, Patrick and Preskill, John},
  title = {Black holes as mirrors: quantum information in random subsystems},
  journal = {J. High Energy Phys.},
  volume = {2007},
  number = {9},
  pages = {120},
  year = {2007},
  doi = {10.1088/1126-6708/2007/09/120},
}

@article{DAlessio2016QuantumChaosEigenstateThermalization,
  author = {D'Alessio, L. and Kafri, Y. and Polkovnikov, A. and Rigol, M.},
  title = {From quantum chaos and eigenstate thermalization to statistical mechanics and thermodynamics},
  journal = {Adv. Phys.},
  volume = {65},
  pages = {239--362},
  year = {2016},
  doi = {10.1080/00018732.2016.1198134},
}

@article{page1993average,
  author = {Page, Don N},
  title = {Average entropy of a subsystem},
  journal = {Phys. Rev. Lett.},
  volume = {71},
  pages = {1291--1294},
  year = {1993},
  doi = {10.1103/physrevlett.71.1291},
}

@inproceedings{page1994black,
  author = {Page, Don N},
  title = {Black hole information},
  booktitle = {Proceedings of the 5th Canadian Conference on General Relativity and Relativistic Astrophysics},
  editor = {Mann, R. B. and McLenaghan, R. G.},
  pages = {1--41},
  publisher = {World Scientific},
  address = {Singapore},
  year = {1994},
  eprint = {hep-th/9305040},
  archivePrefix = {arXiv},
  primaryClass = {hep-th},
}

@article{Choi2023PreparingRandom,
  author = {Choi, J. and Shaw, A. L. and Madjarov, I. S. and Xie, X. and Finkelstein, R. and Covey, J. P. and Cotler, J. S. and Mark, D. K. and Huang, H.-Y. and Kale, A. and Pichler, H. and Brand{\~{a}}o, F. G. S. L. and Choi, S. and Endres, M.},
  title = {Preparing random states and benchmarking with many-body quantum chaos},
  journal = {Nature},
  volume = {613},
  pages = {468--473},
  year = {2023},
  doi = {10.1038/s41586-022-05442-1},
}

@article{Arute2019QuantumSupremacy,
  author = {Arute, F. and Arya, K. and Babbush, R. and Bacon, D. and Bardin, J. C. and Barends, R. and Biswas, R. and Boixo, S. and Brand{\~{a}}o, F. G. S. L. and Buell, D. A. and Burkett, B. and Chen, Y. and Chen, Z. and Chiaro, B. and Collins, R. and Courtney, W. and Dunsworth, A. and Farhi, E. and Foxen, B. and Fowler, A. and Gidney, C. and Giustina, M. and Graff, R. and Guerin, K. and Habegger, S. and Harrigan, M. P. and Hartmann, M. J. and Ho, A. and Hoffmann, M. and Huang, T. and Humble, T. S. and Isakov, S. V. and Jeffrey, E. and Jiang, Z. and Kafri, D. and Kechedzhi, K. and Kelly, J. and Klimov, P. V. and Knysh, S. and Korotkov, A. N. and Kostritsa, F. and Landhuis, D. and Lindmark, M. and Lucero, E. and Lyakh, D. and Mandr{\`{a}}, S. and McClean, J. R. and McEwen, M. and Megrant, A. and Mi, X. and Michielsen, K. and Mohseni, M. and Mutus, J. and Naaman, O. and Neeley, M. and Neill, C. and Niu, M. Y. and Ostby, E. and Petukhov, A. and Platt, J. C. and Quintana, C. and Rieffel, E. G. and Roushan, P. and Rubin, N. C. and Sank, D. and Satzinger, K. J. and Smelyanskiy, V. and Sung, K. J. and Trevithick, M. D. and Vainsencher, A. and Villalonga, B. and White, T. and Yao, Z. J. and Yeh, P. and Zalcman, A. and Neven, H. and Martinis, J. M.},
  title = {Quantum supremacy using a programmable superconducting processor},
  journal = {Nature},
  volume = {574},
  pages = {505--510},
  year = {2019},
  doi = {10.1038/s41586-019-1666-5},
}

@article{Huang2020PredictingManyPropertiesQuantum,
  author = {Huang, Hsin-Yuan and Kueng, Richard and Preskill, John},
  title = {Predicting many properties of a quantum system from very few measurements},
  journal = {Nat. Phys.},
  volume = {16},
  pages = {1050--1057},
  year = {2020},
  doi = {10.1038/s41567-020-0932-7},
}

@article{elben2023randomized,
  author = {Elben, Andreas and Flammia, Steven T and Huang, Hsin-Yuan and Kueng, Richard and Preskill, John and Vermersch, Beno{\^\i}t and Zoller, Peter},
  title = {The randomized measurement toolbox},
  journal = {Nat. Rev. Phys.},
  volume = {5},
  pages = {9--24},
  year = {2023},
  doi = {10.1038/s42254-022-00535-2},
}

@article{huang2025certifying,
  author = {Huang, Hsin-Yuan and Preskill, John and Soleimanifar, Mehdi},
  title = {Certifying almost all quantum states with few single-qubit measurements},
  journal = {Nat. Phys.},
  volume = {21},
  pages = {1834--1841},
  year = {2025},
  doi = {10.1038/s41567-025-03025-1},
}

@article{tran2023measuring,
  author = {Tran, Minh C and Mark, Daniel K and Ho, Wen Wei and Choi, Soonwon},
  title = {Measuring arbitrary physical properties in analog quantum simulation},
  journal = {Phys. Rev. X},
  volume = {13},
  pages = {011049},
  year = {2023},
  doi = {10.1103/physrevx.13.011049},
}

@misc{Venkatesh2026ErrorCorrection,
  author = {Venkatesh, A. and Allen, R. R. and Pilatowsky-Cameo, S. and Ye, B. and Choi, S.},
  title = {Quantum thermalization achieves optimal approximate quantum error correction},
  year = {2026},
  eprint = {2609.04121},
  archivePrefix = {arXiv},
  primaryClass = {quant-ph},
}

@article{dankert2009exact,
  author = {Dankert, Christoph and Cleve, Richard and Emerson, Joseph and Livine, Etera},
  title = {Exact and approximate unitary 2-designs and their application to fidelity estimation},
  journal = {Phys. Rev. A},
  volume = {80},
  pages = {012304},
  year = {2009},
  doi = {10.1103/physreva.80.012304},
}

@inproceedings{ji2018pseudorandom,
  author = {Ji, Zhengfeng and Liu, Yi-Kai and Song, Fang},
  title = {Pseudorandom quantum states},
  booktitle = {Advances in Cryptology -- CRYPTO 2018},
  pages = {126--152},
  publisher = {Springer International Publishing},
  year = {2018},
  doi = {10.1007/978-3-319-96878-0_5},
}

@article{goldstein2006distribution,
  author = {Goldstein, Sheldon and Lebowitz, Joel L and Tumulka, Roderich and Zangh{\`\i}, Nino},
  title = {On the distribution of the wave function for systems in thermal equilibrium},
  journal = {J. Stat. Phys.},
  volume = {125},
  pages = {1193--1221},
  year = {2006},
  doi = {10.1007/s10955-006-9210-z},
}

@article{goldstein2016universal,
  author = {Goldstein, Sheldon and Lebowitz, Joel L and Mastrodonato, Christian and Tumulka, Roderich and Zangh{\`\i}, Nino},
  title = {Universal probability distribution for the wave function of a quantum system entangled with its environment},
  journal = {Commun. Math. Phys.},
  volume = {342},
  pages = {965--988},
  year = {2016},
  doi = {10.1007/s00220-015-2536-0},
}

@article{mcginley2023shadow,
  author = {McGinley, Max and Fava, Michele},
  title = {Shadow tomography from emergent state designs in analog quantum simulators},
  journal = {Phys. Rev. Lett.},
  volume = {131},
  pages = {160601},
  year = {2023},
  doi = {10.1103/physrevlett.131.160601},
}

@article{du2026no,
  author = {Du, Zhenyu and Cheng, Siyuan and Zhao, Qi and Ma, Xiongfeng and Yuan, Xiao},
  title = {No cloning of quantum ensembles},
  journal = {Nat. Commun.},
  volume = {17},
  pages = {8052},
  year = {2026},
  doi = {10.1038/s41467-026-74924-x},
}

@article{bacon2006efficient,
  author = {Bacon, Dave and Chuang, Isaac L and Harrow, Aram W},
  title = {Efficient quantum circuits for {Schur} and {Clebsch--Gordan} transforms},
  journal = {Phys. Rev. Lett.},
  volume = {97},
  pages = {170502},
  year = {2006},
  doi = {10.1103/physrevlett.97.170502},
}

@misc{burchardt2025high,
  author = {Burchardt, Adam and Fei, Jiani and Grinko, Dmitry and Larocca, Martin and Ozols, Maris and Timmerman, Sydney and Visnevskyi, Vladyslav},
  title = {High-dimensional quantum Schur transforms},
  year = {2025},
  eprint = {2509.22640},
  archivePrefix = {arXiv},
  primaryClass = {quant-ph},
}

@misc{harrow2013church,
  author = {Harrow, Aram W},
  title = {The church of the symmetric subspace},
  year = {2013},
  eprint = {1308.6595},
  archivePrefix = {arXiv},
  primaryClass = {quant-ph},
}

@article{keyl2001estimating,
  author = {Keyl, Michael and Werner, Reinhard F},
  title = {Estimating the spectrum of a density operator},
  journal = {Phys. Rev. A},
  volume = {64},
  pages = {052311},
  year = {2001},
  doi = {10.1103/physreva.64.052311},
}

@misc{childs2012hamiltonian,
  author = {Childs, Andrew M and Wiebe, Nathan},
  title = {{Hamiltonian} simulation using linear combinations of unitary operations},
  year = {2012},
  eprint = {1202.5822},
  archivePrefix = {arXiv},
  primaryClass = {quant-ph},
}

@article{yan2026characterizing,
  author = {Yan, Zhiguang and Ge, Zi-Yong and Li, Rui and Zhang, Yu-Ran and Nori, Franco and Nakamura, Yasunobu},
  title = {Characterizing many-body dynamics with projected ensembles on a superconducting quantum processor},
  journal = {Sci. Adv.},
  volume = {12},
  pages = {eaeb8213},
  year = {2026},
  doi = {10.1126/sciadv.aeb8213},
}

@article{liu2026observation,
  author = {Liu, Wenquan and Pan, Zou-Wei and Fu, Yue and Ho, Wen Wei and Rong, Xing},
  title = {Observation of hierarchy of {Hilbert}-space ergodicities in the quantum dynamics of a single spin},
  journal = {Phys. Rev. Lett.},
  volume = {136},
  pages = {020401},
  year = {2026},
  doi = {10.1103/6msb-cxbc},
}

@article{morvan2024phase,
  author = {Morvan, A. and Villalonga, B. and Mi, X. and Mandr{\`a}, S. and Bengtsson, A. and Klimov, P. V. and Chen, Z. and Hong, S. and Erickson, C. and Drozdov, I. K. and Chau, J. and Laun, G. and Movassagh, R. and Asfaw, A. and Brand{\~{a}}o, L. T. A. N. and Peralta, R. and Abanin, D. and Acharya, R. and Allen, R. and Andersen, T. I. and Anderson, K. and Ansmann, M. and Arute, F. and Arya, K. and Atalaya, J. and Bardin, J. C. and Bilmes, A. and Bortoli, G. and Bourassa, A. and Bovaird, J. and Brill, L. and Broughton, M. and Buckley, B. B. and Buell, D. A. and Burger, T. and Burkett, B. and Bushnell, N. and Campero, J. and Chang, H.-S. and Chiaro, B. and Chik, D. and Chou, C. and Cogan, J. and Collins, R. and Conner, P. and Courtney, W. and Crook, A. L. and Curtin, B. and Debroy, D. M. and Barba, A. Del Toro and Demura, S. and Paolo, A. Di and Dunsworth, A. and Faoro, L. and Farhi, E. and Fatemi, R. and Ferreira, V. S. and Burgos, L. Flores and Forati, E. and Fowler, A. G. and Foxen, B. and Garcia, G. and Genois, {\'E}. and Giang, W. and Gidney, C. and Gilboa, D. and Giustina, M. and Gosula, R. and Dau, A. Grajales and Gross, J. A. and Habegger, S. and Hamilton, M. C. and Hansen, M. and Harrigan, M. P. and Harrington, S. D. and Heu, P. and Hoffmann, M. R. and Huang, T. and Huff, A. and Huggins, W. J. and Ioffe, L. B. and Isakov, S. V. and Iveland, J. and Jeffrey, E. and Jiang, Z. and Jones, C. and Juhas, P. and Kafri, D. and Khattar, T. and Khezri, M. and Kieferov{\'a}, M. and Kim, S. and Kitaev, A. and Klots, A. R. and Korotkov, A. N. and Kostritsa, F. and Kreikebaum, J. M. and Landhuis, D. and Laptev, P. and Lau, K.-M. and Laws, L. and Lee, J. and Lee, K. W. and Lensky, Y. D. and Lester, B. J. and Lill, A. T. and Liu, W. and Livingston, W. P. and Locharla, A. and Malone, F. D. and Martin, O. and Martin, S. and McClean, J. R. and McEwen, M. and Miao, K. C. and Mieszala, A. and Montazeri, S. and Mruczkiewicz, W. and Naaman, O. and Neeley, M. and Neill, C. and Nersisyan, A. and Newman, M. and Ng, J. H. and Nguyen, A. and Nguyen, M. and Niu, M. Yuezhen and O'Brien, T. E. and Omonije, S. and Opremcak, A. and Petukhov, A. and Potter, R. and Pryadko, L. P. and Quintana, C. and Rhodes, D. M. and Rocque, C. and Rosenberg, E. and Rubin, N. C. and Saei, N. and Sank, D. and Sankaragomathi, K. and Satzinger, K. J. and Schurkus, H. F. and Schuster, C. and Shearn, M. J. and Shorter, A. and Shutty, N. and Shvarts, V. and Sivak, V. and Skruzny, J. and Smith, W. C. and Somma, R. D. and Sterling, G. and Strain, D. and Szalay, M. and Thor, D. and Torres, A. and Vidal, G. and Heidweiller, C. Vollgraff and White, T. and Woo, B. W. K. and Xing, C. and Yao, Z. J. and Yeh, P. and Yoo, J. and Young, G. and Zalcman, A. and Zhang, Y. and Zhu, N. and Zobrist, N. and Rieffel, E. G. and Biswas, R. and Babbush, R. and Bacon, D. and Hilton, J. and Lucero, E. and Neven, H. and Megrant, A. and Kelly, J. and Roushan, P. and Aleiner, I. and Smelyanskiy, V. and Kechedzhi, K. and Chen, Y. and Boixo, S.},
  title = {Phase transitions in random circuit sampling},
  journal = {Nature},
  volume = {634},
  pages = {328--333},
  year = {2024},
  doi = {10.1038/s41586-024-07998-6},
}

@article{google2025observation,
  author = {{Google Quantum AI and Collaborators}},
  title = {Observation of constructive interference at the edge of quantum ergodicity},
  journal = {Nature},
  volume = {646},
  pages = {825--830},
  year = {2025},
  doi = {10.1038/s41586-025-09526-6},
}

@article{manna2025projected,
  author = {Manna, Sandipan and Roy, Sthitadhi and Sreejith, G. J.},
  title = {Projected ensemble in a system with locally supported conserved charges},
  journal = {Phys. Rev. B},
  volume = {111},
  pages = {144302},
  year = {2025},
  doi = {10.1103/physrevb.111.144302},
}

@misc{liu2026conditional,
  author = {Liu, Yinchen and McGinley, Max and Schuster, Thomas and Gosset, David},
  title = {Conditional dependence and {Scrooge} ensembles in shallow random quantum circuits},
  year = {2026},
  eprint = {2608.12255},
  archivePrefix = {arXiv},
  primaryClass = {quant-ph},
}

@misc{liu2026emergence,
  author = {Liu, Zeyu and Zhang, Pengfei},
  title = {Emergence of the {Scrooge} ensemble in the {Sachdev--Ye--Kitaev} model},
  year = {2026},
  eprint = {2607.04864},
  archivePrefix = {arXiv},
  primaryClass = {quant-ph},
}

@misc{mcginley2025scroogeensemblemanybodyquantum,
  author = {McGinley, Max and Schuster, Thomas},
  title = {The {Scrooge} ensemble in many-body quantum systems},
  year = {2025},
  eprint = {2511.17172},
  archivePrefix = {arXiv},
  primaryClass = {quant-ph},
}

@misc{mok2026nature,
  author = {Mok, Wai-Keong and Haug, Tobias and Ho, Wen Wei and Preskill, John},
  title = {Nature is stingy: Universality of {Scrooge} ensembles in quantum many-body systems},
  year = {2026},
  eprint = {2601.00266},
  archivePrefix = {arXiv},
  primaryClass = {quant-ph},
}

@article{chang2025deep,
  author = {Chang, Rui-An and Shrotriya, Harshank and Ho, Wen Wei and Ippoliti, Matteo},
  title = {Deep thermalization under charge-conserving quantum dynamics},
  journal = {PRX Quantum},
  volume = {6},
  pages = {020343},
  year = {2025},
  doi = {10.1103/prxquantum.6.020343},
}

@article{varikuti2024unraveling,
  author = {Varikuti, Naga Dileep and Bandyopadhyay, Soumik},
  title = {Unraveling the emergence of quantum state designs in systems with symmetry},
  journal = {Quantum},
  volume = {8},
  pages = {1456},
  year = {2024},
  doi = {10.22331/q-2024-08-29-1456},
}

@article{liu2024deep,
  author = {Liu, Chang and Huang, Qi Camm and Ho, Wen Wei},
  title = {Deep thermalization in {Gaussian} continuous-variable quantum systems},
  journal = {Phys. Rev. Lett.},
  volume = {133},
  pages = {260401},
  year = {2024},
  doi = {10.1103/physrevlett.133.260401},
}

@article{bejan2025matchgate,
  author = {Bejan, Mircea and B{\'e}ri, Benjamin and McGinley, Max},
  title = {Matchgate circuits deeply thermalize},
  journal = {Phys. Rev. Lett.},
  volume = {135},
  pages = {020401},
  year = {2025},
  doi = {10.1103/v8kp-39ry},
}

@article{lucas2023generalized,
  author = {Lucas, Maxime and Piroli, Lorenzo and De Nardis, Jacopo and De Luca, Andrea},
  title = {Generalized deep thermalization for free fermions},
  journal = {Phys. Rev. A},
  volume = {107},
  pages = {032215},
  year = {2023},
  doi = {10.1103/physreva.107.032215},
}

@misc{mark2026deep,
  author = {Mark, Daniel K and Endres, Manuel and Ippoliti, Matteo and Ho, Wen Wei and Choi, Soonwon},
  title = {Deep thermalization and {Hilbert}-space ergodicity},
  year = {2026},
  eprint = {2609.30248},
  archivePrefix = {arXiv},
  primaryClass = {quant-ph},
}

@misc{wu2026exact,
  author = {Wu, Yue and Tong, Yuzhi and Mao, Liang and Zhang, Pengfei},
  title = {Exact {Hilbert}-space ergodicity from continuous monitoring},
  year = {2026},
  eprint = {2606.29042},
  archivePrefix = {arXiv},
  primaryClass = {quant-ph},
}

@inproceedings{ananth2022cryptography,
  author = {Ananth, Prabhanjan and Qian, Luowen and Yuen, Henry},
  title = {Cryptography from pseudorandom quantum states},
  booktitle = {Annual International Cryptology Conference},
  pages = {208--236},
  publisher = {Springer},
  year = {2022},
}

@inproceedings{kretschmer2023quantum,
  author = {Kretschmer, William and Qian, Luowen and Sinha, Makrand and Tal, Avishay},
  title = {Quantum cryptography in algorithmica},
  booktitle = {Proceedings of the 55th Annual {ACM} Symposium on Theory of Computing},
  pages = {1589--1602},
  year = {2023},
}

@article{sekino2008fast,
  author = {Sekino, Yasuhiro and Susskind, Leonard},
  title = {Fast scramblers},
  journal = {Journal of High Energy Physics},
  volume = {2008},
  number = {10},
  pages = {065},
  year = {2008},
  doi = {10.1088/1126-6708/2008/10/065},
}

@article{low2019hamiltonian,
  author = {Low, Guang Hao and Chuang, Isaac L.},
  title = {Hamiltonian simulation by qubitization},
  journal = {Quantum},
  volume = {3},
  pages = {163},
  year = {2019},
  doi = {10.22331/q-2019-07-12-163},
}

@inproceedings{chakraborty2019power,
  author = {Chakraborty, Shantanav and Gily{\'e}n, Andr{\'a}s
            and Jeffery, Stacey},
  title = {The power of block-encoded matrix powers:
           Improved regression techniques via faster
           {Hamiltonian} simulation},
  booktitle = {46th International Colloquium on Automata,
               Languages, and Programming (ICALP 2019)},
  series = {Leibniz International Proceedings in Informatics},
  volume = {132},
  pages = {33:1--33:14},
  year = {2019},
  doi = {10.4230/LIPIcs.ICALP.2019.33},
}

@article{martyn2021grand,
  author = {Martyn, John M. and Rossi, Zane M. and
            Tan, Andrew K. and Chuang, Isaac L.},
  title = {Grand unification of quantum algorithms},
  journal = {PRX Quantum},
  volume = {2},
  pages = {040203},
  year = {2021},
  doi = {10.1103/PRXQuantum.2.040203},
}

@article{low2017optimal,
  author = {Low, Guang Hao and Chuang, Isaac L.},
  title = {Optimal {Hamiltonian} simulation by quantum signal processing},
  journal = {Physical Review Letters},
  volume = {118},
  pages = {010501},
  year = {2017},
  doi = {10.1103/PhysRevLett.118.010501},
}

@article{childs2017quantum,
  author = {Childs, Andrew M. and Kothari, Robin and Somma, Rolando D.},
  title = {Quantum algorithm for systems of linear equations
           with exponentially improved dependence on precision},
  journal = {SIAM Journal on Computing},
  volume = {46},
  number = {6},
  pages = {1920--1950},
  year = {2017},
  doi = {10.1137/16M1087072},
}

@article{pawlik2026restricted,
  title={Restricted typicality in non-equilibrium quantum many-body systems},
  author={Pawlik, Konrad and Sierant, Piotr and Zakrzewski, Jakub},
  journal={arXiv preprint arXiv:2609.07832},
  year={2026}
}
\clearpage
\appendix
\onecolumngrid

\numberwithin{equation}{section}
\setcounter{theorem}{0}
\setcounter{figure}{0}
\renewcommand{\thefigure}{S\arabic{figure}}
\renewcommand{\thelemma}{S\arabic{lemma}}
\renewcommand{\thetheorem}{S\arabic{theorem}}
\renewcommand{\thecorollary}{S\arabic{corollary}}

\makeatletter
\begingroup
\let\SavedContentsLine\contentsline
\let\SavedTocRestore\tocdepth@restore
\renewcommand{\contentsline}[4]{}
\def\tocdepth@restore{%
  \SavedTocRestore
  \let\contentsline\SavedContentsLine
}
\tableofcontents
\endgroup
\makeatother

\section{Notations}
\label{sec:common-definitions}

\begin{itemize}
    \item $\rho$: Density operator on a $D$-dimensional Hilbert space $\mathcal H\simeq\mathbb C^D$.

    \item $\M:=\lfloor\|\rho\|_\infty^{-1}\rfloor$: Floored inverse of the largest eigenvalue of $\rho$.


    \item $D$: Dimension of the Hilbert space.

    \item $n:=\lceil\log_2D\rceil$: Number of qubits required to encode the system. For an $n$-qubit system, $D=2^n$.

    \item $k$: Target moment order.

    \item $\lambda$: Polynomial-approximation order.

    \item $a_i:=(i-1)(k-1)$: Number of auxiliary subsystems in the $i$-th polynomial term.

    \item $k_i:=k+a_i=i(k-1)+1$: Total number of subsystems in the $i$-th polynomial term.

    \item $c_i:=D^{k_i}/\binom{D+k_i-1}{k_i}$: Coefficient associated with the $i$-th polynomial term.


    \item $r_\phi:=D\bra{\phi}\rho\ket{\phi}$: Rescaled expectation value of $\rho$ in the state $\ket{\phi}$.

    \item $\Exp{\phi\sim\mathrm{Haar}}{f(\phi)}$: Haar expectation of $f(\phi)$.

    \item $\norm{f(\phi)}_{L^p}:=\bigl(\Exp{\phi\sim\mathrm{Haar}}{|f(\phi)|^p}\bigr)^{1/p}$: $L^p$ norm with respect to the Haar measure.

    \item $\Pi_{\mathrm{sym}}^{(t)}$: Projector onto the totally symmetric subspace $\mathrm{Sym}^t(\mathbb C^D)$.

    \item $\chi_{\mathrm{Haar}}^{(t)}:=\Exp{\phi\sim\mathrm{Haar}}{\ketbrat{\phi}^{\otimes t}}=\Pi_{\mathrm{sym}}^{(t)}/\binom{D+t-1}{t}$: Haar $t$-th moment.

    \item $\chi_{\mathrm{Sc}}^{(k)}(\rho)$: Scrooge $k$-th moment associated with $\rho$. The dependence on $\rho$ may be omitted when unambiguous.


    \item $\chi_{\mathrm{PA}}^{(k)}(\lambda)$: Polynomial approximation of order $\lambda$ to the Scrooge $k$-th moment.

    \item $p_{\mathrm{sym}}^{(t)}:=\Tr[\rho^{\otimes t}\Pi_{\mathrm{sym}}^{(t)}]$: Probability of projecting $\rho^{\otimes t}$ onto the totally symmetric subspace, with $p_{\mathrm{sym}}^{(0)}:=1$.
    

    \item $\rho_{\mathrm{sym}}^{(t)}:=\Pi_{\mathrm{sym}}^{(t)}\rho^{\otimes t}\Pi_{\mathrm{sym}}^{(t)}/p_{\mathrm{sym}}^{(t)}$: Normalized projection of $\rho^{\otimes t}$ onto the totally symmetric subspace.

    \item $\Tr_{\mathrm{aux}}$: Partial trace over the last $a_i=k_i-k$ subsystems of the $i$-th polynomial term.


    \item $D^{(t)}:=D(D+1)\cdots(D+t-1)$: Rising factorial.

    \item $\epsilon_{\mathrm t}$: Total target error.

    \item $\epsilon_{\mathrm{Sch}}$: Schur-transform implementation error.

    \item $\delta$: Failure probability.
\end{itemize}

\section{Polynomial approximation for Scrooge moments}
\label{app:approximation}
In this appendix, we derive additive and relative error bounds for the polynomial approximation introduced in the main text. We use the notation of Appendix~\ref{sec:common-definitions} and consider integers $k\geq2$ and $\lambda\geq1$. Recall that the Scrooge $k$-th moment is
\begin{equation}
\chi_{\mathrm{Sc}}^{(k)}(\rho)
=\mathbb{E}_{\phi\sim\mathrm{Haar}}\left[
\frac{\left(\sqrt{D\rho}\ketbra{\phi}{\phi}\sqrt{D\rho}\right)^{\otimes k}}{r_\phi^{k-1}}
\right].
\label{eq:scrooge-moment}
\end{equation}
The polynomial
\begin{equation}
g_\lambda(x)
:=\sum_{j=0}^{\lambda-1}(1-x)^j
=\sum_{i=1}^{\lambda}\binom{\lambda}{i}(-1)^{i-1}x^{i-1}
\end{equation}
approximates the normalization factor by replacing $1/r_\phi^{k-1}$ with $g_\lambda(r_\phi^{k-1})$. This gives
\begin{align}
\chi_{\mathrm{PA}}^{(k)}(\lambda)
&:=\mathbb{E}_{\phi\sim\mathrm{Haar}}\left[
\left(\sqrt{D\rho}\ketbra{\phi}{\phi}\sqrt{D\rho}\right)^{\otimes k}
g_\lambda(r_\phi^{k-1})
\right]\nonumber\\
&=\sum_{i=1}^{\lambda}\binom{\lambda}{i}(-1)^{i-1}D^{k_i}
\mathrm{Tr}_{\mathrm{aux}}\!\left[
\rho^{\otimes k_i}\chi_{\mathrm{Haar}}^{(k_i)}
\right],
\end{align}
as in Eq.~\eqref{eq:pm-decomposition}, where $\mathrm{Tr}_{\mathrm{aux}}$ traces out the last $a_i=k_i-k$ systems in each term.

The error analysis starts from the exact identity
\begin{equation}
\frac{1}{r^{k-1}}-g_\lambda(r^{k-1})
=\frac{(1-r^{k-1})^{\lambda}}{r^{k-1}},
\qquad r>0.
\end{equation}
This identity holds for every finite $\lambda$, without a convergence assumption on the infinite geometric series. It reduces the approximation error to Haar averages involving $(1-r_\phi^{k-1})^{\lambda}$. We first bound the additive error in trace norm and then establish a relative error bound in the positive-semidefinite order. Theorems~\ref{thm:additive-error} and~\ref{thm:relative-error} identify ranges of $\lambda$ over which the corresponding upper bounds decrease monotonically.

\subsection{Additive error}
Here we derive the additive error bound for the polynomial approximation formula. Let
\begin{equation}
\M:=\left\lfloor \|\rho\|_\infty^{-1}\right\rfloor \le D.
\label{eq:poly-M-def}
\end{equation}
We first establish the Haar moment bounds used throughout the proof.
These bounds follow from Lemma 2 in
Ref.~\cite{mcginley2025scroogeensemblemanybodyquantum}; a self-contained
proof with the constants needed here is given below.
\begin{lemma}
[Haar moment bounds]
\label{lem:tommy's}
For integers $s\geq1$, $1\leq t\leq\M$, and $0\leq q<\M$,
\begin{align}
\mathbb{E}_{\phi\sim\mathrm{Haar}}[r_\phi^s]
&\le
\exp\!\left(\frac{s(s-1)}{2\M}\right),
\label{eq:poly-mcg-positive-moment}\\
\mathbb{E}_{\phi\sim\mathrm{Haar}}\left[|1-r_\phi|^t\right]
&\le
\,\frac{\Gamma(t/2+1)}{(\M/6)^{t/2}},
\label{eq:poly-mcg-central-moment}\\
\mathbb{E}_{\phi\sim\mathrm{Haar}}[r_\phi^{-q}]
&\le
\exp\!\left(\frac{q(q+1)}{2(\M-q)}\right).
\label{eq:poly-mcg-reciprocal-moment}
\end{align}
\end{lemma}

\begin{proof}
If $\M=D$, then $\rho=I/D$ and $r_\phi=1$, so the claims are immediate. Assume $\M<D$. Write the eigenvalues of
$\rho$ as $\beta_j$, so that
\begin{equation}
 r_\phi=D\sum_{j=1}^D\beta_j|\langle j|\phi\rangle|^2,
 \qquad 0\leq\beta_j\leq1/\M,\qquad \sum_j\beta_j=1.
\end{equation}
For convex $f$, the expectation $\mathbb E[f(r_\phi)]$, viewed as a function of the eigenvalues, is convex and invariant under permutations of the eigenvalues. It is therefore Schur-convex: $\mathbb E_{\phi}\!\left[f(r_\phi(\boldsymbol{x}))\right]
\ge
\mathbb E_{\phi}\!\left[f(r_\phi(\boldsymbol{y}))\right]$ whenever the eigenvalue vector $\boldsymbol{x}$ majorizes $\boldsymbol{y}$. Here, majorization means that, after arranging both vectors in nonincreasing order, the sum of the largest $k$ eigenvalues of
$\boldsymbol{x}$ is at least that of $\boldsymbol{y}$ for every
$k=1,\ldots,D-1$, while their total sums are equal. The spectrum of $\rho$ is majorized by that of the flat density matrix
\begin{equation}
 \operatorname{diag}\!\left(
 \underbrace{\frac1\M,\ldots,\frac1\M}_{\M\ \mathrm{entries}},
 0,\ldots,0\right).
\end{equation}
Indeed, the largest $\ell$ eigenvalues sum to at most $\ell/\M$ when
$\ell\leq\M$, and to at most one when $\ell>\M$. Hence
\begin{equation}
 \mathbb E[f(r_\phi)]\leq\mathbb E[f(R)],
 \qquad R:=\frac D\M\sum_{j=1}^{\M}|\langle j|\phi\rangle|^2.
\end{equation}
We apply this to the convex functions $r^s$, $r^{-q}$, and $|1-r|^t$.

A Haar state can be generated by normalizing independent complex
Gaussian amplitudes. Their squared moduli are independent unit-rate exponential variables; the normalized weights are uniform on the probability simplex~\cite{sykora1974quantum}. Thus the fraction of the norm in the first $\M$ components satisfies
$\sum_{j=1}^{\M}|\langle j|\phi\rangle|^2 = \frac{\M}{D}R\sim\operatorname{Beta}(\M,D-\M)$.
We use $r$ to denote a possible value of $R$, and $p(r)$ to denote the probability density of $R$ at $r$:
\begin{equation}
 p(r)=\frac{(\M/D)^\M}{\mathrm B(\M,D-\M)}
 r^{\M-1}\left(1-\frac{\M r}{D}\right)^{D-\M-1},
 \qquad 0<r<\frac D\M.
\end{equation}
The beta density therefore gives the following exact moment formula:
\begin{equation}
 \mathbb E[R^s]
 =\int_0^{D/\M}r^s p(r)\,dr
 =\frac{D^s\Gamma(D)}{\Gamma(D+s)}
   \frac{\Gamma(\M+s)}{\M^s\Gamma(\M)}.
 \label{eq:haar-flat-exact-moment}
\end{equation}
Here we use the fact that $\mathrm B(a,b)=\Gamma(a)\Gamma(b)/\Gamma(a+b)$.

For integer $s\geq1$, the gamma recurrence yields
\begin{equation}
 \mathbb E[r_\phi^s]
 \leq\prod_{j=0}^{s-1}\frac{1+j/\M}{1+j/D}
 \leq\exp\!\left(\sum_{j=0}^{s-1}\frac j\M\right)
 =e^{s(s-1)/(2\M)},
\end{equation}
where we used $\log(1+x)\leq x$.
Similarly, taking $s=-q$ in Eq.~\eqref{eq:haar-flat-exact-moment} gives,
for integer $0\leq q<\M$,
\begin{equation}
 \mathbb E[r_\phi^{-q}]
 \leq\prod_{j=1}^{q}\frac{1-j/D}{1-j/\M}
 \leq\exp\!\left(\sum_{j=1}^{q}\frac j{\M-j}\right)
 \leq e^{q(q+1)/(2(\M-q))}.
\end{equation}
The last two steps use $-\log(1-x)\leq x/(1-x)$ and
$\M-j\geq\M-q$.

For the central moments, we first derive a recurrence relating
orders $t$ and $t-2$, then iterate it from the zeroth and first moments. For integer $t\geq2$,  $ \mathbb E[|R-1|^t]
 =\int_0^{D/\M}|r-1|^t p(r)\,dr$. To lower the order of this moment by integration by parts, we write
$|r-1|^t=(r-1)|r-1|^{t-2}(r-1)$ and use the identity
\begin{equation}
 \frac{d}{dr}\left[r\left(1-\frac{\M r}{D}\right)p(r)\right]
 =\M(1-r)p(r),
\end{equation}
then 
\begin{equation}
 \begin{aligned}
 \mathbb E[|R-1|^t]
 &=-\frac1\M\int_0^{D/\M}(r-1)|r-1|^{t-2}
 \frac{d}{dr}\left[r\left(1-\frac{\M r}{D}\right)p(r)\right]dr\\
 &=\frac{t-1}{\M}\int_0^{D/\M}
 |r-1|^{t-2}r\left(1-\frac{\M r}{D}\right)p(r)\,dr\\
 &=\frac{t-1}{\M}
 \mathbb E\!\left[R\left(1-\frac{\M R}{D}\right)|R-1|^{t-2}\right].
 \end{aligned}
 \label{eq:beta-central-identity}
\end{equation}
The boundary term vanishes. Finally, $0\leq1-\M R/D\leq1$ and $R\leq1+|R-1|$ imply
\begin{equation}
 \mathbb E[|R-1|^t]
 \leq\frac{t-1}{\M}
 \left(\mathbb E[|R-1|^{t-2}]+\mathbb E[|R-1|^{t-1}]\right).
\end{equation}
To obtain a recurrence involving only orders $t$ and $t-2$, we bound
the $(t-1)$-th moment using
\begin{equation}
 \frac{t-1}{\M}|R-1|^{t-1}
 \leq\frac12|R-1|^t
 +\frac{(t-1)^2}{2\M^2}|R-1|^{t-2}.
\end{equation}
Therefore, 
\begin{equation}
 \begin{aligned}
 \mathbb E[|R-1|^t]
 &\leq\frac{t-1}{\M}\left(2+\frac{t-1}{\M}\right)
        \mathbb E[|R-1|^{t-2}]\\
 &\leq\frac{3t}{\M}\mathbb E[|R-1|^{t-2}],
 \qquad 2\leq t\leq\M.
 \end{aligned}
 \label{eq:beta-central-recursion}
\end{equation}
Here $t\leq\M$ ensures $(t-1)/\M\leq1$, which gives the last inequality. 
The zeroth and first moments satisfy
$\mathbb E[|R-1|^0]=1$ and
\begin{equation}
 \mathbb E[|R-1|]\leq\sqrt{\operatorname{Var}(R)}
 =\sqrt{\frac{D-\M}{\M(D+1)}}
 \leq\frac1{\sqrt\M}\leq\Gamma(3/2)\sqrt{\frac6\M}.
\end{equation}
The factor $3t/\M$ in Eq.~\eqref{eq:beta-central-recursion} is exactly
the ratio of $\Gamma(t/2+1)(6/\M)^{t/2}$ to the same expression at
order $t-2$, since $\Gamma(t/2+1)=(t/2)\Gamma(t/2)$.
Thus induction from these two initial estimates gives
\begin{equation}
 \mathbb E[|1-r_\phi|^t]\leq\mathbb E[|R-1|^t]
 \leq\Gamma(t/2+1)\left(\frac6\M\right)^{t/2},
 \qquad 1\leq t\leq\M.
\end{equation}
This proves the central-moment bound.
\end{proof}

With these bounds, we can prove the additive error bound.

\begin{lemma}
\label{lem:additive}
Assume $k\ge 2$ and $2\lambda \le \M$. Denote the additive error as
\begin{equation}
\Delta_{k,\lambda}^{\mathrm{add}}
:=
\norm{\chi_{\mathrm{Sc}}^{(k)}-\chi_{\mathrm{PA}}^{(k)}(\lambda)}_1.
\end{equation}
Then
\begin{equation}
\Delta_{k,\lambda}^{\mathrm{add}}
\le
 6^{\lambda /2} \frac{\sqrt{\lambda!}}{\M^{\lambda/2}}
\,(k-1)^{\lambda-1}
\sum_{i=0}^{k-2}
\exp\!\left(\frac{(2i\lambda+1)(2i\lambda+2)}{4\M}\right).
\label{eq:poly-add-explicit}
\end{equation}
\end{lemma}

\begin{proof}
Comparing the Scrooge $k$-th moment and our approximation formula, we have
\begin{align}
    \Delta_{k,\lambda}^{\mathrm{add}}&=\norm{\chi_{\mathrm{Sc}}^{(k)}-\chi_{\mathrm{PA}}^{(k)}(\lambda)}_1\notag\\
    &\leq \mathbb{E}_{\phi\sim\mathrm{Haar}}\left[\norm{
\frac{\left(\sqrt{D\rho}\ketbra{\phi}{\phi}\sqrt{D\rho}\right)^{\otimes k}}{r_\phi^{k-1}}\left(1-r^{k-1}_\phi\right)^\lambda}_1
\right]\notag\\
&=\Exp{\phi\sim\mathrm{Haar}}{
\frac{\left|1-r^{k-1}_\phi\right|^\lambda}{r_\phi^{k-1}}
\trace{}{\left(\sqrt{D\rho}\ketbra{\phi}{\phi}\sqrt{D\rho}\right)^{\otimes k}}
}\notag\\
&=\Exp{\phi\sim\mathrm{Haar}}{
r_\phi\left|1-r^{k-1}_\phi\right|^\lambda
}
\end{align}

Using $1-r^{k-1}=(1-r)(1+r+r^2+\cdots+r^{k-2})$ and Jensen's inequality,
\begin{align}
    \left(\sum_{i=1}^n\nu_ir_i\right)^\lambda\leq \sum_{i=1}^n\nu_ir_i^\lambda
\end{align}
for $\lambda\geq 1$, $\nu_i\geq 0$ and $\sum_{i=1}^n\nu_i=1$, we obtain
\begin{align}
\Delta_{k,\lambda}^{\mathrm{add}}
&\leq\mathbb{E}_{\phi\sim\mathrm{Haar}}
\left[r_\phi|1-r_\phi|^\lambda(1+r_\phi+\cdots+r_\phi^{k-2})^\lambda\right]
\notag\\
&\leq (k-1)^{\lambda-1}\Exp{\phi\sim\mathrm{Haar}}{|1-r_\phi|^\lambda\sum_{i=0}^{k-2}r_\phi^{i\lambda+1}}
\notag\\
&=
(k-1)^{\lambda-1}
\sum_{i=0}^{k-2}
\mathbb{E}_{\phi\sim\mathrm{Haar}}
\left[
|1-r_\phi|^\lambda r_\phi^{i\lambda+1}
\right]
\nonumber\\
&\le
(k-1)^{\lambda-1}
\sum_{i=0}^{k-2}
\left[
\mathbb{E}_{\phi\sim\mathrm{Haar}}\left[|1-r_\phi|^{2\lambda}\right]
\right]^{1/2}
\left[
\mathbb{E}_{\phi\sim\mathrm{Haar}}\left[r_\phi^{2i\lambda+2}\right]
\right]^{1/2}.
\label{eq:poly-cs-reduction-reorg}
\end{align}
In the last inequality we use the Cauchy--Schwarz inequality.

We now use Lemma~\ref{lem:tommy's} to bound the two factors separately. First, by
\eqref{eq:poly-mcg-central-moment} with $t=2\lambda$,
\begin{equation}
\left[
\mathbb{E}_{\phi\sim\mathrm{Haar}}\left[|1-r_\phi|^{2\lambda}\right]
\right]^{1/2}
\le
\left(
\frac{\Gamma(\lambda+1)}{(\M/6)^\lambda}
\right)^{1/2}
=
 6^{\lambda /2} \frac{\sqrt{\lambda!}}{\M^{\lambda/2}}.
\label{eq:poly-central-moment-reorg}
\end{equation}
Second, by \eqref{eq:poly-mcg-positive-moment},
\begin{equation}
\left[
\mathbb{E}_{\phi\sim\mathrm{Haar}}\left[r_\phi^{2i\lambda+2}\right]
\right]^{1/2}
\le
\exp\!\left(\frac{(2i\lambda+1)(2i\lambda+2)}{4\M}\right).
\label{eq:poly-positive-moment-reorg}
\end{equation}
Substituting \eqref{eq:poly-central-moment-reorg} and
\eqref{eq:poly-positive-moment-reorg} into
\eqref{eq:poly-cs-reduction-reorg} yields \eqref{eq:poly-add-explicit}.
\end{proof}

\begin{theorem}
[Additive error of polynomial approximation]
\label{thm:additive-error}
For any integer $k\geq2$ satisfying
$12(k-1)^2\leq \M$, define
\begin{equation}
    \lambda_* :=
    \left\lfloor
    \frac{\M}{12(k-1)^2}
    \right\rfloor
    \geq 1.
    \label{eq:poly-lambda-star}
\end{equation}
Then there exists a universal constant
\begin{equation}
    C=e^{5/12}
\end{equation}
such that, for every integer $1\leq \lambda\leq\lambda_*$,
\begin{equation}
    \Delta_{k,\lambda}^{\mathrm{add}}
    =
    \norm{
        \chi_{\mathrm{Sc}}^{(k)}
        -
        \chi_{\mathrm{PA}}^{(k)}(\lambda)
    }_1
    \leq
    C
    \left(
        \frac{6\lambda(k-1)^2}
        {e^{1/3}\M}
    \right)^{\lambda/2}.
    \label{eq:poly-add-monotone-bound}
\end{equation}
The right-hand side is monotonically decreasing with $\lambda$ for
$1\leq\lambda\leq\lambda_*$.  In particular, the smallest upper bound is
obtained at $\lambda=\lambda_*$:
\begin{equation}
\Delta_{k,\lambda_*}^{\mathrm{add}}
\leq
C_{\add}\exp\!\left[-\frac{\M}{24(k-1)^2}\right],
\label{eq:poly-add-min-bound}
\end{equation}
where
\begin{equation}
C_{\add}:=\sqrt{e}C=e^{11/12}.
\end{equation}
\end{theorem}

\begin{proof}
    This result is obtained by simplifying Lemma~\ref{lem:additive}. Since the
exponential factor is increasing in $i$, the sum is bounded by its largest term:
\begin{equation}
\sum_{i=0}^{k-2}
\exp\!\left(\frac{(2i\lambda+1)(2i\lambda+2)}{4\M}\right)
\le
(k-1)
\exp\!\left(\frac{(2\lambda(k-2)+1)(2\lambda(k-2)+2)}{4\M}\right).
\label{eq:poly-add-simplified-reorg}
\end{equation}

Now we use the Robbins bound $n!\le e\,n^{n+1/2}e^{-n}$:
\begin{equation}
\sqrt{\lambda!}
\le
\sqrt{e}        \lambda^{1/4}\left(\frac{\lambda}{e}\right)^{\lambda/2}.
\end{equation}
Substituting this into \eqref{eq:poly-add-simplified-reorg} gives
\begin{equation}
\Delta_{k,\lambda}^{\mathrm{add}}
\le
\sqrt{e}\lambda^{1/4}
\left(\frac{6\lambda(k-1)^2}{e\M}\right)^{\lambda/2}
\exp\!\left(\frac{(2\lambda(k-2)+1)(2\lambda(k-2)+2)}{4\M}\right).
\label{eq:poly-add-robbins-reorg}
\end{equation}

For $1\leq\lambda\leq\lambda_*$,
\begin{equation}
    \lambda(k-2)
    \leq
    \lambda(k-1)
    \leq
    \frac{\M}{12(k-1)},
\end{equation}
so
\begin{align}
    \frac{(2\lambda(k-2)+1)(2\lambda(k-2)+2)}{4\M}
    &\leq \frac{(2\lambda(k-1)+1)(2\lambda(k-1)+2)}{4\M} \nonumber\\
    &\leq \frac{\lambda_* (k-1)^2}{\M} \lambda+\frac{3\lambda_* (k-1)}{2\M}+\frac{1}{2\M} \nonumber\\
    &\leq \frac{\lambda}{12}+\frac{1}{8(k-1)}+\frac{1}{2\M}
\end{align}
The assumptions $12(k-1)^2\leq \M$ and $k\geq2$ give
\begin{equation}
    \frac{1}{8(k-1)} \leq \frac{1}{8} \qquad \text{and} \qquad \frac{1}{2\M} \leq \frac{1}{24} .
\end{equation}
Hence
\begin{equation}
\exp\!\left(\frac{(2\lambda(k-2)+1)(2\lambda(k-2)+2)}{4\M}\right)
    \leq
    \exp\left(\frac{\lambda}{12} +\frac{1}{6}\right).
\label{eq:poly-add-exp-bound}
\end{equation}

Combining Eqs.~\eqref{eq:poly-add-robbins-reorg} and
\eqref{eq:poly-add-exp-bound}, we obtain
\begin{align}
    \Delta_{k,\lambda}^{\mathrm{add}}
    &\leq
    e^{2/3}
    \lambda^{1/4}
    \left(\frac{6\lambda(k-1)^2}{e^{5/6}\M}\right)^{\lambda/2}
    \nonumber\\
    &=
    e^{2/3}
    \left(\lambda^{1/4}e^{-\lambda/4}\right)
    \left(
        \frac{6\lambda(k-1)^2}
        {e^{1/3}\M}
    \right)^{\lambda/2}.
\end{align}
Since $x^{1/4}e^{-x/4}$ is monotonically decreasing for $x\geq1$,
\begin{equation}
    \lambda^{1/4}e^{-\lambda/4}
    \leq e^{-1/4}.
\end{equation}
Thus
\begin{equation}
    \Delta_{k,\lambda}^{\mathrm{add}}
    \leq
    C
    \left(
        \frac{6\lambda(k-1)^2}
        {e^{1/3}\M}
    \right)^{\lambda/2},
\end{equation}
where
\begin{equation}
    C=e^{5/12}.
\end{equation}

To see that the right-hand side decreases monotonically with $\lambda$, 
we consider
\begin{equation}
    \lambda
    \log\left(
        \frac{6\lambda(k-1)^2}
        {e^{1/3}\M}
    \right).
\end{equation}
Its derivative with respect to $\lambda$ is
\begin{equation}
    \log\left(
        \frac{6\lambda(k-1)^2}
        {e^{1/3}\M}
    \right)+1.
\end{equation}
It is non-positive for $1\leq\lambda\leq\lambda_*$, so the right-hand side of
Eq.~\eqref{eq:poly-add-monotone-bound}
is monotonically decreasing with $\lambda$ throughout
$1\leq\lambda\leq\lambda_*$.

At $\lambda=\lambda_*$,
\begin{equation}
\frac{6\lambda_*(k-1)^2}{e^{1/3}\M}
\leq\frac{1}{2e^{1/3}}\leq e^{-1}.
\end{equation}
So, 
\begin{align}
    C
    \left(
        \frac{6\lambda_*(k-1)^2}
        {e^{1/3}\M}
    \right)^{\lambda_*/2} &\leq C \exp{\left(-\frac{\lambda
    _*}{2}\right)} \leq C \exp{\left(-\frac{\left(\frac{\M}{12(k-1)^2}-1\right)}{2}\right)} \nonumber \\
    & = \sqrt{e}C \exp{\left(-\frac{\M}{24(k-1)^2}\right)},
\end{align}
which proves Eq.~\eqref{eq:poly-add-min-bound}.
\end{proof}

\begin{corollary}
[Sufficient $\lambda$ for a given additive error]
\label{coro:minimal-additive-error}
Let $k\geq2$ be an integer with $12(k-1)^2\leq\M$, and let
$C_{\add}$ be the constant in Theorem~\ref{thm:additive-error}.
If $\epsilon>0$ satisfies
\begin{equation}
\epsilon\geq
C_{\add}\exp\!\left[-\frac{\M}{24(k-1)^2}\right],
\label{eq:additive-certified-floor}
\end{equation}
it suffices to choose
\begin{equation}
\lambda=\max\left\{
\left\lfloor
2\log\left(\frac{C_{\add}
}{\epsilon}\right)
\right\rfloor,1
\right\}.
\label{eq:additive-order-choice}
\end{equation}
Then $\norm{\chi_{\mathrm{Sc}}^{(k)}
-\chi_{\mathrm{PA}}^{(k)}(\lambda)}_1\leq\epsilon$.
\end{corollary}

\begin{proof}
The assumption on
$\epsilon$ gives
\begin{equation}
2\log\left(\frac{C_{\add}
}{\epsilon}\right)\leq\frac{\M}{12(k-1)^2}.
\end{equation}
Since $\lambda_*\geq1$, taking the floor and then the maximum
with $1$ yields $1\leq\lambda\leq\lambda_*$.
Theorem~\ref{thm:additive-error} therefore applies. Moreover,
\begin{equation}
\frac{6\lambda(k-1)^2}{e^{1/3}\M}
\leq\frac{1}{2e^{1/3}}\leq e^{-1},
\qquad
\lambda\geq2\log\left(\frac{C_{\add}
}{\epsilon}\right)-1.
\end{equation}
Using these inequalities and $C_{\add}=\sqrt{e}C$, we obtain
\begin{equation}
\Delta_{k,\lambda}^{\mathrm{add}}
\leq C
    \left(
        \frac{6\lambda(k-1)^2}
        {e^{1/3}\M}
    \right)^{\lambda/2}
\leq C e^{-\lambda/2}
\leq C e^{-\log\left(\frac{C_{\add}
}{\epsilon}\right)+1/2}
=\epsilon.
\end{equation}
\end{proof}

\subsection{Relative error}
Now we give a relative error upper bound.

For this subsection, we introduce the following notations.
\begin{equation}
\Pi_\phi^{(k)} := \ketbra{\phi}{\phi}^{\otimes k},
\qquad
\Xi_{\mathrm{Sc}}^{(k)}
:=
\mathbb{E}_{\phi\sim\mathrm{Haar}}
\left[
 r_\phi^{1-k}\Pi_\phi^{(k)}
\right],
\end{equation}
and similarly
\begin{equation}
\Xi_{\mathrm{PA}}^{(k)}
:=
\mathbb{E}_{\phi\sim\mathrm{Haar}}
\left[
 g_\lambda(r_\phi^{k-1})\Pi_\phi^{(k)}
\right].
\end{equation}
By the bound on negative moments in Lemma~\ref{lem:tommy's}, these operators are well defined for $k-1<\M$, which follows from the assumptions below. Here $\phi\sim\mathrm{Haar}(D)$. Their difference is
\begin{equation}
\Xi_{\mathrm{diff}}
:=
\Xi_{\mathrm{Sc}}^{(k)}-\Xi_{\mathrm{PA}}^{(k)}
=
\mathbb{E}_{\phi\sim\mathrm{Haar}}
\left[
 r_\phi^{1-k}(1-r_\phi^{k-1})^\lambda\Pi_\phi^{(k)}
\right].
\end{equation}
For a general positive integer $\lambda$, the operator $\Xi_{\mathrm{diff}}$ need not be positive semidefinite, so we estimate its quadratic form in absolute value.

Both $\Xi_{\mathrm{Sc}}^{(k)}$ and $\Xi_{\mathrm{PA}}^{(k)}$ are supported on $\Sym^k(\mathbb C^D)$, the symmetric subspace of $(\mathbb C^D)^{\otimes k}$. Let $\norm{f(\phi)}_{L^p}:=(\Exp{\phi\sim\mathrm{Haar}}{|f(\phi)|^p})^{1/p}$ define the $L^p$ norm over Haar measure. We first state a supporting lemma.
\begin{lemma}
\label{lem:holder}
For any normalized $\ket{x}\in\Sym^k(\mathbb C^D)$, define
\begin{equation}
y_x(\phi):=\mel{x}{\Pi_\phi^{(k)}}{x}.
\end{equation}
Then
\begin{align}
    \norm{y_x(\phi)}_{L^s}\leq 4^{k(1-1/s)}\frac{k!}{D^{(k)}},\qquad \forall s\in[1,2],
\end{align}
\end{lemma}

\begin{proof}
We use the Haar-moment identity
\begin{equation}
 \chi_{\mathrm{Haar}}^{(k)}
 =\frac{\Pi_{\mathrm{sym}}^{(k)}}{\binom{D+k-1}{k}}
 =\frac{k!}{D^{(k)}}\Pi_{\mathrm{sym}}^{(k)},
 \label{eq:common-haar-moment}
\end{equation}
where $\Pi_{\mathrm{sym}}^{(k)}$ projects onto the totally symmetric subspace.
Each tensor power is symmetric, and unitary invariance makes its Haar
average proportional to this projector; unit trace fixes the coefficient.
Observe that
\begin{equation}
\norm{y_x(\phi)}^2_{L^2}=
\mathbb{E}_{\phi\sim\mathrm{Haar}}[y_x(\phi)^2]
=\mel{x^{\otimes 2}}{\chi_{\mathrm{Haar}}^{(2k)}}{x^{\otimes 2}}
\le
\frac{(2k)!}{D^{(2k)}}
\le
4^k\left(\frac{k!}{D^{(k)}}\right)^2,
\label{eq:poly-y-second-moment}
\end{equation}
since $D^{(2k)}\ge (D^{(k)})^2$ and $(2k)!\le 4^k(k!)^2$. Using the interpolation inequality between $L^1$ and $L^2$, we have for $s\in[1,2]$
\begin{equation}
\norm{y_x}_{L^s}
\le
\norm{y_x}_{L^1}^{2/s-1}\norm{y_x}_{L^2}^{2-2/s}\leq
\left(\frac{k!}{D^{(k)}}\right)^{2/s-1}
\left(2^k\frac{k!}{D^{(k)}}\right)^{2-2/s}=4^{k(1-1/s)}\frac{k!}{D^{(k)}}.
\end{equation}
\end{proof}

\begin{lemma}
\label{lem:relative-error}
Assume $k\geq2$, $\lambda\geq1$, $2k\lambda\leq\M$, and
$2k(k-1)<\M$.
Then
\begin{equation}
(1-\Delta^{\mathrm{rel}}_{k,\lambda})\chi_{\mathrm{Sc}}^{(k)}
\preceq
\chi_{\mathrm{PA}}^{(k)}(\lambda)
\preceq
(1+\Delta^{\mathrm{rel}}_{k,\lambda})\chi_{\mathrm{Sc}}^{(k)},
\label{eq:poly-rel-main}
\end{equation}
where $\Delta^{\mathrm{rel}}_{k,\lambda}$ satisfies
\begin{equation}
\Delta^{\mathrm{rel}}_{k,\lambda}
\le
16(k-1)^\lambda
\exp\!\left(\frac{(k-1)(k(k-1)-1)}{2\M}\right)
\left(\frac{\Gamma(\lambda k+1)}{(\M/6)^{\lambda k}}\right)^{\!1/(2k)}
\Phi_{k,\lambda}^{\max},
\label{eq:poly-rel-explicit}
\end{equation}
with

\begin{equation}
\Phi_{k,\lambda}^{\max}
=
\max\set{
\exp\left[
\frac{(k-1)(1+2k(k-1))}
{2(\M-2k(k-1))}
\right],
\exp\left[
\frac{
((k-2)\lambda+1-k)
\left(2k((k-2)\lambda+1-k)-1\right)
}
{2\M}
\right]
}.
\label{eq:poly-rel-theta-max}
\end{equation}
\end{lemma}

\begin{proof}
Denote $y_x(\phi):=\mel{x}{\Pi_\phi^{(k)}}{x}$.
For any normalized $\ket{x}\in\Sym^k(\mathbb C^D)$, we obtain
\begin{align}
\abs{\mel{x}{\Xi_{\mathrm{diff}}}{x}}
&\le
(k-1)^{\lambda-1}\sum_{i=0}^{k-2}
\mathbb{E}_{\phi\sim\mathrm{Haar}}
\left[
|1-r_\phi|^{\lambda} r_\phi^{i\lambda+1-k} y_x(\phi)
\right] \nonumber\\
&\le
(k-1)^{\lambda-1}\sum_{i=0}^{k-2}
\norm{|1-r_\phi|^\lambda}_{L^{2k}}
\norm{r_\phi^{i\lambda+1-k}}_{L^{2k}}
\norm{y_x}_{L^{k/(k-1)}}.
\label{eq:poly-rel-numerator-holder}
\end{align}
On the first line we use Jensen's inequality. On the second we use H\"older's inequality with exponents
$(2k,2k,k/(k-1))$.

Since $2k\lambda\leq\M$, Eq.~\eqref{eq:poly-mcg-central-moment}
with $t=2k\lambda$ gives
\begin{equation}
\norm{|1-r_\phi|^\lambda}_{L^{2k}}
\le
\left(\frac{\Gamma(\lambda k+1)}{(\M/6)^{\lambda k}}\right)^{\!1/(2k)}.
\label{eq:poly-rel-central-L2k}
\end{equation}
For $u_i:=i\lambda+1-k$, if $u_i\ge 0$, we have
\begin{equation}
\norm{r_\phi^{u_i}}_{L^{2k}}
\le
\exp\!\left(\frac{u_i(2ku_i-1)}{2\M}\right),
\label{eq:poly-rel-positive-L2k}
\end{equation}
where the case $u_i=0$ is immediate, and the case $u_i\geq1$ follows from
Eq.~\eqref{eq:poly-mcg-positive-moment} with $s=2ku_i$. If $u_i<0$, the assumption $2k(k-1)<\M$ ensures
$2k(-u_i)<\M$, and Eq.~\eqref{eq:poly-mcg-reciprocal-moment} yields
\begin{equation}
\norm{r_\phi^{u_i}}_{L^{2k}}
\le
\exp\!\left(-
\frac{u_i(1-2ku_i)}{2(\M+2ku_i)}
\right).
\label{eq:poly-rel-negative-L2k}
\end{equation}
Define
\begin{equation}
\Phi_i:=
\begin{cases}
\exp\!\left(\dfrac{u_i(2ku_i-1)}{2\M}\right),
& u_i\geq0, \\
\exp\!\left(\dfrac{-u_i(1-2ku_i)}{2(\M+2ku_i)}\right),
& u_i<0.
\end{cases}
\label{eq:poly-rel-theta}
\end{equation}
Since $u_i=i\lambda+1-k$ increases with $i$, the negative branch is maximized at $i=0$, while the nonnegative branch, when present, is maximized at $i=k-2$. Hence, for $\lambda\geq2$ and $k\geq3$,
\begin{equation}
\Phi_{k,\lambda}^{\max}
=
\max\set{
\exp\left[
\frac{(k-1)(1+2k(k-1))}
{2(\M-2k(k-1))}
\right],
\exp\left[
\frac{
((k-2)\lambda+1-k)
\left(2k((k-2)\lambda+1-k)-1\right)
}
{2\M}
\right]
}.
\end{equation}
For $\lambda=1$ or $k=2$, only the negative branch is present. Therefore,
\begin{equation}
\sum_{i=0}^{k-2}\Phi_i
\leq
(k-1)\Phi_{k,\lambda}^{\max}.
\end{equation}

Combining all the above with Lemma~\ref{lem:holder}, we get
\begin{align}
\abs{\mel{x}{\Xi_{\mathrm{diff}}}{x}}
&\leq
(k-1)^\lambda
\left(\frac{\Gamma(\lambda k+1)}{(\M/6)^{\lambda k}}\right)^{\!1/(2k)}
\Phi_{k,\lambda}^{\max}
\norm{y_x}_{L^{\frac{k}{k-1}}}
\nonumber\\
&\leq
4(k-1)^\lambda
\left(\frac{\Gamma(\lambda k+1)}{(\M/6)^{\lambda k}}\right)^{\!1/(2k)}
\Phi_{k,\lambda}^{\max}
\frac{k!}{D^{(k)}}.
\label{eq:poly-rel-numerator-final}
\end{align}

To lower-bound the expectation value over $\Xi_{\mathrm{Sc}}^{(k)}$, Cauchy--Schwarz gives
\begin{align}
\left(\mathbb{E}_{\phi\sim\mathrm{Haar}}[y_x(\phi)]\right)^2
&\le
\mathbb{E}_{\phi\sim\mathrm{Haar}}\left[r_\phi^{1-k}y_x(\phi)\right]
\mathbb{E}_{\phi\sim\mathrm{Haar}}\left[r_\phi^{k-1}y_x(\phi)\right] \nonumber\\
&=
\mel{x}{\Xi_{\mathrm{Sc}}^{(k)}}{x}
\mathbb{E}_{\phi\sim\mathrm{Haar}}\left[r_\phi^{k-1}y_x(\phi)\right].
\label{eq:poly-rel-cauchy}
\end{align}
Using H\"older's inequality with exponents $(k,k/(k-1))$, together with
Lemma~\ref{lem:tommy's},
\begin{equation}
\mathbb{E}_{\phi\sim\mathrm{Haar}}\left[r_\phi^{k-1}y_x(\phi)\right]
\leq
\norm{r^{k-1}_\phi}_{L^k}\cdot\norm{y_x}_{L^{\frac{k}{k-1}}}
\le 4
\exp\!\left(\frac{(k-1)(k(k-1)-1)}{2\M}\right)
\cdot \frac{k!}{D^{(k)}}.
\label{eq:poly-rel-denominator-upper}
\end{equation}
These give
\begin{equation}
\mel{x}{\Xi_{\mathrm{Sc}}^{(k)}}{x}
\ge
\frac14
\exp\!\left(-\frac{(k-1)(k(k-1)-1)}{2\M}\right)
\frac{k!}{D^{(k)}}.
\label{eq:poly-rel-denominator-lower}
\end{equation}

Comparing \eqref{eq:poly-rel-numerator-final} with \eqref{eq:poly-rel-denominator-lower}, we obtain
\begin{equation}
\abs{\mel{x}{\Xi_{\mathrm{diff}}}{x}}
\le
\Delta^{\mathrm{rel}}_{k,\lambda}\mel{x}{\Xi_{\mathrm{Sc}}^{(k)}}{x}
\end{equation}
for every normalized $\ket{x}\in\Sym^k(\mathbb C^D)$, with
$\Delta^{\mathrm{rel}}_{k,\lambda}$ given by \eqref{eq:poly-rel-explicit}. Hence
\begin{equation}
-\Delta^{\mathrm{rel}}_{k,\lambda}\Xi_{\mathrm{Sc}}^{(k)}
\preceq
\Xi_{\mathrm{diff}}
\preceq
\Delta^{\mathrm{rel}}_{k,\lambda}\Xi_{\mathrm{Sc}}^{(k)}
\qquad\text{on }\Sym^k(\mathbb C^D).
\end{equation}
Since both operators are supported on $\Sym^k(\mathbb C^D)$, the same inequality holds on $(\mathbb C^D)^{\otimes k}$. By definition,
\begin{align*}
\chi_{\mathrm{Sc}}^{(k)}
&=(\sqrt{D\rho})^{\otimes k}\Xi_{\mathrm{Sc}}^{(k)}(\sqrt{D\rho})^{\otimes k},\\
\chi_{\mathrm{PA}}^{(k)}(\lambda)
&=(\sqrt{D\rho})^{\otimes k}\Xi_{\mathrm{PA}}^{(k)}(\sqrt{D\rho})^{\otimes k}.
\end{align*}
Multiplying each operator on the left and right by $(\sqrt{D\rho})^{\otimes k}$ does not increase the relative error and yields Eq.~\eqref{eq:poly-rel-main}. No full-rank assumption on $\rho$ is required.
\end{proof}

\begin{theorem}
[Relative error of polynomial approximation]
\label{thm:relative-error}
For any integer $k\geq2$ satisfying
$12k(k-1)^2\leq\M$, define
\begin{equation}
    \lambda_*
    :=
    \left\lfloor
    \frac{\M}
    {12k(k-1)^2}
    \right\rfloor
    \geq1.
    \label{eq:poly-rel-lambda-star}
\end{equation}
Then there exists a universal constant
\begin{equation}
    C=
    2^{33/8} e^{7/24}
\end{equation}
such that, for every integer
$1\leq\lambda\leq\lambda_*$,
\begin{equation}
    \Delta_{k,\lambda}^{\mathrm{rel}}
    \leq
    C
    \left(
        \frac{6\lambda k(k-1)^2}
        {e^{1/3}\M}
    \right)^{\lambda/2},
    \label{eq:poly-rel-monotone-bound}
\end{equation}
and
\begin{equation}
    (1-\Delta_{k,\lambda}^{\mathrm{rel}})
    \chi_{\mathrm{Sc}}^{(k)}
    \preceq
    \chi_{\mathrm{PA}}^{(k)}(\lambda)
    \preceq
    (1+\Delta_{k,\lambda}^{\mathrm{rel}})
    \chi_{\mathrm{Sc}}^{(k)}.
\end{equation}
The right-hand side of
Eq.~\eqref{eq:poly-rel-monotone-bound}
is monotonically decreasing with $\lambda$ for
$1\leq\lambda\leq\lambda_*$.
In particular, the smallest upper bound is obtained at
$\lambda=\lambda_*$:
\begin{align}
    \Delta_{k,\lambda_*}^{\mathrm{rel}}
    \leq
    C_{\rel}
    \exp\left[
        -\frac{\M}{24k(k-1)^2}
    \right],
    \label{eq:poly-rel-min-bound}
\end{align}
where $C_{\rel}=\sqrt{e}C$.
\end{theorem}
We note that the constant $C$ in
Theorem~\ref{thm:relative-error} can be quantitatively different from
the constants in Theorem~\ref{thm:additive-error}.
\begin{proof}
The proof begins from Lemma~\ref{lem:relative-error}.
Indeed, $1\leq\lambda\leq\lambda_*$ and $12k(k-1)^2\leq\M$ imply
$2k\lambda\leq\M$ and $2k(k-1)<\M$, so all hypotheses of that lemma hold.
Using the Robbins bound
$n!\leq e\,n^{n+1/2}e^{-n}$
with $n=\lambda k$,
\begin{equation}
\left(
\frac{\Gamma(\lambda k+1)}
{(\M/6)^{\lambda k}}
\right)^{\!1/(2k)}
\leq
e^{1/(2k)}(\lambda k)^{1/(4k)}
\left(
\frac{6\lambda k}{e\M}
\right)^{\lambda/2}.
\label{eq:poly-rel-gamma-robbins}
\end{equation}
Since both exponents in Eq.~\eqref{eq:poly-rel-theta-max} are nonnegative, their sum gives an upper bound on the maximum. Moreover, since $(k-2)\lambda+1-k<(k-2)\lambda$, we obtain
\begin{equation}
\Phi_{k,\lambda}^{\max}
\leq
\exp\left(
\frac{(k-1)(2k(k-1)+1)}{2(\M-2k(k-1))}
+
\frac{k\lambda^2(k-2)^2}{\M}
\right).
\label{eq:poly-rel-theta-upper}
\end{equation}
Substituting Eqs.~\eqref{eq:poly-rel-gamma-robbins} and
\eqref{eq:poly-rel-theta-upper} into
Eq.~\eqref{eq:poly-rel-explicit} gives
\begin{align}
\Delta_{k,\lambda}^{\mathrm{rel}}
&\leq
2^{4}e^{1/(2k)}(\lambda k)^{1/(4k)}
\left(
\frac{6\lambda k(k-1)^2}
{e\M}
\right)^{\lambda/2}
\nonumber\\
&\quad\times
\exp\left(
\frac{(k-1)(k(k-1)-1)}{2\M}
+
\frac{(k-1)(2k(k-1)+1)}{2(\M-2k(k-1))}
+
\frac{k\lambda^2(k-2)^2}{\M}
\right).
\label{eq:poly-rel-robbins-explicit}
\end{align}

For
$1\leq\lambda\leq\lambda_*$,
\begin{equation}
    \lambda
    \leq
    \frac{\M}
    {12k(k-1)^2},
\end{equation}
while the assumption
$12k(k-1)^2\leq\M$ implies
\begin{equation}
\frac{(k-1)(k(k-1)-1)}{2\M}
\leq\frac{1}{24},
\end{equation}
\begin{equation}
\frac{(k-1)(2k(k-1)+1)}{2(\M-2k(k-1))}
\leq\frac{1}{8},
\end{equation}
and
\begin{equation}
\frac{k\lambda^2(k-2)^2}{\M}
\leq\frac{\lambda}{12}.
\end{equation}
Therefore,
\begin{equation}
\exp\left(
\frac{(k-1)(k(k-1)-1)}{2\M}
+
\frac{(k-1)(2k(k-1)+1)}{2(\M-2k(k-1))}
+
\frac{k\lambda^2(k-2)^2}{\M}
\right)
\leq
\exp\left(\frac{\lambda}{12}+\frac16\right).
\label{eq:poly-rel-exp-bound}
\end{equation}

Combining Eqs.~\eqref{eq:poly-rel-robbins-explicit} and
\eqref{eq:poly-rel-exp-bound}, we obtain
\begin{align}
\Delta_{k,\lambda}^{\mathrm{rel}}
&\leq
2^{4}
\exp\left(
\frac{1}{2k}+\frac16
\right)
k^{1/(4k)}
\lambda^{1/(4k)}
\left(
\frac{6\lambda k(k-1)^2}
{e^{5/6}\M}
\right)^{\lambda/2}
\nonumber\\
&=
2^{4}
\exp\left(
\frac{1}{2k}+\frac16
\right)
k^{1/(4k)}
\left(
\lambda^{1/(4k)}e^{-\lambda/(4k)}
\right)
\left(
\frac{6\lambda k(k-1)^2}
{e^{5/6-1/(2k)}\M}
\right)^{\lambda/2}.
\end{align}
Since $x^{1/(4k)}e^{-x/(4k)}$ is monotonically decreasing for
$x\geq1$,
\begin{equation}
    \lambda^{1/(4k)}e^{-\lambda/(4k)}
    \leq
    e^{-1/(4k)}.
\end{equation}
Moreover, for $k\geq2$,
\begin{equation}
    \frac56-\frac{1}{2k}\geq\frac{7}{12}\geq\frac13,
\end{equation}
so $e^{5/6-1/(2k)}$ in the denominator can be replaced by $e^{1/3}$.
Also, because
$(\log k+1)/(4k)$ is decreasing for $k\geq2$,
\begin{equation}
    k^{1/(4k)}e^{1/(4k)}
    \leq
    2^{1/8}e^{1/8}.
\end{equation}
Thus
\begin{equation}
    \Delta_{k,\lambda}^{\mathrm{rel}}
    \leq
    C
    \left(
        \frac{6\lambda k(k-1)^2}
        {e^{1/3}\M}
    \right)^{\lambda/2},
\end{equation}
with
\begin{equation}
    C=
    2^{33/8} e^{7/24}.
\end{equation}

To see that the right-hand side decreases monotonically with $\lambda$,
we consider
\begin{equation}
    \frac{\lambda}{2}
    \log\left(
        \frac{6\lambda k(k-1)^2}
        {e^{1/3}\M}
    \right).
\end{equation}
Its derivative with respect to $\lambda$ is
\begin{equation}
    \frac12\left[
    \log\left(
        \frac{6\lambda k(k-1)^2}
        {e^{1/3}\M}
    \right)+1
    \right].
\end{equation}
It is non-positive for
$1\leq\lambda\leq\lambda_*$,
so the right-hand side of
Eq.~\eqref{eq:poly-rel-monotone-bound}
is monotonically decreasing with $\lambda$ throughout this interval.

Hence the smallest upper bound within this interval is obtained at
$\lambda=\lambda_*$. Since
\begin{equation}
\frac{6\lambda_*k(k-1)^2}{e^{1/3}\M}
\leq\frac{1}{2e^{1/3}}\leq e^{-1},
\end{equation}
then 
\begin{align}
    C
    \left(
        \frac{6\lambda_*k(k-1)^2}
        {e^{1/3}\M}
    \right)^{\lambda_*/2} &\leq C \exp{\left(-\frac{\lambda
    _*}{2}\right)} \leq C \exp{\left(-\frac{\left(\frac{\M}{12k(k-1)^2}-1\right)}{2}\right)} \nonumber \\
    & = \sqrt{e}C \exp{\left(-\frac{\M}{24k(k-1)^2}\right)},
\end{align}
which proves Eq.~\eqref{eq:poly-rel-min-bound}.
\end{proof}

In parallel with Corollary~\ref{coro:minimal-additive-error}, 
we obtain the following result for the relative error.

\begin{corollary}
[Sufficient $\lambda$ for a given relative error]
\label{coro:minimal-relative-error}
Let $k\geq2$ be an integer with $12k(k-1)^2\leq\M$, and let
$C_{\rel}$ be the constant in Theorem~\ref{thm:relative-error}.
If $\epsilon>0$ satisfies
\begin{equation}
\epsilon\geq
C_{\rel}\exp\!\left[-\frac{\M}{24k(k-1)^2}\right],
\end{equation}
it suffices to choose
\begin{equation}
\lambda=\max\left\{
\left\lfloor2\log\left(\frac{C_{\rel}}{\epsilon}\right)\right\rfloor,1
\right\}.
\end{equation}
Then
\begin{equation}
(1-\epsilon)\chi_{\mathrm{Sc}}^{(k)}
\preceq
\chi_{\mathrm{PA}}^{(k)}(\lambda)
\preceq
(1+\epsilon)\chi_{\mathrm{Sc}}^{(k)}.
\end{equation}
\end{corollary}
\begin{proof}
The proof method is the same as Corollary~\ref{coro:minimal-additive-error}. The only difference is replacing $(k-1)^2$ with $k(k-1)^2$.
\end{proof}

\section{Efficient implementation of symmetric subspace projector}\label{app:projection}
Recall our approximation formula,
\begin{align}
\chi_{\mathrm{PA}}^{(k)}(\lambda)
=\sum_{i=1}^{\lambda}\binom{\lambda}{i}(-1)^{i-1} D^{k+a_i}
\mathrm{Tr}_{\mathrm{aux}}
\!\left[
\rho^{\otimes (k+a_i)}
\chi_{\mathrm{Haar}}^{(k+a_i)}
\right].
\end{align}
To predict properties of the Scrooge ensemble, we estimate $\trace{}{\chi_{\mathrm{PA}}^{(k)}(\lambda)O_k}$ as a good approximation for the expectation value $\mathrm{Tr}\!\left[{\chi_{\mathrm{Sc}}^{(k)}O_k}\right]$, for any $k$-copy operator $O_k$. 
To achieve this goal we evaluate the $\lambda$ terms separately. Namely, we measure $\mathrm{Tr}
\!\left[
\rho^{\otimes k_i}
\chi_{\mathrm{Haar}}^{(k_i)}\left(O_k\otimes I^{\otimes a_i}\right)
\right]$ and combine the resulting estimates with the corresponding coefficients. However, operationally, the expectation value of $O_k\otimes I^{\otimes a_i}$ must be evaluated with respect to a normalized density operator. Let 
\begin{equation}
\rho_{\mathrm{sym}}^{(k_i)}=\rho^{\otimes k_i}\Pi_{\mathrm{sym}}^{(k_i)}/p_{\mathrm{sym}}
^{(k_i)},
\end{equation}
where $p_{\mathrm{sym}}^{(k_i)}=\mathrm{Tr}\!\left[\rho^{\otimes k_i}\Pi_{\mathrm{sym}}^{(k_i)}\right]$ normalizes the density matrix.
Using the facts that $\rho^{\otimes k_i}$ commutes with $\Pi_{\mathrm{sym}}^{(k_i)}$ and that $\Pi_{\mathrm{sym}}^{(k_i)}\Pi_{\mathrm{sym}}^{(k_i)}=\Pi_{\mathrm{sym}}^{(k_i)}$, we can write $\rho_{\mathrm{sym}}^{(k_i)}$ as $\Pi_{\mathrm{sym}}^{(k_i)}\rho^{\otimes  k_i}\Pi_{\mathrm{sym}}^{(k_i)}/p_{\mathrm{sym}}
^{(k_i)}$, which is the normalized projection of $\rho^{\otimes k_i}$ onto the totally symmetric subspace. Meanwhile, $p_{\mathrm{sym}}^{(k_i)}=\mathrm{Tr}\!\left[\Pi_{\mathrm{sym}}^{(k_i)}\rho^{\otimes k_i}\Pi_{\mathrm{sym}}^{(k_i)}\right]$ is the probability of projecting onto the totally symmetric subspace. In this section, we focus on implementing this quantum state using many copies of $\rho$.

\subsection{Schur transform}\label{sec:schur-transform}

The central tool we use is the Schur transform.
Let $\mathbb{C}^D$ be the Hilbert space of a single $D$-dimensional
qudit. Schur transforms deal with the decomposition $(\mathbb{C}^D)^{\otimes k}$ in terms of irreducible representations of the symmetric group $\mathcal{S}_k$ and the unitary group $\mathcal{U}_D$
\begin{align}
(\mathbb{C}^D)^{\otimes k}
\cong
\bigoplus_{\gamma \vdash k,\ \ell(\gamma )\le D}
\mathcal Q_\gamma ^D\otimes \mathcal P_\gamma .
\end{align}
Here, $\gamma \vdash k$ means that $\gamma $ is a partition of $k$,
and $\ell(\gamma )$ denotes the number of nonzero parts of $\gamma $.
Every partition $\gamma$ corresponds to a \emph{Young diagram} of shape $\gamma$, see Fig.~\ref{fig:tableaux} for an example of $\gamma=(4,1)$.
The space $\mathcal Q_\gamma ^D$ carries the irreducible representation
of the unitary group $\mathcal{U}_D$ labeled by $\gamma $, while
$\mathcal P_\gamma $ carries the irreducible representation of the
symmetric group $\mathcal{S}_k$ labeled by the same partition. 

The Schur transform $U_{\mathrm{Sch}}$ maps the standard computational basis of
$(\mathbb{C}^D)^{\otimes k}$ to the Schur basis
$\ket{\gamma ,q_\gamma ,p_\gamma }$, where
$q_\gamma $ labels a basis vector of $\mathcal Q_\gamma ^D$, and
$p_\gamma $ labels a basis vector of $\mathcal P_\gamma $.

\begin{lemma}[Runtime of the Schur transform \cite{harrow2005applications}]
\label{lem:time-schur}
For any $0<\epsilon_{\mathrm{Sch}}<1$, the Schur transform on
$(\mathbb{C}^{D})^{\otimes k}$ can be approximated to error
$\epsilon_{\mathrm{Sch}}$ by a quantum circuit with runtime
\[
\operatorname{poly}\!\left(
k,\,
\log D,\,
\log\frac{1}{\epsilon_{\mathrm{Sch}}}
\right).
\]
Here, the approximation error $\epsilon_{\mathrm{Sch}}$ is defined by
\begin{equation}
\left\lVert
\widetilde{U}_{\mathrm{Sch}}
-
U_{\mathrm{Sch}}
\right\rVert_\infty
\leq \epsilon_{\mathrm{Sch}},
\end{equation}
where $\widetilde{U}_{\mathrm{Sch}}$ and $U_{\mathrm{Sch}}$ denote the implemented approximate and ideal Schur transforms, respectively.
\end{lemma}

\begin{figure}[t!]
    \centering
    \tableau{1,2,3,4}{5}{palegreen}{3}{white}%
    \hspace{5mm}%
    \tableau{1,2,3,5}{4}{palegreen}{3}{white}%
    \hspace{5mm}%
    \tableau{1,2,4,5}{3}{palepink}{2}{palepink}%
    \hspace{5mm}%
    \tableau{1,3,4,5}{2}{palepink}{2}{palepink}
    \caption{Four tableaux of shape $(4,1)$.}
    \label{fig:tableaux}
\end{figure}
Crucially, the states $\ket{p_\gamma}$ form the
\emph{Young--Yamanouchi basis}, whose elements are indexed by
\emph{standard Young tableaux} $p_\gamma$ of shape $\gamma$.
A standard Young tableau is a filling of the boxes of $\gamma$
with the numbers $1,\ldots,k$, each appearing exactly once,
such that the entries increase from left to right along each row
and from top to bottom along each column.
Figure~\ref{fig:tableaux} shows the four standard Young tableaux
of shape $(4,1)$.

In any standard Young tableau, the boxes containing
$1,\ldots,k'$ form a Young diagram for every $k'\leq k$.
The shape of this smaller diagram identifies the irreducible
representation of $\calS_{k'}$ to which the corresponding basis
state belongs, where $\calS_{k'}\subseteq \calS_k$ is the subgroup
that permutes only the first $k'$ registers.
For example, Fig.~\ref{fig:tableaux} illustrates the restriction
from $\calS_5$ to $\calS_3$, with the boxes containing $1,2,3$ shaded.
Each of the two green tableaux has restricted shape $(3)$,
so each corresponding basis state spans a copy of the trivial
representation of $\calS_3$.
The two red tableaux have restricted shape $(2,1)$, and their
corresponding basis states together span the two-dimensional
irreducible representation of $\calS_3$ labeled by $(2,1)$.

Consequently, the projector onto the subspace symmetric under
permutations of the first $k'$ registers,
\begin{align}
    \Pi_{\mathrm{sym}}^{k\rightarrow k'}
=
\Pi_{\mathrm{sym}}^{(k')}\otimes I^{\otimes(k-k')},
\end{align}
takes the following form in the Schur basis:
\begin{align}\label{eq:projector}
\Pi_{\mathrm{sym}}^{(k')}\otimes I^{\otimes(k-k')}
\cong
\bigoplus_{\substack{\gamma\vdash k, \ell(\gamma)\leq D}}
I_\gamma\otimes P_\gamma^{(k')}.
\end{align}
Here, $P_\gamma^{(k')}$ is the orthogonal projector onto the
span of the Young--Yamanouchi basis states whose tableaux,
when restricted to the entries $1,\ldots,k'$, have shape $(k')$.
Equivalently, these are precisely the tableaux in which
$1,\ldots,k'$ all lie in the first row. In words, we can efficiently construct $P_\gamma^{(k')}$ for the symmetric subspace projector $\Pi_{\mathrm{sym}}^{k\rightarrow k'}$ with Young--Yamanouchi basis states, which is mapped to the standard computational basis after Schur transform. We will use this fact to implement the symmetric subspace projector below.

\subsection{Implementing the symmetric subspace projection at low cost}
\label{Implementing the symmetric subspace projection at low cost}

To prepare
$\rho_{\mathrm{sym}}^{(k_i)}
=\rho^{\otimes k_i}\Pi_{\mathrm{sym}}^{(k_i)}/p_{\mathrm{sym}}^{(k_i)}$
from copies of $\rho$, a natural approach is to take $k_i$ copies
and measure the symmetric subspace projector
$\Pi_{\mathrm{sym}}^{(k_i)}$ using the Schur transform.
Conditioned on obtaining the symmetric outcome, the resulting
state is precisely $\rho_{\mathrm{sym}}^{(k_i)}$.
However, this outcome occurs with probability
\begin{align}
p_{\mathrm{sym}}^{(k_i)}
=\mathrm{Tr}\!\left[
\rho^{\otimes k_i}\Pi_{\mathrm{sym}}^{(k_i)}
\right].
\end{align}
For example, when $\rho=I/D$ and $k_i=k$, this probability is
\begin{align}
p_{\mathrm{sym}}^{(k_i)}
=\frac{\binom{D+k-1}{k}}{D^k}
=\frac{1}{k!}\prod_{j=0}^{k-1}\left(1+\frac{j}{D}\right)
=\frac{1+\mathcal{O}(1/D)}{k!},
\end{align}
Thus, for a highly mixed state in this regime, the expected
number of attempts is approximately $k!$, making direct
post-selection costly as $k$ grows.

To avoid this post-selection cost, we exploit the structure
of the target state in the Schur basis.
Within each representation sector, the required state of the
permutation register is known and can be prepared directly,
without post-selection.
The remaining challenge is to reproduce the relative weights
of the representation sectors.
Using $m$ copies allows us to study and control the approximation
error through these weights.

Suppose we have $\rho^{\otimes m}$, with $m\geq k$, and consider
the target state
\begin{align}\label{eq:target}
\frac{\Pi_{\mathrm{sym}}^{m\rightarrow k}\rho^{\otimes m}}
{\trace{}{\Pi_{\mathrm{sym}}^{m\rightarrow k}\rho^{\otimes m}}}
=
\rho_{\mathrm{sym}}^{(k)}
\otimes \rho^{\otimes (m-k)}.
\end{align}
Tracing out the last $m-k$ registers therefore yields the
desired state $\rho_{\mathrm{sym}}^{(k)}$.

The Schur decomposition of the input state is
\begin{align}\label{eq:density-matrix}
\rho^{\otimes m}
\cong
\bigoplus_{\substack{\gamma\vdash m,\ \ell(\gamma)\leq D}}
\mu_\gamma(\rho)Q_\gamma(\rho)
\otimes \frac{I_\gamma}{r_\gamma},
\end{align}
where $Q_\gamma(\rho)$ is a normalized density matrix,
the non-negative weights $\mu_\gamma(\rho)$ sum to one,
and $r_\gamma$ is the dimension of $\mathcal{P}_\gamma$.
In particular, in every subspace labeled by  $\gamma$,
the permutation register is maximally mixed and independent
of the register containing $Q_\gamma(\rho)$.

Using Eqs.~\eqref{eq:projector} and \eqref{eq:target},
we can express the target state as
\begin{align}
\rho_{\mathrm{sym}}^{(k)}\otimes\rho^{\otimes(m-k)}
&\cong
\frac{1}{
\sum_\gamma
\frac{\mu_\gamma(\rho)\trace{}{P_\gamma^{(k)}}}{r_\gamma}
}
\bigoplus_{\substack{\gamma\vdash m,\ \ell(\gamma)\leq D}}
\mu_\gamma(\rho)Q_\gamma(\rho)
\otimes\frac{P_\gamma^{(k)}}{r_\gamma}
\notag\\
&=
\bigoplus_{\substack{\gamma\vdash m,\ \ell(\gamma)\leq D}}
\frac{\mu_\gamma(\rho)a_\gamma^{(k)}}
{\sum_\gamma\mu_\gamma(\rho)a_\gamma^{(k)}}
Q_\gamma(\rho)\otimes
\frac{P_\gamma^{(k)}}{\trace{}{P_\gamma^{(k)}}},
\end{align}
where
\begin{align}
    a_\gamma^{(k)}
=\frac{\trace{}{P_\gamma^{(k)}}}{r_\gamma}
\end{align}
is the probability of obtaining the symmetric subspace
conditioned on the sector $\gamma$.
Thus, post-selection has two effects: it replaces the
maximally mixed permutation state by
$P_\gamma^{(k)}/\trace{}{P_\gamma^{(k)}}$, and it reweights
the sectors by $a_\gamma^{(k)}$.
Whenever $\trace{}{P_\gamma^{(k)}}=0$, we interpret the normalized
factor $P_\gamma^{(k)}/\trace{}{P_\gamma^{(k)}}$ as an arbitrary
fixed density matrix on $\mathcal{P}_\gamma$ and set $a_\gamma^{(k)}=0$.
Such sectors have zero overlap in the target state.

Since $P_\gamma^{(k)}/\trace{}{P_\gamma^{(k)}}$ depends only on
$\gamma$ and $k$, it can be prepared without knowing $\rho$.
We therefore define a quantum channel that acts in the Schur basis and
performs this replacement:
\begin{align}
\calN_{m,k}\!\left(
U_{\mathrm{Sch}}\rho^{\otimes m}U_{\mathrm{Sch}}^\dagger
\right)
:=
\bigoplus_{\substack{\gamma\vdash m,\ \ell(\gamma)\leq D}}
\mu_\gamma(\rho)Q_\gamma(\rho)\otimes
\frac{P_\gamma^{(k)}}{\trace{}{P_\gamma^{(k)}}}.
\end{align}
The channel $\calN_{m,k}$ acts directly in the Schur basis:
it measures the Schur label $\gamma$ and replaces the permutation
register with the corresponding prescribed state, while leaving
$Q_\gamma(\rho)$ unchanged.

Its approximation error therefore comes entirely from the
difference between the original weights $\mu_\gamma(\rho)$
and the post-selected weights
\[
\frac{\mu_\gamma(\rho)a_\gamma^{(k)}}
{\sum_\gamma\mu_\gamma(\rho)a_\gamma^{(k)}}.
\]
The following lemma bounds this error by the relative
fluctuations of $a_\gamma^{(k)}$ under $\mu_\gamma(\rho)$.
\begin{lemma}
\label{lem:channel-error}
    \begin{align}
        \norm{U_{\mathrm{Sch}}^\dagger\calN_{m,k}\left(
U_{\mathrm{Sch}}\rho^{\otimes m}U_{\mathrm{Sch}}^\dagger\right)U_{\mathrm{Sch}}-\rho_{\mathrm{sym}}^{(k)}\otimes \rho^{\otimes (m-k)}}_1\leq 
        \sqrt{\sum_\gamma \mu_\gamma(\rho)\left(\frac{a_\gamma^{(k)}}{\sum_\gamma \mu_\gamma(\rho)a_\gamma^{(k)}}-1\right)^2}.
    \end{align}
\end{lemma}
\begin{proof}
    Since $Q_{\gamma}Q_{\gamma'}=0$ when $\gamma\neq \gamma'$, we have
    \begin{align}
        \norm{U_{\mathrm{Sch}}^\dagger\calN_{m,k}\left(
U_{\mathrm{Sch}}\rho^{\otimes m}U_{\mathrm{Sch}}^\dagger\right)U_{\mathrm{Sch}}-\rho_{\mathrm{sym}}^{(k)}\otimes \rho^{\otimes (m-k)}}_1
        &=\sum_\gamma \left|\mu_\gamma(\rho)-\frac{\mu_\gamma(\rho) a_\gamma^{(k)}}{\sum_\gamma \mu_\gamma(\rho) a_\gamma^{(k)}}\right|\norm{Q_\gamma(\rho)\otimes \frac{P^{(k)}_\gamma}{\trace{}{P^{(k)}_\gamma}}}_1
        \notag\\
        &=\sum_\gamma\mu_\gamma(\rho) \left|1-\frac{ a_\gamma^{(k)}}{\sum_\gamma \mu_\gamma(\rho) a_\gamma^{(k)}}\right|
        \notag\\
        &\leq 
        \sqrt{\sum_\gamma \mu_\gamma(\rho)}\cdot
        \sqrt{\sum_\gamma \mu_\gamma(\rho)\left(\frac{a_\gamma^{(k)}}{\sum_\gamma \mu_\gamma(\rho)a_\gamma^{(k)}}-1\right)^2}
        \notag\\
        &=\sqrt{\sum_\gamma \mu_\gamma(\rho)\left(\frac{a_\gamma^{(k)}}{\sum_\gamma \mu_\gamma(\rho)a_\gamma^{(k)}}-1\right)^2}.
    \end{align}
    In the third line we use the Cauchy--Schwarz inequality.
\end{proof}

The error has a state-independent upper bound.
\begin{lemma}[State-independent error bound]
\label{lem:relative-concentration}
For every density operator $\rho$ on $\mathbb{C}^D$ and
integers $m\geq k\geq1$,
\begin{align}
\sum_\gamma \mu_\gamma(\rho)
\left(
\frac{a_\gamma^{(k)}}
{\sum_\gamma\mu_\gamma(\rho)a_\gamma^{(k)}}-1
\right)^2
\leq
\exp\!\left(\frac{k^2(k-1)}{m}\right)-1.
\label{eq:relative-concentration-bound}
\end{align}
\end{lemma}

\begin{proof}
Define
\begin{align}
p_{\mathrm{sym}}^{(t)}:=\trace{}{\rho^{\otimes t}\Pi_{\mathrm{sym}}^{(t)}},
\qquad p_{\mathrm{sym}}^{(0)}:=1.
\end{align}
If $\beta_1,\ldots,\beta_D$ are the eigenvalues of $\rho$, then
\begin{align}
p_{\mathrm{sym}}^{(t)}
&=\sum_{\substack{n_1+\cdots+n_D=t\\n_i\geq0}}
\prod_{i=1}^D \beta_i^{n_i},
&
p_{\mathrm{sym}}^{(t)}=p_{\mathrm{sym}}^{(t)}\sum_{i=1}^D \beta_i
&\leq(t+1)p_{\mathrm{sym}}^{(t+1)}.
\end{align}
The inequality holds because each monomial of degree $t+1$
appears at most $t+1$ times in the product.
In particular, $p_{\mathrm{sym}}^{(k)}\geq1/k!>0$, and
\begin{align}
\frac{p_{\mathrm{sym}}^{(k-j)}}{p_{\mathrm{sym}}^{(k)}}\leq(k)_j\leq k^j,
\qquad 0\leq j\leq k,
\end{align}
where $(x)_j=x(x-1)\cdots(x-j+1)$ and $(x)_0=1$.

For a $k$-element subset $S\subseteq[m]$, let
$\Pi_{\mathrm{sym}}^{(S)}$ be the symmetric subspace projector onto the symmetric subspace of the registers in $S$
and act as the identity elsewhere.
We first show that, for $|S|=|T|=k$ and $j=|S\cap T|$,
\begin{align}
0\leq
\trace{}{
\rho^{\otimes m}
\Pi_{\mathrm{sym}}^{(S)}
\Pi_{\mathrm{sym}}^{(T)}
}
\leq \left(p_{\mathrm{sym}}^{(k)}\right)^2 k^j.
\label{eq:overlap-estimate}
\end{align}
For $j=0$, the trace equals $\left(p_{\mathrm{sym}}^{(k)}\right)^2$.
For $j\geq1$, define an operator on the shared registers by
\begin{align}
B_j:=
\operatorname{Tr}_{k-j}\!\left[
\left(I\otimes(\rho^{1/2})^{\otimes(k-j)}\right)
\Pi_{\mathrm{sym}}^{(k)}
\left(I\otimes(\rho^{1/2})^{\otimes(k-j)}\right)
\right].
\end{align}
Tracing out the nonshared registers gives
\begin{align}
\trace{}{
\rho^{\otimes m}
\Pi_{\mathrm{sym}}^{(S)}
\Pi_{\mathrm{sym}}^{(T)}
}
&=\trace{}{\rho^{\otimes j}B_j^2},
&
\trace{}{\rho^{\otimes j}B_j}&=p_{\mathrm{sym}}^{(k)}.
\end{align}
Since
$\Pi_{\mathrm{sym}}^{(k)}
\leq I\otimes\Pi_{\mathrm{sym}}^{(k-j)}$,
we have $0\leq B_j\leq p_{\mathrm{sym}}^{(k-j)}I$ and hence
\begin{align}
0\leq\trace{}{\rho^{\otimes j}B_j^2}
\leq p_{\mathrm{sym}}^{(k-j)}\trace{}{\rho^{\otimes j}B_j}
=p_{\mathrm{sym}}^{(k-j)}p_{\mathrm{sym}}^{(k)}
\leq \left(p_{\mathrm{sym}}^{(k)}\right)^2 k^j.
\end{align}
This proves Eq.~\eqref{eq:overlap-estimate}.

Now consider the averaged projector
\begin{align}
A_{m,k}:=
\binom{m}{k}^{-1}
\sum_{\substack{S\subseteq[m]\\|S|=k}}
\Pi_{\mathrm{sym}}^{(S)}.
\end{align}
It commutes with the entire permutation group,
so it acts as a scalar on each $\mathcal{P}_\gamma$ by Schur's lemma.
Every summand is conjugate to
$\Pi_{\mathrm{sym}}^{m\rightarrow k}$; taking the trace on
$\mathcal{P}_\gamma$ therefore identifies this scalar as
$a_\gamma^{(k)}$.
Consequently,
\begin{align}
\trace{}{\rho^{\otimes m}A_{m,k}}
&=\sum_\gamma\mu_\gamma(\rho)a_\gamma^{(k)}=p_{\mathrm{sym}}^{(k)},
\\
\trace{}{\rho^{\otimes m}A_{m,k}^2}
&=\sum_\gamma\mu_\gamma(\rho)
\left(a_\gamma^{(k)}\right)^2.
\end{align}

Let $S,T$ be independent  $k$-element subsets of $[m]$,
and let $L=|S\cap T|$.
Expanding $A_{m,k}^2$ and using
Eq.~\eqref{eq:overlap-estimate}, we obtain
\begin{align}
&\sum_\gamma\mu_\gamma(\rho)
\left(
\frac{a_\gamma^{(k)}}
{\sum_\gamma\mu_\gamma(\rho)a_\gamma^{(k)}}-1
\right)^2
\notag\\
&\qquad=
\frac{\trace{}{\rho^{\otimes m}A_{m,k}^2}}{\left(p_{\mathrm{sym}}^{(k)}\right)^2}-1
\leq\mathbb{E}[k^L]-1.
\end{align}
Conditioned on $S$, each fixed $j$-element subset of $S$
is contained in $T$ with probability $(k)_j/(m)_j$.
Thus,
\begin{align}
\mathbb{E}\!\left[\binom{L}{j}\right]
=\binom{k}{j}\frac{(k)_j}{(m)_j}.
\end{align}
Expanding $k^L=(1+(k-1))^L$ and using
$(k)_j/(m)_j\leq(k/m)^j$ yields
\begin{align}
\mathbb{E}[k^L]
&=\sum_{j=0}^k
(k-1)^j\binom{k}{j}\frac{(k)_j}{(m)_j}
\notag\\
&\leq
\sum_{j=0}^k\binom{k}{j}
\left(\frac{k(k-1)}{m}\right)^j
\notag\\
&=
\left(1+\frac{k(k-1)}{m}\right)^k
\leq
\exp\!\left(\frac{k^2(k-1)}{m}\right).
\end{align}
Combining the last two bounds proves the claim.
\end{proof}

With all the above results, we claim that there is an efficient construction of the symmetric subspace projection by consuming many copies of $\rho$.
\begin{theorem}
    [Efficient implementation of symmetric-subspace projection]
    \label{thm:efficient-projection}
    For $n\in\bbZ^+$, $1\leq k<m$,
    and $0<\epsilon<1$, there exists a quantum channel $\mathcal M_{m,k}:
\mathcal L\!\left((\mathbb C^{2^n})^{\otimes m}\right)
\longrightarrow
\mathcal L\!\left((\mathbb C^{2^n})^{\otimes k}\right)$ such that for every density operator $\rho$ on $\calH=\mathbb{C}^{2^n}$,
    \begin{align}
        \norm{\calM_{m,k}(\rho^{\otimes m})
        -\rho_{\mathrm{sym}}^{(k)}
        }_1\leq \epsilon
    \end{align}
    holds
    as long as 
    \begin{align}
        m\geq\frac{5k^2(k-1)}{\epsilon^2}.
    \end{align}
     $\calM_{m,k}$ can be implemented with $\mathrm{poly}\left(m,n,\log\frac{1}{\epsilon}\right)$ number of gates, computational-basis measurement and feedback, $\calO(mn)$ bits of classical randomness and efficient classical computing. 
\end{theorem}
\begin{proof}
    $\calM_{m,k}$ is ideally implemented by applying the Schur transform $U_{\mathrm{Sch}}$ on $m$ copies of $\rho$, followed by $\calN_{m,k}$ and $U^\dagger_{\mathrm{Sch}}$. Schur transform can be implemented up to $\epsilon_{\mathrm{Sch}}$ error in operator norm. Let $\widetilde{U}_{\mathrm{Sch}}$ denote the approximate Schur transform, then we implement
    \begin{align}
        \calM_{m,k}(\rho^{\otimes m})=\trace{\mathrm{aux}}{\widetilde{U}_{\mathrm{Sch}}^\dagger \calN_{m,k}(\widetilde{U}_{\mathrm{Sch}}\rho^{\otimes m} \widetilde{U}_{\mathrm{Sch}}^\dagger)
        \widetilde{U}_{\mathrm{Sch}}},
    \end{align}
    where $\trace{\mathrm{aux}}{\cdot}$ traces out all the qubits except for the first $k$ registers.

    Using Lemma~\ref{lem:channel-error} and \ref{lem:relative-concentration}, we have
    \begin{align}
        \norm{\trace{\mathrm{aux}}{U_{\mathrm{Sch}}^\dagger \calN_{m,k}(U_{\mathrm{Sch}}\rho^{\otimes m} U_{\mathrm{Sch}}^\dagger)
        U_{\mathrm{Sch}}}-\rho_{\mathrm{sym}}^{(k)}}_1\leq \sqrt{\exp\left(\frac{k^2(k-1)}{m}\right)-1}.
    \end{align}
    For any density matrix $\sigma$,
    \begin{align}
        \norm{\widetilde{U}\sigma\widetilde{U}^\dagger-U\sigma U^\dagger
        }_1\leq 2\norm{U-\widetilde{U}}_\infty.
    \end{align}
    Using this inequality twice, we have
    \begin{align}
        \norm{\calM_{m,k}(\rho^{\otimes m})-\rho_{\mathrm{sym}}^{(k)}}_1=
        \norm{\trace{\mathrm{aux}}{\widetilde{U}_{\mathrm{Sch}}^\dagger \calN_{m,k}(\widetilde{U}_{\mathrm{Sch}}\rho^{\otimes m} \widetilde{U}_{\mathrm{Sch}}^\dagger)
        \widetilde{U}_{\mathrm{Sch}}}-\rho_{\mathrm{sym}}^{(k)}}_1\leq \sqrt{\exp\left(\frac{k^2(k-1)}{m}\right)-1}+4\epsilon_{\mathrm{Sch}}.
    \end{align}
    In order to achieve the prescribed error $\epsilon$, we set $4\epsilon_{\mathrm{Sch}}=\epsilon/2$. We also demand $\sqrt{\exp\left(k^2(k-1)/m\right)-1}\leq \epsilon/2$, which gives
    \begin{align}
        m\geq \frac{k^2(k-1)}{\log(1+\epsilon^2/4)},
    \end{align}
    So it's sufficient to choose
    \begin{align}
    m\geq\frac{5k^2(k-1)}{\epsilon^2}.
    \end{align}
    Here we use $\log(1+x)>x/(1+x)$.
    
    According to Lemma~\ref{lem:time-schur}, $\widetilde{U}_{\mathrm{Sch}}$ can be implemented using $\mathrm{poly}(m,n,\log(1/\epsilon))$ gates. $\calN_{m,k}$ can be implemented using measurement-feedback and classical randomness. The classical randomness is introduced to prepare density matrix $P^{(k)}_\gamma/\trace{}{P^{(k)}_\gamma}$. Since $\trace{}{P^{(k)}_\gamma}$ cannot exceed $2^{mn}$, one needs at most $mn$ bits of classical randomness. The preparation of $P^{(k)}_\gamma/\trace{}{P^{(k)}_\gamma}$ is efficient because Young-Yamanouchi basis is mapped to computational basis via Schur transform.
\end{proof}

Since $\calM_{m,k}$ can be implemented using an efficient quantum circuit and classical computing, it can always be purified with at most polynomial overhead.
\begin{corollary}
\label{purification-of-channel}
    For $n\in\bbZ^+$, $1\leq k<m$,
    and $0<\epsilon<1$, there exists a quantum circuit $U_{m,k}$, such that for every density operator $\rho$ on $\calH=\mathbb{C}^{2^n}$,
\begin{equation}
\left\|
\Tr_{\mathrm{aux}}\!\left[
U_{m,k}
\left(
\rho^{\otimes m}\otimes\ketbra{0^s}{0^s}
\right)
U_{m,k}^\dagger
\right]
-
\rho_{\mathrm{sym}}^{(k)}
\right\|_1
\leq
\epsilon.
\end{equation}
    Here $\trace{\mathrm{aux}}{\cdot}$ traces out all the qubits except for the first $k$ registers. $s$ is the number of ancilla qubits.
     $U_{m,k}$ can be implemented with $\mathrm{poly}\left(m,n,\log\frac{1}{\epsilon}\right)$ number of gates with $s=\calO(\mathrm{poly}\left(m,n\right))$.
\end{corollary}\noindent
Note that since we only need to purify $\calN_{m,k}$, the number of ancilla qubits does not depend on $\epsilon$.

\subsection{Corollary: precise estimation of $p_{\mathrm{sym}}$}
\label{Corollary: precise estimation of p}

As a useful by-product, the method introduced in the previous section
enables an ultra-precise estimation of the post-selection probability
$p_{\mathrm{sym}}$ by using median-of-means estimation, which will be
useful later.

\begin{lemma}[Relative-error estimation of the post-selection probability]
\label{lem:estimate-symmetric-probability}
Let $\rho$ be a density operator on $\mathbb{C}^{2^n}$,
$k\geq2$, $0<\varepsilon<1$, and $0<\delta<1$.
There exists an estimator $\widehat p_{\mathrm{sym}}^{(k)}$,
obtained from independent Schur-label measurements, such that
\begin{align}
\Pr\!\left[
\left|
\widehat p_{\mathrm{sym}}^{(k)}-p_{\mathrm{sym}}^{(k)}
\right|
>
\varepsilon p_{\mathrm{sym}}^{(k)}
\right]
\leq\delta.
\end{align}
The estimator uses
\begin{align}
N
&=
k^2(k-1)
\left\lceil
\frac{6(e-1)}{\varepsilon^2}
\right\rceil
\left[
2\left\lceil
\frac94\ln\frac1\delta-\frac12
\right\rceil+1
\right]
\end{align}
copies of $\rho$.
\end{lemma}

\begin{proof}
Recall that
\begin{align}
p_{\mathrm{sym}}^{(k)}
:=
\trace{}{
\rho^{\otimes k}\Pi_{\mathrm{sym}}^{(k)}
}
\geq\frac{1}{k!}>0.
\end{align}
Choose
\begin{align}
m:=k^2(k-1).
\end{align}
A Schur-label measurement on $\rho^{\otimes m}$ produces
a partition $\gamma$ with probability $\mu_\gamma(\rho)$.
For each outcome, define
\begin{align}
X
:=
a_\gamma^{(k)}
=
\frac{\trace{}{P_\gamma^{(k)}}}{r_\gamma}.
\end{align}
The expectation of $X$ over the Schur-label measurement outcomes is
precisely the probability of projecting onto the subspace symmetric
under permutations of the first $k$ registers:
\begin{align}
\mathbb{E}[X]
&=
\sum_\gamma
\mu_\gamma(\rho)a_\gamma^{(k)}
\notag\\
&=
\trace{}{
\rho^{\otimes m}
\Pi_{\mathrm{sym}}^{m\rightarrow k}
}
=
p_{\mathrm{sym}}^{(k)},
\end{align}
where the sum runs over
$\gamma\vdash m$ with $\ell(\gamma)\leq 2^n$.

By Lemma~\ref{lem:relative-concentration} and the choice of $m$,
\begin{align}
\frac{\operatorname{Var}(X)}
{\left(p_{\mathrm{sym}}^{(k)}\right)^2}
&=
\frac{\sum_\gamma
\mu_\gamma(\rho)\left(a_{\gamma}^{(k)}-p_{\mathrm{sym}}^{(k)}\right)^2}
{\left(p_{\mathrm{sym}}^{(k)}\right)^2}
\notag\\
&=
\sum_\gamma
\mu_\gamma(\rho)
\left(
\frac{a_\gamma^{(k)}}{p_{\mathrm{sym}}^{(k)}}
-1
\right)^2
\notag\\
&\leq
\exp\!\left(
\frac{k^2(k-1)}{m}
\right)-1
=
e-1.
\end{align}

Set
\begin{align}
L
&:=
\left\lceil
\frac{6(e-1)}{\varepsilon^2}
\right\rceil,
\\
R
&:=
2\left\lceil
\frac94\ln\frac1\delta-\frac12
\right\rceil+1.
\end{align}
Then $R$ is odd and satisfies
\begin{align}
R\geq\frac92\ln\frac1\delta.
\end{align}

Perform $LR$ independent Schur-label measurements,
each on a fresh batch of $m$ copies.
Denote their outcomes by $\gamma_{r,\ell}$,
where $1\leq r\leq R$ and $1\leq\ell\leq L$.
Define the group means and their median by
\begin{align}
Y_r
&:=
\frac1L
\sum_{\ell=1}^L
a_{\gamma_{r,\ell}}^{(k)},
\\
\widehat p_{\mathrm{sym}}^{(k)}
&:=
\operatorname{median}(Y_1,\ldots,Y_R).
\end{align}
Since the samples within each group are independent,
\begin{align}
\mathbb{E}[Y_r]
&=
p_{\mathrm{sym}}^{(k)},
\\
\operatorname{Var}(Y_r)
&=
\frac{\operatorname{Var}(X)}{L}
\leq
\frac{e-1}{L}
\left(p_{\mathrm{sym}}^{(k)}\right)^2.
\end{align}
Chebyshev's inequality therefore gives
\begin{align}
\Pr\!\left[
\left|
Y_r-p_{\mathrm{sym}}^{(k)}
\right|
>
\varepsilon p_{\mathrm{sym}}^{(k)}
\right]
\leq
\frac{e-1}{L\varepsilon^2}
\leq
\frac16.
\end{align}

Let $Z_r$ be the indicator that the $r$th group mean lies
outside the interval
$\left[
(1-\varepsilon)p_{\mathrm{sym}}^{(k)},
(1+\varepsilon)p_{\mathrm{sym}}^{(k)}
\right]$.
Then $Z_1,\ldots,Z_R$ are independent,
take values in $\{0,1\}$, and satisfy
\begin{align}
\mathbb{E}[Z_r]\leq\frac16.
\end{align}
Since $R$ is odd, the median lies outside this interval only if
at least $(R+1)/2$ group means lie outside it. Hence
\begin{align}
&\Pr\!\left[
\left|
\widehat p_{\mathrm{sym}}^{(k)}-p_{\mathrm{sym}}^{(k)}
\right|
>
\varepsilon p_{\mathrm{sym}}^{(k)}
\right]
\notag\\
&\qquad\leq
\Pr\!\left[
\sum_{r=1}^R Z_r
\geq
\frac{R+1}{2}
\right]
\notag\\
&\qquad\leq
\Pr\!\left[
\sum_{r=1}^R
\left(
Z_r-\mathbb{E}[Z_r]
\right)
\geq
\frac{R}{3}
\right]
\notag\\
&\qquad\leq
\exp\!\left(-\frac{2R}{9}\right)
\leq
\delta,
\end{align}
where the last two inequalities follow from Hoeffding's inequality
and the choice of $R$.

Finally, each Schur-label sample consumes $m$ copies, so the total
number of copies is
\begin{align}
N
&=
mLR
\notag\\
&=
k^2(k-1)
\left\lceil
\frac{6(e-1)}{\varepsilon^2}
\right\rceil
\left[
2\left\lceil
\frac94\ln\frac1\delta-\frac12
\right\rceil+1
\right].
\end{align}
\end{proof}

In practice, however, only an approximate Schur transform
$\widetilde U_{\mathrm{Sch}}^{(m)}$ is implemented.
We now show that the resulting error remains controlled.

Define the completely positive trace-nonincreasing map
\begin{equation}
\mathcal{R}_{m,k}(X)
:=
\bigoplus_{\substack{\gamma\vdash m\\ \ell(\gamma)\leq D}}
a_\gamma^{(k)}
\left[
\Tr_{\mathcal{P}_\gamma}
\left(
\Pi_\gamma X\Pi_\gamma
\right)
\otimes
\frac{P_\gamma^{(k)}}
{\trace{}{P_\gamma^{(k)}}}
\right],
\label{eq:Rmk}
\end{equation}
where the term corresponding to
$\trace{}{P_\gamma^{(k)}}=0$ is defined to be zero.
The map $\mathcal{R}_{m,k}$ differs from $\mathcal{N}_{m,k}$
only by the additional sector-dependent factor $a_\gamma^{(k)}$,
which accounts for the estimator weight.

Let
\begin{align}
\rho_{\mathrm{Sch}}
&:=
U_{\mathrm{Sch}}^{(m)}
\rho^{\otimes m}
U_{\mathrm{Sch}}^{(m)\dagger},
\\
\widetilde{\rho}_{\mathrm{Sch}}
&:=
\widetilde U_{\mathrm{Sch}}^{(m)}
\rho^{\otimes m}
\widetilde U_{\mathrm{Sch}}^{(m)\dagger},
\end{align}
and define
\begin{align}
\widetilde p_{\mathrm{sym}}^{(k)}
:=
\trace{}{
\mathcal{R}_{m,k}
\left(
\widetilde{\rho}_{\mathrm{Sch}}
\right)
}.
\end{align}
For the ideal Schur transform,
\begin{align}
p_{\mathrm{sym}}^{(k)}
=
\trace{}{
\mathcal{R}_{m,k}
\left(
\rho_{\mathrm{Sch}}
\right)
}.
\end{align}

\begin{lemma}[Error induced by an approximate Schur transform]
\label{lem:estimation-schur-error}
Suppose that
\begin{equation}
\left\|
\widetilde U_{\mathrm{Sch}}^{(m)}
-
U_{\mathrm{Sch}}^{(m)}
\right\|_\infty
\leq
\epsilon_{\mathrm{Sch}}.
\end{equation}
Then
\begin{equation}
\left|
\widetilde p_{\mathrm{sym}}^{(k)}
-
p_{\mathrm{sym}}^{(k)}
\right|
\leq
\epsilon_{\mathrm{Sch}}.
\end{equation}
\end{lemma}

\begin{proof}
The ideal and approximate post-Schur states satisfy
\begin{equation}
\left\|
\widetilde{\rho}_{\mathrm{Sch}}
-
\rho_{\mathrm{Sch}}
\right\|_1
\leq
2\epsilon_{\mathrm{Sch}}.
\end{equation}
Since $\mathcal{R}_{m,k}$ is completely positive and
trace-nonincreasing,
\begin{equation}
0
\preceq
\mathcal{R}_{m,k}^\dagger(I)
\preceq
I.
\end{equation}
Moreover,
$\widetilde{\rho}_{\mathrm{Sch}}-\rho_{\mathrm{Sch}}$
is Hermitian and traceless. Therefore,
\begin{align}
\left|
\widetilde p_{\mathrm{sym}}^{(k)}-p_{\mathrm{sym}}^{(k)}
\right|
&=
\left|
\Tr\left[
\mathcal{R}_{m,k}^\dagger(I)
\left(
\widetilde{\rho}_{\mathrm{Sch}}
-
\rho_{\mathrm{Sch}}
\right)
\right]
\right|
\notag\\
&\leq
\frac12
\left\|
\widetilde{\rho}_{\mathrm{Sch}}
-
\rho_{\mathrm{Sch}}
\right\|_1
\leq
\epsilon_{\mathrm{Sch}}.
\end{align}
\end{proof}

\begin{lemma}[Median-of-means estimation of $p_{\mathrm{sym}}^{(k)}$ with an approximate Schur transform]
\label{lem:approximate-schur-mom}
Let
\begin{equation}
m:=k^2(k-1),
\end{equation}
and suppose that
\begin{equation}
\left\|
\widetilde U_{\mathrm{Sch}}^{(m)}
-
U_{\mathrm{Sch}}^{(m)}
\right\|_\infty
\leq
\epsilon_{\mathrm{Sch}}.
\end{equation}
For any $0<\varepsilon<1$ and $0<\delta<1$, if
\begin{equation}
\epsilon_{\mathrm{Sch}}
\leq
\frac{1}{(k!)^2},
\end{equation}
then there exists a median-of-means estimator
$\widehat p_{\mathrm{sym}}^{(k)}$ satisfying
\begin{equation}
\Pr\left[
\left|
\widehat p_{\mathrm{sym}}^{(k)}
-
\widetilde p_{\mathrm{sym}}^{(k)}
\right|
>
\varepsilon p_{\mathrm{sym}}^{(k)}
\right]
\leq
\delta.
\label{eq:approximate-schur-mom}
\end{equation}
Consequently,
\begin{equation}
\Pr\left[
\left|
\widehat p_{\mathrm{sym}}^{(k)}
-
p_{\mathrm{sym}}^{(k)}
\right|
>
\varepsilon p_{\mathrm{sym}}^{(k)}
+
\epsilon_{\mathrm{Sch}}
\right]
\leq
\delta.
\end{equation}
The estimator uses
\begin{equation}
N
=
k^2(k-1)
\left\lceil
\frac{6(e+1)}{\varepsilon^2}
\right\rceil
\left[
2\left\lceil
\frac94\ln\frac1\delta-\frac12
\right\rceil+1
\right]
\end{equation}
copies of $\rho$.
\end{lemma}

\begin{proof}
One repetition of the estimation procedure applies
$\widetilde U_{\mathrm{Sch}}^{(m)}$ to $\rho^{\otimes m}$,
measures the Schur label $\gamma$, and outputs
\begin{equation}
\widetilde X
:=
a_\gamma^{(k)}.
\end{equation}
By definition,
\begin{equation}
\mathbb{E}[\widetilde X]
=
\widetilde p_{\mathrm{sym}}^{(k)}.
\end{equation}

Define
\begin{equation}
A
:=
\sum_{\substack{\gamma\vdash m\\ \ell(\gamma)\leq D}}
\left(a_\gamma^{(k)}\right)^2
\Pi_\gamma.
\end{equation}
Since $0\leq a_\gamma^{(k)}\leq1$,
\begin{equation}
0\preceq A\preceq I.
\end{equation}
Therefore,
\begin{align}
\mathbb{E}\!\left[\widetilde X^2\right]
&=
\Tr\left[
\widetilde{\rho}_{\mathrm{Sch}}A
\right]
\notag\\
&\leq
\Tr\left[
\rho_{\mathrm{Sch}}A
\right]
+
\epsilon_{\mathrm{Sch}}
\notag\\
&=
\sum_\gamma
\mu_\gamma(\rho)
\left(a_\gamma^{(k)}\right)^2
+
\epsilon_{\mathrm{Sch}}.
\end{align}
By Lemma~\ref{lem:relative-concentration} and
$m=k^2(k-1)$,
\begin{equation}
\sum_\gamma
\mu_\gamma(\rho)
\left(a_\gamma^{(k)}\right)^2
\leq
e\left(p_{\mathrm{sym}}^{(k)}\right)^2.
\end{equation}
Moreover,
\begin{equation}
p_{\mathrm{sym}}^{(k)}
\geq
\frac{1}{k!},
\end{equation}
so that
\begin{equation}
\epsilon_{\mathrm{Sch}}
\leq
\frac{1}{(k!)^2}
\leq
\left(p_{\mathrm{sym}}^{(k)}\right)^2.
\end{equation}
Hence
\begin{equation}
\operatorname{Var}(\widetilde X)
\leq
\mathbb{E}\!\left[\widetilde X^2\right]
\leq
(e+1)\left(p_{\mathrm{sym}}^{(k)}\right)^2.
\end{equation}

We now apply the median-of-means procedure.
Set
\begin{align}
L
&:=
\left\lceil
\frac{6(e+1)}{\varepsilon^2}
\right\rceil,
\\
R
&:=
2\left\lceil
\frac94\ln\frac1\delta-\frac12
\right\rceil+1.
\end{align}
Perform $LR$ independent repetitions of the estimation procedure.
Denote the measured Schur labels by $\gamma_{r,\ell}$,
where $1\leq r\leq R$ and $1\leq\ell\leq L$, and define
\begin{equation}
\widetilde X_{r,\ell}
:=
a_{\gamma_{r,\ell}}^{(k)}.
\end{equation}
For each group, define
\begin{equation}
\widetilde Y_r
:=
\frac1L
\sum_{\ell=1}^L
\widetilde X_{r,\ell}.
\end{equation}
Then
\begin{equation}
\mathbb{E}[\widetilde Y_r]
=
\widetilde p_{\mathrm{sym}}^{(k)},
\qquad
\operatorname{Var}(\widetilde Y_r)
\leq
\frac{e+1}{L}
\left(p_{\mathrm{sym}}^{(k)}\right)^2.
\end{equation}
Chebyshev's inequality gives
\begin{equation}
\Pr\left[
\left|
\widetilde Y_r
-
\widetilde p_{\mathrm{sym}}^{(k)}
\right|
>
\varepsilon p_{\mathrm{sym}}^{(k)}
\right]
\leq
\frac16.
\end{equation}
Taking the median of $\widetilde Y_1,\ldots,\widetilde Y_R$ and applying the same
median-of-means argument as in
Lemma~\ref{lem:estimate-symmetric-probability} gives
\begin{equation}
\Pr\left[
\left|
\widehat p_{\mathrm{sym}}^{(k)}
-
\widetilde p_{\mathrm{sym}}^{(k)}
\right|
>
\varepsilon p_{\mathrm{sym}}^{(k)}
\right]
\leq
\delta.
\end{equation}
Combining this with
Lemma~\ref{lem:estimation-schur-error}
gives
\begin{equation}
\Pr\left[
\left|
\widehat p_{\mathrm{sym}}^{(k)}
-
p_{\mathrm{sym}}^{(k)}
\right|
>
\varepsilon p_{\mathrm{sym}}^{(k)}
+
\epsilon_{\mathrm{Sch}}
\right]
\leq
\delta.
\end{equation}

Finally, each sample consumes $m=k^2(k-1)$ copies of $\rho$,
so the total number of copies is
\begin{align}
N
&=
mLR
\notag\\
&=
k^2(k-1)
\left\lceil
\frac{6(e+1)}{\varepsilon^2}
\right\rceil
\left[
2\left\lceil
\frac94\ln\frac1\delta-\frac12
\right\rceil+1
\right].
\end{align}
\end{proof}

\section{Estimation of operator expectation values}\label{app:expectation-value}
\subsection{Overview}
In this section, we aim to efficiently evaluate $\mathrm{Tr}\!\left[{\chi_{\mathrm{Sc}}^{(k)}O_k}\right]$, assuming access to registers containing many copies of $\rho$. Here we denote the operator as $O_k$ to indicate that it acts on the Hilbert space $\mathcal{H}^{\otimes k}$.

Using the polynomial approximation, we evaluate the following quantity
as an approximation to the true expectation value:
\begin{equation}
   \mathrm{Tr}
\!\left[\chi_{\mathrm{PA}}^{(k)}(\lambda)O_k\right]=\sum_{i=1}^{\lambda}\binom{\lambda}{i}(-1)^{i-1} D^{k+a_i}
\mathrm{Tr}
\!\left[
\rho^{\otimes (k+a_i)}
\chi_{\mathrm{Haar}}^{(k+a_i)}\left(O_k\otimes I^{\otimes a_i}\right)
\right].
\end{equation}

To achieve this goal, we evaluate the $\lambda$ terms separately.
Namely, we estimate $\mathrm{Tr}
\!\left[
\rho^{\otimes k_i}
\chi_{\mathrm{Haar}}^{(k_i)}\left(O_k\otimes I^{\otimes a_i}\right)
\right]$ and combine the resulting estimates with the corresponding coefficients. Let 
\begin{equation}
\rho_{\mathrm{sym}}^{(k_i)}
:=
\frac{
\Pi_{\mathrm{sym}}^{(k_i)}
\rho^{\otimes k_i}
\Pi_{\mathrm{sym}}^{(k_i)}
}{
p_{\mathrm{sym}}^{(k_i)}
},
\end{equation}
where
$p_{\mathrm{sym}}^{(k_i)}
:=\mathrm{Tr}\!\left[\rho^{\otimes k_i}\Pi_{\mathrm{sym}}^{(k_i)}\right]$
is the normalization factor.
For the $i$-th term, define $o_i :=\mathrm{Tr}
\!\left[
\rho_{\mathrm{sym}}^{(k_i)}\left(O_k\otimes I^{\otimes a_i}\right)
\right]$.
Accordingly, the quantity we estimate is
\begin{equation}
       \mathrm{Tr}
\!\left[\chi_{\mathrm{PA}}^{(k)}(\lambda)O_k\right]=\sum_{i=1}^{\lambda}\binom{\lambda}{i}(-1)^{i-1} c_i p_{\mathrm{sym}}^{(k_i)}o_i .
\end{equation}
Completing our task requires estimating $p_{\mathrm{sym}}^{(k_i)}o_i$ for each $i$.

\subsection{Detailed procedure}
To reduce the sample complexity, we directly estimate the product
$p_{\mathrm{sym}}^{(k_i)}o_i$. 
Define 
\begin{equation} 
m_i:=k_i^2(k_i-1). 
\end{equation} 
The estimator is constructed by the following procedure: 
\begin{enumerate} 
    \item Prepare $\rho^{\otimes m_i}$ and apply the Schur transform 
    $U_{\mathrm{Sch}}^{(m_i)}$. 
    Measuring the irrep register produces a partition 
    $\gamma\vdash m_i$ with probability $\mu_\gamma(\rho)$. 

    \item For the measured sector $\gamma$, define 
    \begin{equation} 
    a_\gamma^{(k_i)} 
    := 
    \frac{ 
    \trace{}{P_\gamma^{(k_i)}} 
    }{ 
    r_\gamma 
    }, 
    \end{equation} 
    where $r_\gamma$ is the dimension of the permutation-space 
    irrep $\mathcal{P}_\gamma$, and $P_\gamma^{(k_i)}$ is the 
    projector onto the subspace of $\mathcal{P}_\gamma$ that is 
    symmetric under permutations of the first $k_i$ registers. 

    Conditioned on $\gamma$, replace the permutation register by 
    \begin{equation} 
    \frac{ 
    P_\gamma^{(k_i)} 
    }{ 
    \trace{}{P_\gamma^{(k_i)}} 
    } 
    \end{equation} 
    whenever $\trace{}{P_\gamma^{(k_i)}}>0$. 
    If $\trace{}{P_\gamma^{(k_i)}}=0$, prepare an arbitrary fixed 
    state in the permutation register and set 
    $a_\gamma^{(k_i)}=0$. 

    \item Apply the inverse Schur transform and denote the resulting 
    state on the $m_i$ registers by $\tau_\gamma^{(i)}$. 
    Measure 
    \begin{equation} 
    \left( 
    O_k\otimes I^{\otimes a_i} 
    \right) 
    \otimes 
    I^{\otimes(m_i-k_i)} 
    = 
    O_k\otimes I^{\otimes(m_i-k)} 
    \end{equation} 
    on $\tau_\gamma^{(i)}$. In repetition $(r,\ell)$, denote the
    measured Schur label and measurement outcome by
    $\gamma_{i,r,\ell}$ and $W_{i,r,\ell}$, respectively, and define
    the corresponding estimator sample by
    \begin{equation} 
    X_{i,r,\ell}
    := 
    a_{\gamma_{i,r,\ell}}^{(k_i)}W_{i,r,\ell}.
    \end{equation}

    \item Repeat the above procedure independently and use the
    median-of-means estimator. Specifically, divide the samples into
    $R_i$ groups, each containing $L_i$ samples. For the $r$-th
    group, define the sample mean
    \begin{equation}
    Y_{i,r}
    :=
    \frac{1}{L_i}
    \sum_{\ell=1}^{L_i}X_{i,r,\ell},
    \end{equation}
    and define
    \begin{equation}
    \widehat{p_{\mathrm{sym}}^{(k_i)}o_i}
    :=
    \operatorname{median}
    \left\{
    Y_{i,1},\ldots,Y_{i,R_i}
    \right\}.
    \end{equation}
\end{enumerate} 
Each sample satisfies
\begin{equation} 
\mathbb{E}[X_{i,r,\ell}]
= 
p_{\mathrm{sym}}^{(k_i)}o_i. 
\end{equation} 
Therefore, the above median-of-means estimator directly estimates
$p_{\mathrm{sym}}^{(k_i)}o_i$, without separately estimating
$p_{\mathrm{sym}}^{(k_i)}$ and $o_i$.

\begin{lemma}[Median-of-means estimation of $p_{\mathrm{sym}}^{(k)}o$]
\label{lem:estimate-symmetric-contribution}
Let $\rho$ be a density operator on $\mathbb{C}^D$,
$k\geq2$, and let $O_k$ be a Hermitian operator on
$(\mathbb{C}^D)^{\otimes k}$ satisfying
$\lVert O_k\rVert_\infty\leq1$.
Define
\begin{equation}
o:=
\trace{}{
\rho_{\mathrm{sym}}^{(k)}O_k
}.
\end{equation}
For any $0<\varepsilon<1$ and $0<\delta<1$,
there exists an estimator
$\widehat{p_{\mathrm{sym}}^{(k)}o}$ such that
\begin{align}
\Pr\!\left[
\left|
\widehat{p_{\mathrm{sym}}^{(k)}o}
-
p_{\mathrm{sym}}^{(k)} o
\right|
>
\varepsilon p_{\mathrm{sym}}^{(k)}
\right]
\leq\delta.
\end{align}
The estimator uses
\begin{align}
N
&=
k^2(k-1)
\left\lceil
\frac{6e}{\varepsilon^2}
\right\rceil
\left[
2\left\lceil
\frac94\ln\frac1\delta-\frac12
\right\rceil+1
\right]
\end{align}
copies of $\rho$.
\end{lemma}

\begin{proof}
Set $m:=k^2(k-1)$ and let
$X:=a_\gamma^{(k)}W$ be one sample produced by the procedure above.
Let $\tau_\gamma$ denote the conditional state obtained after the
permutation-register replacement and the inverse Schur transform.
The same Schur decomposition used in
Lemma~\ref{lem:estimate-symmetric-probability} gives
\begin{align}
\mathbb{E}[X]
&=
\sum_\gamma
\mu_\gamma(\rho)a_\gamma^{(k)}
\Tr\!\left[
\tau_\gamma
\left(O_k\otimes I^{\otimes(m-k)}\right)
\right]
\notag\\
&=
p_{\mathrm{sym}}^{(k)}o.
\end{align}
Moreover, $|W|\leq1$, and Lemma~\ref{lem:relative-concentration}
implies
\begin{align}
\mathbb{E}[X^2]
&\leq
\sum_{\gamma}
\mu_\gamma(\rho)
\left(a_\gamma^{(k)}\right)^2
\leq
e\left(p_{\mathrm{sym}}^{(k)}\right)^2.
\end{align}
Hence
$\operatorname{Var}(X)\leq
e\left(p_{\mathrm{sym}}^{(k)}\right)^2$.

Choose
\begin{equation}
L
:=
\left\lceil\frac{6e}{\varepsilon^2}\right\rceil,
\qquad
R
:=
2\left\lceil
\frac94\ln\frac1\delta-\frac12
\right\rceil+1.
\end{equation}
Applying the median-of-means argument from
Lemma~\ref{lem:estimate-symmetric-probability} to $R$ groups of
$L$ samples proves the stated probability bound. Each sample uses
$m$ copies of $\rho$, so $N=mLR$, as claimed.
\end{proof}

In practice, however, only an approximate Schur transform
$\widetilde{U}_{\mathrm{Sch}}^{(m_i)}$ is implemented.
Let $\widetilde{X}_{i,r,\ell}$ denote the estimator sample obtained from
the same procedure as above, but with the forward and inverse Schur
transforms replaced by
$\widetilde{U}_{\mathrm{Sch}}^{(m_i)}$ and
$\widetilde{U}_{\mathrm{Sch}}^{(m_i)\dagger}$, respectively.
We denote the corresponding expectation value by
\begin{equation}
\widetilde{p_{\mathrm{sym}}^{(k_i)}o_i}
:=
\mathbb{E}\!\left[\widetilde{X}_{i,r,\ell}\right].
\end{equation}
For the exact Schur transform,
\begin{equation}
\mathbb{E}\!\left[X_{i,r,\ell}\right]
=
p_{\mathrm{sym}}^{(k_i)}o_i.
\end{equation}
We next bound the deviation
\begin{equation}
\left|
\widetilde{p_{\mathrm{sym}}^{(k_i)}o_i}
-
p_{\mathrm{sym}}^{(k_i)} o_i
\right|
\end{equation}
in terms of the Schur-transform error
$\epsilon_{\mathrm{Sch}}^i$.

\begin{lemma}[Error induced by an approximate Schur transform]
\label{lem:joint-estimation-schur-error}
Let
\begin{equation}
\left\|
\widetilde{U}_{\mathrm{Sch}}^{(m_i)}
-
U_{\mathrm{Sch}}^{(m_i)}
\right\|_{\infty}
\leq
\epsilon_{\mathrm{Sch}}^i.
\end{equation}
Then
\begin{equation}
\left|
\widetilde{p_{\mathrm{sym}}^{(k_i)}o_i}
-
p_{\mathrm{sym}}^{(k_i)} o_i
\right|
\leq
4\epsilon_{\mathrm{Sch}}^i
\lVert O_k\rVert_{\infty}.
\end{equation}
\end{lemma}

\begin{proof}
Recall the completely positive trace-nonincreasing map
$\mathcal{R}_{m_i,k_i}$ defined in Eq.~\eqref{eq:Rmk}, and set
\begin{align}
\rho_{\mathrm{Sch}}^i
&:=
U_{\mathrm{Sch}}^{(m_i)}
\rho^{\otimes m_i}
U_{\mathrm{Sch}}^{(m_i)\dagger},
\\
\widetilde{\rho}_{\mathrm{Sch}}^i
&:=
\widetilde U_{\mathrm{Sch}}^{(m_i)}
\rho^{\otimes m_i}
\widetilde U_{\mathrm{Sch}}^{(m_i)\dagger}.
\end{align}
The unitary error assumption implies
\begin{equation}
\left\|
\widetilde{\rho}_{\mathrm{Sch}}^i-\rho_{\mathrm{Sch}}^i
\right\|_1
\leq
2\epsilon_{\mathrm{Sch}}^i.
\end{equation}
By construction,
\begin{align}
p_{\mathrm{sym}}^{(k_i)}o_i
&=
\Tr\!\left[
\mathcal{R}_{m_i,k_i}\!\left(\rho_{\mathrm{Sch}}^i\right)
U_{\mathrm{Sch}}^{(m_i)}
\left(
O_k\otimes I^{\otimes(m_i-k)}
\right)
U_{\mathrm{Sch}}^{(m_i)\dagger}
\right],
\\
\widetilde{p_{\mathrm{sym}}^{(k_i)}o_i}
&=
\Tr\!\left[
\mathcal{R}_{m_i,k_i}\!\left(
\widetilde{\rho}_{\mathrm{Sch}}^i
\right)
\widetilde U_{\mathrm{Sch}}^{(m_i)}
\left(
O_k\otimes I^{\otimes(m_i-k)}
\right)
\widetilde U_{\mathrm{Sch}}^{(m_i)\dagger}
\right].
\end{align}
Adding and subtracting the corresponding cross term and applying
H\"older's inequality gives
\begin{align}
&
\left|
\widetilde{p_{\mathrm{sym}}^{(k_i)}o_i}
-
p_{\mathrm{sym}}^{(k_i)}o_i
\right|
\nonumber\\
&\quad\leq
\left\|
\mathcal{R}_{m_i,k_i}\!\left(
\widetilde{\rho}_{\mathrm{Sch}}^i
-
\rho_{\mathrm{Sch}}^i
\right)
\right\|_1
\lVert O_k\rVert_\infty
\nonumber\\
&\qquad+
\left\|
\mathcal{R}_{m_i,k_i}\!\left(
\rho_{\mathrm{Sch}}^i
\right)
\right\|_1
\Bigl\|
\widetilde U_{\mathrm{Sch}}^{(m_i)}
\left(
O_k\otimes I^{\otimes(m_i-k)}
\right)
\widetilde U_{\mathrm{Sch}}^{(m_i)\dagger}
-
U_{\mathrm{Sch}}^{(m_i)}
\left(
O_k\otimes I^{\otimes(m_i-k)}
\right)
U_{\mathrm{Sch}}^{(m_i)\dagger}
\Bigr\|_\infty.
\end{align}
Since $\mathcal{R}_{m_i,k_i}$ is trace-nonincreasing and
trace-norm contractive on Hermitian operators,
\begin{equation}
\left\|
\mathcal{R}_{m_i,k_i}\!\left(
\widetilde{\rho}_{\mathrm{Sch}}^i
-
\rho_{\mathrm{Sch}}^i
\right)
\right\|_1
\leq
2\epsilon_{\mathrm{Sch}}^i,
\qquad
\left\|
\mathcal{R}_{m_i,k_i}\!\left(
\rho_{\mathrm{Sch}}^i
\right)
\right\|_1
\leq1.
\end{equation}
Moreover,
\begin{align}
&
\Bigl\|
\widetilde U_{\mathrm{Sch}}^{(m_i)}
\left(
O_k\otimes I^{\otimes(m_i-k)}
\right)
\widetilde U_{\mathrm{Sch}}^{(m_i)\dagger}
-
U_{\mathrm{Sch}}^{(m_i)}
\left(
O_k\otimes I^{\otimes(m_i-k)}
\right)
U_{\mathrm{Sch}}^{(m_i)\dagger}
\Bigr\|_\infty
\nonumber\\
&\quad\leq
2\epsilon_{\mathrm{Sch}}^i
\lVert O_k\rVert_\infty.
\end{align}
Therefore,
\begin{equation}
\left|
\widetilde{p_{\mathrm{sym}}^{(k_i)}o_i}
-
p_{\mathrm{sym}}^{(k_i)}o_i
\right|
\leq
4\epsilon_{\mathrm{Sch}}^i
\lVert O_k\rVert_\infty.
\end{equation}
\end{proof}

\begin{lemma}[Median-of-means estimation of $p_{\mathrm{sym}}^{(k_i)} o_i$ with an approximate Schur transform]
\label{lem:approximate-schur-joint-mom}
Let $m_i:=k_i^2(k_i-1)$, and let $O_k$ be a Hermitian operator
satisfying $\lVert O_k\rVert_\infty\leq1$. Suppose that
\begin{equation}
\left\|
\widetilde{U}_{\mathrm{Sch}}^{(m_i)}
-
U_{\mathrm{Sch}}^{(m_i)}
\right\|_\infty
\leq
\epsilon_{\mathrm{Sch}}^i.
\end{equation}
For any $0<\varepsilon<1$ and $0<\delta<1$, if
\begin{equation}
\epsilon_{\mathrm{Sch}}^i
\leq
\frac{1}{(k_i!)^2},
\end{equation}
there exists a
median-of-means estimator $\widehat{p_{\mathrm{sym}}^{(k_i)}o_i}$ satisfying
\begin{equation}
\Pr\left[
\left|
\widehat{p_{\mathrm{sym}}^{(k_i)}o_i}
-
\widetilde{p_{\mathrm{sym}}^{(k_i)}o_i}
\right|
>
\varepsilon p_{\mathrm{sym}}^{(k_i)}
\right]
\leq
\delta.
\label{eq:approximate-schur-joint-mom}
\end{equation}
The estimator uses
\begin{equation}
N_i
=
k_i^2(k_i-1)
\left\lceil
\frac{6(e+1)}{\varepsilon^2}
\right\rceil
\left[
2\left\lceil
\frac94\ln\frac1\delta-\frac12
\right\rceil+1
\right]
\end{equation}
copies of $\rho$.
\end{lemma}

\begin{proof}
For repetition $(r,\ell)$, let
\begin{equation}
\widetilde X_{i,r,\ell}
:=
a_{\gamma_{i,r,\ell}}^{(k_i)}
\widetilde W_{i,r,\ell},
\qquad
\mathbb{E}[\widetilde X_{i,r,\ell}]
=
\widetilde{p_{\mathrm{sym}}^{(k_i)}o_i}.
\end{equation}
Here $|\widetilde W_{i,r,\ell}|\leq1$. Define
\begin{equation}
A_i
:=
\sum_{\gamma\vdash m_i}
\left(a_\gamma^{(k_i)}\right)^2\Pi_\gamma,
\qquad
0\preceq A_i\preceq I.
\end{equation}
Let $\rho_{\mathrm{Sch}}^i$ and
$\widetilde\rho_{\mathrm{Sch}}^i$ denote the ideal and approximate
post-Schur states, respectively, as in the proof of
Lemma~\ref{lem:joint-estimation-schur-error}. Using
$\|\widetilde{\rho}_{\mathrm{Sch}}^i-\rho_{\mathrm{Sch}}^i\|_1
\leq2\epsilon_{\mathrm{Sch}}^i$, the second-moment bound from
Lemma~\ref{lem:estimate-symmetric-contribution}, and
$p_{\mathrm{sym}}^{(k_i)}\geq1/k_i!$, we obtain
\begin{align}
\mathbb{E}\!\left[\widetilde X_{i,r,\ell}^{\,2}\right]
&\leq
\Tr\!\left[\widetilde{\rho}_{\mathrm{Sch}}^iA_i\right]
\leq
e\left(p_{\mathrm{sym}}^{(k_i)}\right)^2
+\epsilon_{\mathrm{Sch}}^i
\nonumber\\
&\leq
(e+1)\left(p_{\mathrm{sym}}^{(k_i)}\right)^2.
\end{align}
Thus
$\operatorname{Var}(\widetilde X_{i,r,\ell})
\leq(e+1)(p_{\mathrm{sym}}^{(k_i)})^2$.

Choose
\begin{equation}
L_i
:=
\left\lceil\frac{6(e+1)}{\varepsilon^2}\right\rceil,
\qquad
R_i
:=
2\left\lceil
\frac94\ln\frac1\delta-\frac12
\right\rceil+1,
\end{equation}
and define
\begin{equation}
\widetilde Y_{i,r}
:=
\frac1{L_i}
\sum_{\ell=1}^{L_i}
\widetilde X_{i,r,\ell}.
\end{equation}
Define
\begin{equation}
\widehat{p_{\mathrm{sym}}^{(k_i)}o_i}
:=
\operatorname{median}\!\left(
\widetilde Y_{i,1},\ldots,\widetilde Y_{i,R_i}
\right).
\end{equation}
The median-of-means argument from
Lemma~\ref{lem:estimate-symmetric-probability} then yields
Eq.~\eqref{eq:approximate-schur-joint-mom}. Since each sample uses
$m_i=k_i^2(k_i-1)$ copies of $\rho$, the total copy complexity is
$N_i=m_iL_iR_i$, as stated.
\end{proof}

\subsection{Supporting lemmas}

In this section we present necessary technical lemmas.

The union bound is useful for our analysis.
\begin{lemma}[Union bound for a sum of random variables]
\label{lem:union-bound}
Let $X_1,\ldots,X_k$ be arbitrary random variables, and let
$a_1,\ldots,a_k$ be positive real numbers. Then
\begin{equation}
\Pr\!\left[
\sum_{i=1}^k X_i
>
\sum_{i=1}^k a_i
\right]
\leq
\sum_{i=1}^k
\Pr\!\left[
X_i>a_i
\right].
\end{equation}
\end{lemma}

\begin{proof}
If $X_i\leq a_i$ holds for every $i$, then
\begin{equation}
\sum_{i=1}^k X_i
\leq
\sum_{i=1}^k a_i.
\end{equation}
Therefore,
\begin{equation}
\left\{
\sum_{i=1}^k X_i
>
\sum_{i=1}^k a_i
\right\}
\subseteq
\bigcup_{i=1}^k
\left\{
X_i>a_i
\right\}.
\end{equation}
Applying the union bound gives
\begin{equation}
\Pr\!\left[
\sum_{i=1}^k X_i
>
\sum_{i=1}^k a_i
\right]
\leq
\sum_{i=1}^k
\Pr\!\left[
X_i>a_i
\right].
\end{equation}
\end{proof}

We also need to bound the post-selection probability.
\begin{lemma}[Bounds on the post-selection probability]
\label{lem:weighted-symmetric-probability}
Let $\rho$ be a density operator on a $D$-dimensional Hilbert space $\mathcal{H}$, with eigenvalues $\beta_1,\ldots,\beta_D$ and the largest eigenvalue $\beta_{\mathrm{max}}$.
Then for integer $k\geq1$,
\begin{equation}
\mathrm{Tr}\!\left[
\rho^{\otimes k}\Pi_{\mathrm{sym}}^{(k)}
\right]
\leq
\frac{1}{k!}
\exp\left[
\frac{k(k-1)\beta_{\mathrm{max}}}{2}
\right],
\label{eq:weighted-symmetric-probability-bound}
\end{equation}
where $\Pi_{\mathrm{sym}}^{(k)}
=
\frac{1}{k!}
\sum_{\pi\in\mathcal{S}_{k}}\pi$
is the projector onto the totally symmetric subspace, and each
$\pi\in\mathcal{S}_{k}$ denotes the unitary representation of a permutation in the symmetric group $\mathcal{S}_{k}$ acting on the $k$ copies of $\mathcal{H}$.
\end{lemma}
\begin{proof}
    
We use $\pi$ to denote both a group element and its unitary
representation. Let $\pi\in\calS_k$ have cycle structure described by
the nonnegative integers $\nu_{\pi,1},\ldots,\nu_{\pi,k}$; that is,
$\pi$ has $\nu_{\pi,1}$ one-cycles, $\nu_{\pi,2}$ two-cycles, and so
on. These integers satisfy
$\nu_{\pi,1}+2\nu_{\pi,2}+\cdots+k\nu_{\pi,k}=k$.
Then we have
\begin{align}
    \mathrm{Tr}\!\left[\rho^{\otimes k}\pi\right]
&=\prod_{m=1}^k\left(\mathrm{Tr}\!\left[\rho^m\right]\right)^{\nu_{\pi,m}}
    =
    \prod_{m=1}^k\left(\sum_{i=1}^D \beta_i^m\right)^{\nu_{\pi,m}},
\end{align}
where $\beta_1,\ldots,\beta_D$ denote the eigenvalues of $\rho$. The sum of all symmetric group elements yields a closed form,
\begin{align}
    \sum_{\pi\in\calS_k}\mathrm{Tr}\!\left[\rho^{\otimes k}\pi\right]
    =
    k!h_k(\beta_1,\ldots,\beta_D),
\end{align}
where
\begin{align}
    h_k(x_1,\ldots,x_D)
    =
    \sum_{1\leq i_1\leq\cdots\leq i_k\leq D}
    \prod_{m=1}^k x_{i_m}
\end{align}
is the complete symmetric polynomial of degree $k$. It has a well-known generating function
\begin{align}
    \log\left(
    \sum_{l=0}^{\infty}
    h_l(\beta_1,\ldots,\beta_D)t^l
    \right)
    &=
    \sum_{m=1}^\infty
    \frac{t^m}{m}
    \sum_{i=1}^D\beta_i^m
    \notag\\
    &\leq 
    \sum_{m=1}^\infty
    \frac{t^m}{m}
    \beta_{\mathrm{max}}^{m-1}
    \sum_{i=1}^D\beta_i
    \notag\\
    &=
    \sum_{m=1}^\infty
    \frac{t^m}{m}
    \beta_{\mathrm{max}}^{m-1}
    \notag\\
    &=
    -\frac{1}{\beta_{\mathrm{max}}}
    \log\left(1-\beta_{\mathrm{max}}t\right).
\end{align}
The preceding inequality holds coefficientwise as a formal power series.
Exponentiation preserves coefficientwise inequalities between series with
nonnegative coefficients and zero constant term. Therefore,
\begin{align}
    \sum_{l=0}^{\infty}
    h_l(\beta_1,\ldots,\beta_D)t^l
    &\leq
    \frac{1}
    {\left(1-\beta_{\mathrm{max}}t\right)^{1/\beta_{\mathrm{max}}}}
    \notag\\
    &=
    \sum_{l=0}^\infty
    \frac{
    \frac{1}{\beta_{\mathrm{max}}}
    \left(\frac{1}{\beta_{\mathrm{max}}}+1\right)
    \cdots
    \left(\frac{1}{\beta_{\mathrm{max}}}+l-1\right)
    }{l!}
    (t\beta_{\mathrm{max}})^l
\end{align}
again coefficientwise. Hence,
\begin{align}
    h_k(\beta_1,\ldots,\beta_D)
    \leq 
    \frac{
    \frac{1}{\beta_{\mathrm{max}}}
    \left(\frac{1}{\beta_{\mathrm{max}}}+1\right)
    \cdots
    \left(\frac{1}{\beta_{\mathrm{max}}}+k-1\right)
    }{k!}
    \cdot
    \beta_{\mathrm{max}}^k.
\end{align}
This gives
\begin{align}
    \mathrm{Tr}\!\left[\rho^{\otimes k}\sum_{\pi\in\calS_k}\pi\right]
    &\leq 
    \frac{1}{\beta_{\mathrm{max}}}
    \left(\frac{1}{\beta_{\mathrm{max}}}+1\right)
    \cdots
    \left(\frac{1}{\beta_{\mathrm{max}}}+k-1\right)
    \cdot
    \beta_{\mathrm{max}}^k
    \notag\\
    &=
    \prod_{j=1}^{k-1}
    \left(1+j\beta_{\mathrm{max}}\right)
    \notag\\
    &\leq
    e^{\frac{k(k-1)\beta_{\mathrm{max}}}{2}}.
\end{align}
The last inequality is obtained by $1+x\leq e^x$.
\end{proof}

\subsection{Sample and circuit complexity analysis}

Combining the preceding results, we now give the sample and circuit
complexities of the expectation-value estimation algorithm. We consider an
$n$-qubit system, so that $D=2^n$, and retain the definition
$\M=\lfloor\lVert\rho\rVert_\infty^{-1}\rfloor$.

\begin{theorem}
[Sample and circuit complexities of operator expectation value estimation algorithm]
\label{Sample and circuit complexities of operator expectation value estimation algorithm}
Let $\rho$ be an $n$-qubit density matrix, let $k\geq2$, and suppose
$12(k-1)^2\leq \M$. Let $O_k$ be a Hermitian operator on
$(\mathbb C^{2^n})^{\otimes k}$ satisfying
$\lVert O_k\rVert_\infty\leq1$, and let $0<\epsilon_{\mathrm t}<1$ satisfy
\begin{equation}
\epsilon_{\mathrm t}
\geq
2C_{\add}
\exp\left[
-\frac{\M}{24(k-1)^2}
\right],
\label{eq:bound-epsilon}
\end{equation}
where $C_{\add}=e^{11/12}$ is the constant in
Corollary~\ref{coro:minimal-additive-error}. For any $0<\delta<1$,
assuming an efficient implementation of the measurement of $O_k$, there
is a quantum algorithm that estimates
$\trace{}{\chi_{\mathrm{Sc}}^{(k)}(\rho)O_k}$ up to additive error
$\epsilon_{\mathrm t}$ with success probability at least $1-\delta$.

A sufficient total number of copies of $\rho$ is
\begin{align}
N_{\mathrm{total}}^{\mathrm{suff}}
=
\mathcal O\Bigg(
k^3
\log^4\left(
\frac{2C_{\add}}{\epsilon_{\mathrm t}}
\right)
\left[
\log\frac{1}{\delta}
+
\log\log\left(
\frac{2C_{\add}}{\epsilon_{\mathrm t}}
\right)
\right]
\left(
\frac{2C_{\add}}{\epsilon_{\mathrm t}}
\right)^5
\Bigg).
\label{eq:sample-complexity-final-epsilon}
\end{align}
Each independent estimator sample can be implemented with circuit
complexity
\begin{equation}
\operatorname{poly}\left(
n,k,\log\frac{1}{\epsilon_{\mathrm t}}
\right).
\end{equation}
Consequently, a sufficient total circuit complexity is
\begin{equation}
C_{\mathrm{total}}^{\mathrm{suff}}
=
\operatorname{poly}\!\left(
n,k,\frac{1}{\epsilon_{\mathrm t}},
\log\frac{1}{\delta}
\right).
\label{eq:total-circuit-complexity}
\end{equation}
\end{theorem}

\begin{proof}
We first allocate an error budget $\epsilon_{\mathrm t}/2$ to the
polynomial approximation. By
Corollary~\ref{coro:minimal-additive-error},
\begin{align}
\left|
\trace{}{
\left(
\chi_{\mathrm{PA}}^{(k)}(\lambda)
-
\chi_{\mathrm{Sc}}^{(k)}
\right)
O_k
}
\right|
\leq
\left\|
\chi_{\mathrm{PA}}^{(k)}(\lambda)
-
\chi_{\mathrm{Sc}}^{(k)}
\right\|_1
\lVert O_k\rVert_\infty.
\end{align}
Since $\lVert O_k\rVert_\infty\leq1$, it suffices to choose
\begin{equation}
\lambda
=
\max\left\{
\left\lfloor
2\log\left(
\frac{2C_{\add}}{\epsilon_{\mathrm t}}
\right)
\right\rfloor,
1
\right\}.
\label{eq:bound-lambda}
\end{equation}
Then
\begin{equation}
\left|
\trace{}{
\left(
\chi_{\mathrm{PA}}^{(k)}(\lambda)
-
\chi_{\mathrm{Sc}}^{(k)}
\right)
O_k
}
\right|
\leq
\frac{\epsilon_{\mathrm t}}{2}.
\end{equation}

It remains to estimate
\begin{equation}
\trace{}{\chi_{\mathrm{PA}}^{(k)}(\lambda)O_k}
=
\sum_{i=1}^{\lambda}
\binom{\lambda}{i}
(-1)^{i-1}
c_i p_{\mathrm{sym}}^{(k_i)} o_i.
\end{equation}
We divide the remaining error budget equally between the statistical
estimation error and the error induced by the approximate Schur
transforms.

For the latter, Lemma~\ref{lem:joint-estimation-schur-error} gives
\begin{equation}
\left|
\widetilde{p_{\mathrm{sym}}^{(k_i)}o_i}
-
p_{\mathrm{sym}}^{(k_i)} o_i
\right|
\leq
4\epsilon_{\mathrm{Sch}}^i
\lVert O_k\rVert_\infty
\leq
4\epsilon_{\mathrm{Sch}}^i.
\end{equation}
Hence the total contribution of the Schur-transform errors is bounded by
\begin{equation}
4
\sum_{i=1}^{\lambda}
\binom{\lambda}{i}
c_i\epsilon_{\mathrm{Sch}}^i.
\end{equation}
To make this contribution at most $\epsilon_{\mathrm t}/4$, it
suffices to require
\begin{equation}
\epsilon_{\mathrm{Sch}}^i
\leq
\frac{\epsilon_{\mathrm t}}
{16(2^\lambda-1)c_i},
\end{equation}
since
\begin{equation}
\sum_{i=1}^{\lambda}
\binom{\lambda}{i}
=
2^\lambda-1.
\end{equation}
In addition, Lemma~\ref{lem:approximate-schur-joint-mom} requires
$\epsilon_{\mathrm{Sch}}^i\leq1/(k_i!)^2$. We therefore choose
\begin{equation}
\epsilon_{\mathrm{Sch}}^i
=
\min\left\{
\frac{\epsilon_{\mathrm t}}
{16(2^\lambda-1)c_i},
\frac{1}{(k_i!)^2}
\right\}.
\label{eq:schur-precision-choice}
\end{equation}

We next control the statistical estimation error.
For each $i$, let $\varepsilon_i$ denote the relative statistical
accuracy parameter in
Lemma~\ref{lem:approximate-schur-joint-mom}, so that
\begin{equation}
\left|
\widehat{p_{\mathrm{sym}}^{(k_i)}o_i}
-
\widetilde{p_{\mathrm{sym}}^{(k_i)}o_i}
\right|
\leq
\varepsilon_i p_{\mathrm{sym}}^{(k_i)}.
\end{equation}
By Lemma~\ref{lem:weighted-symmetric-probability} and $c_i\leq k_i!$,
\begin{equation}
c_i p_{\mathrm{sym}}^{(k_i)}
\leq
\exp\left[
\frac{k_i(k_i-1)}{2\M}
\right]
\leq
\exp\left[
\frac{k_\lambda(k_\lambda-1)}{2\M}
\right].
\end{equation}
Therefore, the total statistical error is bounded by
\begin{align}
&\sum_{i=1}^{\lambda}
\binom{\lambda}{i}
c_i
\left|
\widehat{p_{\mathrm{sym}}^{(k_i)}o_i}
-
\widetilde{p_{\mathrm{sym}}^{(k_i)}o_i}
\right|\leq
\sum_{i=1}^{\lambda}
\binom{\lambda}{i}
\varepsilon_i
\exp\left[
\frac{k_\lambda(k_\lambda-1)}{2\M}
\right].
\end{align}
We choose the same relative accuracy for all $i$,
\begin{equation}
\varepsilon_i
:=
\frac{\epsilon_{\mathrm t}}
{4(2^\lambda-1)}
\exp\left[
-\frac{k_\lambda(k_\lambda-1)}{2\M}
\right],
\label{eq:relative-statistical-error-choice}
\end{equation}
which makes the total statistical contribution at most
$\epsilon_{\mathrm t}/4$.

Applying Lemma~\ref{lem:approximate-schur-joint-mom} to each contribution
with relative accuracy $\varepsilon_i$ and failure probability
$\delta/\lambda$ gives
\begin{equation}
\Pr\left[
\left|
\widehat{p_{\mathrm{sym}}^{(k_i)}o_i}
-
\widetilde{p_{\mathrm{sym}}^{(k_i)}o_i}
\right|
>
\varepsilon_i p_{\mathrm{sym}}^{(k_i)}
\right]
\leq
\frac{\delta}{\lambda}.
\label{eq:individual-mom-error}
\end{equation}
The corresponding number of copies is
\begin{align}
N_i
=
k_i^2(k_i-1)
\left\lceil
\frac{6(e+1)}{\varepsilon_i^2}
\right\rceil
\left[
2\left\lceil
\frac94
\ln\frac{\lambda}{\delta}
-\frac12
\right\rceil+1
\right].
\label{eq:individual-sample-complexity}
\end{align}
By the union bound, with probability at least $1-\delta$,
Eq.~\eqref{eq:individual-mom-error} holds simultaneously for all
$1\leq i\leq\lambda$.

Combining the above bounds, we obtain
\begin{align}
&\left|
\sum_{i=1}^{\lambda}
\binom{\lambda}{i}
(-1)^{i-1}
c_i
\widehat{p_{\mathrm{sym}}^{(k_i)}o_i}
-
\trace{}{\chi_{\mathrm{PA}}^{(k)}(\lambda)O_k}
\right|
\nonumber\\
&\leq
\sum_{i=1}^{\lambda}
\binom{\lambda}{i}
c_i
\left|
\widehat{p_{\mathrm{sym}}^{(k_i)}o_i}
-
\widetilde{p_{\mathrm{sym}}^{(k_i)}o_i}
\right|+
\sum_{i=1}^{\lambda}
\binom{\lambda}{i}
c_i
\left|
\widetilde{p_{\mathrm{sym}}^{(k_i)}o_i}
-
p_{\mathrm{sym}}^{(k_i)} o_i
\right|
\nonumber\\
&\quad\leq
\frac{\epsilon_{\mathrm t}}{4}
+
\frac{\epsilon_{\mathrm t}}{4}
=
\frac{\epsilon_{\mathrm t}}{2}.
\end{align}
Combining this with the polynomial-approximation error, we
obtain a total additive error at most $\epsilon_{\mathrm t}$ with
probability at least $1-\delta$.

\textbf{Sample complexity.}---
By Eq.~\eqref{eq:individual-sample-complexity}, the total number of
copies of $\rho$ is
\begin{equation}
N_{\mathrm{total}}
=
\sum_{i=1}^{\lambda}N_i.
\end{equation}
Since
\begin{equation}
k_i^2(k_i-1)
\leq
k_\lambda^3,
\end{equation}
and
\begin{equation}
\frac{1}{\varepsilon_i^2}
=
\frac{16(2^\lambda-1)^2}{\epsilon_{\mathrm t}^2}
\exp\left[
\frac{k_\lambda(k_\lambda-1)}{\M}
\right],
\end{equation}
we obtain
\begin{align}
N_{\mathrm{total}}^{\mathrm{suff}}
=
\mathcal O\left(
\lambda k_\lambda^3
\frac{(2^\lambda-1)^2}{\epsilon_{\mathrm t}^2}
\exp\left[
\frac{k_\lambda(k_\lambda-1)}{\M}
\right]
\log\frac{\lambda}{\delta}
\right).
\label{eq:sample-complexity}
\end{align}

We next express this bound in terms of the target error
$\epsilon_{\mathrm t}$. Equation~\eqref{eq:bound-epsilon} implies
\begin{equation}
\log\left(
\frac{2C_{\add}}{\epsilon_{\mathrm t}}
\right)
\leq
\frac{\M}{24(k-1)^2}.
\end{equation}
Moreover, by Eq.~\eqref{eq:bound-lambda},
\begin{equation}
\lambda
\leq
2\log\left(
\frac{2C_{\add}}{\epsilon_{\mathrm t}}
\right),
\end{equation}
and hence
\begin{equation}
k_\lambda
=
\lambda(k-1)+1
\leq
2(k-1)
\log\left(
\frac{2C_{\add}}{\epsilon_{\mathrm t}}
\right)
+1.
\label{eq:lambda-range-bound}
\end{equation}
It follows that
\begin{align}
\frac{k_\lambda(k_\lambda-1)}{\M}
&\leq
\frac{
\left[
2(k-1)
\log\left(
\frac{2C_{\add}}{\epsilon_{\mathrm t}}
\right)
+1
\right]
2(k-1)
\log\left(
\frac{2C_{\add}}{\epsilon_{\mathrm t}}
\right)
}{\M}
\nonumber\\
&\leq
\frac{1}{6}
\log\left(
\frac{2C_{\add}}{\epsilon_{\mathrm t}}
\right)
+
\frac{1}{12},
\end{align}
and therefore
\begin{equation}
\exp\left[
\frac{k_\lambda(k_\lambda-1)}{\M}
\right]
\leq
e^{1/12}
\left(
\frac{2C_{\add}}{\epsilon_{\mathrm t}}
\right)^{1/6}.
\end{equation}
Similarly,
\begin{equation}
(2^\lambda-1)^2
\leq
4^\lambda
\leq
\left(
\frac{2C_{\add}}{\epsilon_{\mathrm t}}
\right)^{2\log4}.
\end{equation}
Substituting these bounds into
Eq.~\eqref{eq:sample-complexity}, together with
\begin{equation}
\log\frac{\lambda}{\delta}
=
\mathcal O\left(
\log\frac1\delta
+
\log\log\frac{2C_{\add}}{\epsilon_{\mathrm t}}
\right),
\end{equation}
gives
\begin{align}
N_{\mathrm{total}}^{\mathrm{suff}}
=
\mathcal O\Bigg(
&k^3
\log^4\left(
\frac{2C_{\add}}{\epsilon_{\mathrm t}}
\right)
\left[
\log\frac1\delta
+
\log\log\left(
\frac{2C_{\add}}{\epsilon_{\mathrm t}}
\right)
\right]
\left(
\frac{2C_{\add}}{\epsilon_{\mathrm t}}
\right)^{
\frac{13}{6}+2\log4
}
\Bigg).
\end{align}
Since
\begin{equation}
\frac{13}{6}+2\log4<5,
\end{equation}
and $2C_{\add}/\epsilon_{\mathrm t}>1$, this implies
\begin{align}
N_{\mathrm{total}}^{\mathrm{suff}}
=
\mathcal O\Bigg(
&k^3
\log^4\left(
\frac{2C_{\add}}{\epsilon_{\mathrm t}}
\right)
\left[
\log\frac1\delta
+
\log\log\left(
\frac{2C_{\add}}{\epsilon_{\mathrm t}}
\right)
\right]
\left(
\frac{2C_{\add}}{\epsilon_{\mathrm t}}
\right)^5
\Bigg),
\end{align}
which proves
Eq.~\eqref{eq:sample-complexity-final-epsilon}.

\textbf{Circuit complexity.}---
For the $i$-th contribution, each estimator sample applies a Schur
transform to
\begin{equation}
m_i
=
k_i^2(k_i-1)
\end{equation}
copies of $\rho$. By Lemma~\ref{lem:time-schur}, the approximate Schur
transform can be implemented with circuit complexity
\begin{equation}
\operatorname{poly}\left(
m_i,n,\log\frac{1}{\epsilon_{\mathrm{Sch}}^i}
\right).
\end{equation}
The permutation-register replacement and the measurement of $O_k$ can
also be implemented efficiently by assumption.

It therefore remains to bound
$\log(1/\epsilon_{\mathrm{Sch}}^i)$. From
Eq.~\eqref{eq:schur-precision-choice},
\begin{align}
\log\frac{1}{\epsilon_{\mathrm{Sch}}^i}
&=
\max\left\{
\log\left[
\frac{16(2^\lambda-1)c_i}
{\epsilon_{\mathrm t}}
\right],
2\log(k_i!)
\right\}.
\end{align}
Using $c_i\leq k_i!$, we obtain
\begin{equation}
\log\frac{1}{\epsilon_{\mathrm{Sch}}^i}
=
\mathcal O\left(
\lambda
+
\log\frac1{\epsilon_{\mathrm t}}
+
\log(k_i!)
\right).
\end{equation}
Moreover,
\begin{equation}
\log(k_i!)
\leq
k_i\log k_i
\leq
k_\lambda\log k_\lambda.
\end{equation}
By Eq.~\eqref{eq:lambda-range-bound},
\begin{align}
\log(k_\lambda!)
=
\mathcal O\Bigg(
&k
\log\left(
\frac{2C_{\add}}{\epsilon_{\mathrm t}}
\right)
\left[
\log k
+
\log\log\left(
\frac{2C_{\add}}{\epsilon_{\mathrm t}}
\right)
\right]
\Bigg),
\end{align}
while
\begin{equation}
m_i
\leq
k_\lambda^3
=
\mathcal O\left(
k^3
\log^3\left(
\frac{2C_{\add}}{\epsilon_{\mathrm t}}
\right)
\right).
\end{equation}
Therefore, each independent experiment has circuit complexity
\begin{equation}
\operatorname{poly}\left(
n,k,\log\frac1{\epsilon_{\mathrm t}}
\right).
\end{equation}
There are at most $N_{\mathrm{total}}^{\mathrm{suff}}$ independent
samples, because every sample consumes at least one copy of $\rho$.
Multiplying the per-sample bound by this upper bound gives
Eq.~\eqref{eq:total-circuit-complexity}.
\end{proof}

\section{Construction of block encoding}\label{ap:Construction of block encoding}

In this section, we provide an efficient construction of the block encoding of Scrooge moment $\chi_\mathrm{Sc}^{(k)}$ when an efficient preparation circuit of a purification of $\rho$ is available. Our method is to construct the block encoding of $\chi_{\mathrm{PA}}^{(k)}(\lambda)$ as a good approximation. Below, we first review necessary algorithmic preliminaries.

\subsection{Algorithmic preliminaries}
\subsubsection{Block encoding}

Block encoding provides a convenient way to represent a generally non-unitary operator as a sub-block of a larger unitary, allowing it to be manipulated coherently within a quantum circuit.
\begin{definition}
    [Block encoding]
    \label{def:block-encoding}
    Suppose that $A$ is an $n$-qubit operator, $\alpha>0$,
$\epsilon\geq0$, and $s\in\mathbb N_0$. We say a $(n+s)$-qubit unitary $U$ is an $(\alpha,s,\epsilon)$-block-encoding of $A$, if
    \begin{align}
        \norm{A-\alpha\left(\bra{0}^{\otimes s}\otimes I\right)U\left(\ket{0}^{\otimes s}\otimes I\right)}_\infty\leq \epsilon
    \end{align}
\end{definition}

In particular, coherent access to a unitary that prepares a purification of a density matrix $\rho$ directly yields an exact block encoding of $\rho$.

\begin{lemma}[Block encoding from a purification circuit, cf.~Lemma~45 in Ref.~\cite{gilyen2019quantum}]
\label{Block encoding from a purification circuit}
Let $\sfS$ be an $n$-qubit system register and $\sfA$ an $s$-qubit ancilla register. Suppose that $U_{\sfA,\sfS}$ prepares a purification of an $n$-qubit density matrix $\rho$, namely, $U_{\sfA,\sfS}\bigl(|0\rangle_{\sfA}^{\otimes {s}}|0\rangle_{\sfS}^{\otimes n}\bigr)=|\Psi_\rho\rangle_{\sfA,\sfS}$, with $\operatorname{Tr}_{\sfA}\bigl(|\Psi_\rho\rangle\langle\Psi_\rho|\bigr)=\rho$. Introduce an additional $n$-qubit register $\sfA'$, and define $W_{\sfA,\sfS,\sfA'}=\bigl(U_{\sfA,\sfS}^{\dagger}\otimes I_{\sfA'}\bigr)\bigl(I_{\sfA}\otimes \operatorname{SWAP}_{\sfS,\sfA'}\bigr)\bigl(U_{\sfA,\sfS}\otimes I_{\sfA'}\bigr)$. Then $W_{\sfA,\sfS,\sfA'}$ is an exact $(1,n+s,0)$-block-encoding of $\rho$ acting on $\sfA'$.
\end{lemma}

\subsubsection{Linear combination of unitaries} 
When block encodings are available for individual terms of an operator, the linear-combination-of-unitaries (LCU) technique allows them to be combined into a block encoding of their linear combination. 
 
\begin{definition}  
    [State preparation pair]  
    \label{def:preparation-pair}  
    Let $y \in \mathbb{C}^m$ and $\|y\|_1 \le \beta$. The pair of unitaries $(P_L, P_R)$  
is called a $(\beta,r,\epsilon)$-state-preparation-pair if  
$P_L \lvert 0 \rangle^{\otimes r} = \sum_{j=0}^{2^r-1} b_j \lvert j \rangle$  
and  
$P_R \lvert 0 \rangle^{\otimes r} = \sum_{j=0}^{2^r-1} d_j \lvert j \rangle$  
such that  
$\sum_{j=0}^{m-1} \left| \beta(b_j^* d_j) - y_j \right| \le \epsilon$  
and for all $j \in \{m, \ldots, 2^r - 1\}$ we have  
$b_j^* d_j = 0.$  
\end{definition}  
  
\begin{lemma}  
    [Linear combination of block encodings, cf.~Lemma~52 in Ref.~\cite{gilyen2019quantum}]  
    \label{lem:lc-BE}  
    Let $A = \sum_{j=0}^{m-1} y_j A_j$ be an $n$-qubit operator and $\epsilon \in \mathbb{R}_+$.  
Suppose that $(P_L, P_R)$ is a $(\beta,r,\epsilon_1)$-state-preparation-pair for $y$,  
$W = \sum_{j=0}^{m-1} \lvert j \rangle \langle j \rvert \otimes U_j + \left( I - \sum_{j=0}^{m-1} \lvert j \rangle \langle j \rvert \right)\otimes I_{\sf A} \otimes I_{\sf S}$  
is an $r+s+n$ qubit unitary such that for all $j \in\{ 0,\ldots,m-1\}$ we have that  
$U_j$ is an $(\alpha,s,\epsilon_2)$-block-encoding of $A_j$.  
Then we can implement an $(\alpha\beta, s+r, \alpha\epsilon_1+\alpha\beta\epsilon_2)$-block-encoding of $A$,  
with a single use of $W$, $P_R$ and $P_L^\dagger$.  
\end{lemma}

\subsection{Detailed Procedure}\label{app:be:procedure}
The polynomial approximation moment can be written as
\begin{align}
\chi_{\mathrm{PA}}^{(k)}(\lambda)
=
\sum_{i=1}^{\lambda}
\binom{\lambda}{i}(-1)^{i-1}
c_i p_{\mathrm{sym}}^{(k_i)}\,
\Tr_{\mathrm{aux}}
\!\left[
\rho_{\mathrm{sym}}^{(k_i)}
\right].
\label{eq:block-direct-decomposition}
\end{align}

For convenience, define
\begin{equation}
\sigma_i
:=
\Tr_{\mathrm{aux}}
\!\left[
\rho_{\mathrm{sym}}^{(k_i)}
\right].
\end{equation}
Then
\begin{equation}
\chi_{\mathrm{PA}}^{(k)}(\lambda)
=
\sum_{i=1}^{\lambda}
\binom{\lambda}{i}(-1)^{i-1}
c_i p_{\mathrm{sym}}^{(k_i)} \sigma_i.
\end{equation}

Our construction proceeds in two steps for each term.

First, we estimate $p_{\mathrm{sym}}^{(k_i)}$ and denote the resulting
estimate by $\widehat p_{\mathrm{sym}}^{(k_i)}$.

Second, assume access to $U_\rho$ which denotes an efficient circuit that prepares a purification of 
$\rho$, satisfying 
\begin{equation} 
    \left(U_\rho\right)_{{\sf SA}}
    \left( 
    \ket{0}_{\sf S}^{\otimes n} 
    \ket{0}_{\sf A}^{\otimes s_\rho} 
    \right) 
    = 
    \ket{\Psi_\rho}_{{\sf SA}}, 
\end{equation} 
with 
\begin{equation} 
    \rho 
    = 
    \Tr_{\sf A} 
    \left[ 
    \ket{\Psi_\rho}\!\bra{\Psi_\rho}_{{\sf SA}} 
    \right], 
\end{equation} 
where $n$ denotes the number of qubits required to represent a single copy of the system, $s_\rho=\left\lceil\log_2\operatorname{rank}(\rho)\right\rceil\leq n$ denotes the number of ancillary qubits required to purify $\rho$, ${\sf S}$ denotes the system register and ${\sf A}$ denotes the purification ancilla register. 

Using the construction in
Corollary~\ref{purification-of-channel},
we obtain a purification circuit for a state that approximates the
normalized symmetric state
$\rho_{\mathrm{sym}}^{(k_i)}$. By treating the $a_i$ system registers traced out in the definition of
$\sigma_i$ as part of the purification ancilla, the same circuit also
provides an approximate purification of a state approximating $\sigma_i$.

Finally, we combine the resulting block encodings using the LCU technique
with coefficients
\begin{equation}
\widehat y_i
:=
\binom{\lambda}{i}(-1)^{i-1}
c_i\widehat p_{\mathrm{sym}}^{(k_i)},
\end{equation}
thereby obtaining an approximate block encoding of
\begin{equation}
\widehat{\chi}_{\mathrm{PA}}^{(k)}(\lambda)
:=
\sum_{i=1}^{\lambda}
\widehat y_i\sigma_i.
\end{equation}

\subsection{Complexity analysis}

We first analyze the resources required to construct a block encoding
of each normalized symmetric-sector operator $\sigma_i$.

\begin{lemma}
[Complexity of the block encoding of each symmetric-sector term]
\label{lem:block-encoding-single-term}
Let $\rho$ be an $n$-qubit density operator, and suppose that
$U_\rho$ efficiently prepares a purification of $\rho$.
Fix $1\leq i\leq\lambda$, and let
$0<\varepsilon_i^{(\sigma)}<1$. Choose
\begin{equation}
m_i
:=
\left\lceil
\frac{
5k_i^2(k_i-1)
}{
\left(\varepsilon_i^{(\sigma)}\right)^2
}
\right\rceil.
\label{eq:block-single-term-state-copies}
\end{equation}
Then one can construct a unitary $U_i$ that is a
$
\left(
1,\,
a_{\mathrm{BE}}^i,\,
\varepsilon_i^{(\sigma)}
\right)
$
block encoding of
$
\sigma_i
=
\Tr_{\mathrm{aux}}
\left[
\rho_{\mathrm{sym}}^{(k_i)}
\right]$,
acting on the first $k$ system registers. The number of block-encoding ancilla qubits satisfies
$
a_{\mathrm{BE}}^i
=
\operatorname{poly}(m_i,n)
$,
and the circuit complexity is
$
\operatorname{poly}\left(
m_i,
n,\log(1/\varepsilon_i^{(\sigma)})
\right)
$
Each use of $U_i$ makes $m_i$ uses of $U_\rho$ and $m_i$ uses of
$U_\rho^\dagger$.
\end{lemma}

\begin{proof}
By Corollary~\ref{purification-of-channel}, the symmetric-state
preparation channel $\mathcal M_{m_i,k_i}$ admits a purified
implementation $U_{m_i,k_i}$ whose output
\begin{equation}
\widetilde{\rho}_{\mathrm{sym}}^{(k_i)}
:=
\mathcal M_{m_i,k_i}
\left(
\rho^{\otimes m_i}
\right)
\end{equation}
satisfies
\begin{equation}
\left\|
\widetilde{\rho}_{\mathrm{sym}}^{(k_i)}
-
\rho_{\mathrm{sym}}^{(k_i)}
\right\|_1
\leq
\varepsilon_i^{(\sigma)}.
\label{eq:block-single-term-full-state-error}
\end{equation}
The input $\rho^{\otimes m_i}$ is coherently prepared by applying
$U_\rho^{\otimes m_i}$ to $m_i$ system registers and their
purification ancillas. Composing this preparation with the purified
implementation of the symmetric-state preparation channel therefore
gives a purification circuit for
$\widetilde{\rho}_{\mathrm{sym}}^{(k_i)}$.

We next regard the last $a_i$ output system registers as part
of the purification environment. The same circuit then prepares a
purification of
\begin{equation}
\widetilde{\sigma}_i
:=
\Tr_{\mathrm{aux}}
\left[
\widetilde{\rho}_{\mathrm{sym}}^{(k_i)}
\right],
\end{equation}
where the first $k$ output system registers are retained.
Since the trace norm is contractive under the partial trace,
Eq.~\eqref{eq:block-single-term-full-state-error} gives
\begin{align}
\left\|
\widetilde{\sigma}_i-\sigma_i
\right\|_1
&=
\left\|
\Tr_{\mathrm{aux}}
\left[
\widetilde{\rho}_{\mathrm{sym}}^{(k_i)}
-
\rho_{\mathrm{sym}}^{(k_i)}
\right]
\right\|_1
\nonumber\\
&\leq
\left\|
\widetilde{\rho}_{\mathrm{sym}}^{(k_i)}
-
\rho_{\mathrm{sym}}^{(k_i)}
\right\|_1
\nonumber\\
&\leq
\varepsilon_i^{(\sigma)}.
\label{eq:block-single-term-reduced-state-error}
\end{align}

Applying
Lemma~\ref{Block encoding from a purification circuit}
to this purification circuit gives an exact
$(1,a_{\mathrm{BE}}^i,0)$ block encoding of
$\widetilde{\sigma}_i$. Moreover,
\begin{equation}
\left\|
\widetilde{\sigma}_i-\sigma_i
\right\|_\infty
\leq
\left\|
\widetilde{\sigma}_i-\sigma_i
\right\|_1
\leq
\varepsilon_i^{(\sigma)}.
\end{equation}
Therefore, the same unitary is a
\begin{equation}
\left(
1,\,
a_{\mathrm{BE}}^i,\,
\varepsilon_i^{(\sigma)}
\right)
\end{equation}
block encoding of $\sigma_i$.

It remains to count the  size of ancilla systems and circuit complexity.
By Corollary~\ref{purification-of-channel}, the purified
symmetric-state preparation channel uses
$\operatorname{poly}(m_i,n)$ ancilla qubits and has circuit
complexity
\begin{equation}
\operatorname{poly}\left(
m_i,\,
n,\,
\log\frac{1}{\varepsilon_i^{(\sigma)}}
\right).
\end{equation}
The $a_i$ system registers traced out in the definition of
$\widetilde{\sigma}_i$ are included among the purification ancillas.
Applying Lemma~\ref{Block encoding from a purification circuit}
requires an additional $kn$-qubit signal register on which
$\widetilde{\sigma}_i$ is encoded. The registers used by the
purification circuit, including its original $kn$-qubit output
register and its purification environment, serve as the
block-encoding ancillas.
Since $k\leq k_i<m_i$, the total number of block-encoding ancilla
qubits remains
\begin{equation}
a_{\mathrm{BE}}^i
=
\operatorname{poly}(m_i,n).
\end{equation}

Finally, one use of the block-encoding unitary consists of one use of
the purification circuit, one use of its inverse, and a SWAP operation
on the $k$ retained system registers. The purification circuit
contains $U_\rho^{\otimes m_i}$, while its inverse contains
$U_\rho^{\dagger\otimes m_i}$. Hence, each use of $U_i$ makes
$m_i$ uses of $U_\rho$ and $m_i$ uses of $U_\rho^\dagger$, and its
circuit complexity is
\begin{equation}
\operatorname{poly}\left(
m_i,\,
n,\,
\log\frac{1}{\varepsilon_i^{(\sigma)}}
\right).
\end{equation}
This completes the proof.
\end{proof}

Combining all the ingredients, we obtain the complexity of the
complete block-encoding construction.

\begin{theorem}
[Complexity of the block encoding of the Scrooge $k$-th moment]
\label{thm:block-encoding-PA}
Let $\rho$ be an $n$-qubit density operator, let $k\geq2$, and suppose
that $12(k-1)^2\leq\M$. Let $0<\epsilon_{\mathrm t}<1$ satisfy
\begin{equation}
\epsilon_{\mathrm t}
\geq
2C_{\add}
\exp\left[
-\frac{\M}{24(k-1)^2}
\right],
\label{eq:block-error-threshold}
\end{equation}
where $C_{\add}=e^{11/12}$. 
Assume the query access to state-preparation circuit $U_\rho$ and $U_\rho^\dagger$.
Let $\delta$ be the failure probability of the whole construction. Our construction takes the following two steps.

\begin{itemize}
    \item We first estimate all
$p_{\mathrm{sym}}^{(k_i)}$ precisely using Lemma~\ref{lem:approximate-schur-mom} and
\begin{align}
N_{\mathrm{prob}}^{\mathrm{suff}}
=
\mathcal O\Bigg(
&k^3
\log^4\left(
\frac{2C_{\add}}{\epsilon_{\mathrm t}}
\right)
\left[
\log\frac1\delta
+
\log\log\left(
\frac{2C_{\add}}{\epsilon_{\mathrm t}}
\right)
\right]
\left(
\frac{2C_{\add}}{\epsilon_{\mathrm t}}
\right)^5
\Bigg).
\label{eq:block-probability-copy-complexity}
\end{align}
number of copies of $\rho$.
\item We then proceed with the construction outlined in Section~\ref{app:be:procedure}. The procedure 
constructs an
$(\alpha,a_{\mathrm{BE}},\epsilon_{\mathrm t})$ block encoding of
$\chi_{\mathrm{Sc}}^{(k)}(\rho)$, where
\begin{equation}
\alpha
=
\mathcal O\left[
\left(
\frac{2C_{\add}}{\epsilon_{\mathrm t}}
\right)^2
\right].
\label{eq:block-normalization-complexity}
\end{equation}
The number of block-encoding ancilla qubits satisfies
\begin{equation}
a_{\mathrm{BE}}
=
\operatorname{poly}\left(
n,k,\frac1{\epsilon_{\mathrm t}}
\right).
\label{eq:block-final-ancilla-complexity}
\end{equation}
The query complexity of $U_\rho$ and $U_\rho^\dagger$ is at most
\begin{equation}
Q_{\rho}^{\mathrm{BE}}
=
\operatorname{poly}\left(
k,\frac1{\epsilon_{\mathrm t}}
\right)
\label{eq:block-final-purification-query-complexity}
\end{equation}
\end{itemize}

The total circuit complexity of the entire process is
\begin{equation}
\operatorname{poly}\left(
n,k,\frac1{\epsilon_{\mathrm t}},
\log\frac1\delta
\right).
\label{eq:block-final-circuit-complexity}
\end{equation}

\end{theorem}

\begin{proof}
We allocate an error budget of $\epsilon_{\mathrm t}/2$ to the
polynomial approximation and the remaining $\epsilon_{\mathrm t}/2$
to the block-encoding construction.

By Corollary~\ref{coro:minimal-additive-error}, it suffices to choose
\begin{equation}
\lambda
:=
\max\left\{
\left\lfloor
2\log\left(
\frac{2C_{\add}}{\epsilon_{\mathrm t}}
\right)
\right\rfloor,
1
\right\}
\label{eq:block-lambda-choice}
\end{equation}
which guarantees that
\begin{equation}
\left\|
\chi_{\mathrm{Sc}}^{(k)}(\rho)
-
\chi_{\mathrm{PA}}^{(k)}(\lambda)
\right\|_1
\leq
\frac{\epsilon_{\mathrm t}}2.
\label{eq:block-polynomial-error}
\end{equation}

For every $1\leq i\leq\lambda$, choose
\begin{equation}
\varepsilon_i^{(p)}
=
\varepsilon_i^{(\sigma)}
:=
\frac{\epsilon_{\mathrm t}}
{4(2^\lambda-1)}
\exp\left[
-\frac{k_\lambda(k_\lambda-1)}{2\M}
\right],
\qquad
\delta_i^{(p)}
:=
\frac{\delta}{\lambda},
\label{eq:block-individual-parameters}
\end{equation}
where $k_\lambda=\lambda(k-1)+1$.

We first estimate $p_{\mathrm{sym}}^{(k_i)}$.
Apply Lemma~\ref{lem:approximate-schur-mom} with statistical relative
accuracy $\varepsilon_i^{(p)}/2$, failure probability
$\delta_i^{(p)}$, and Schur-transform precision
\begin{equation}
\epsilon_{\mathrm{Sch}}^{i,(p)}
:=
\min\left\{
\frac{\varepsilon_i^{(p)}}{2k_i!},
\frac{1}{(k_i!)^2}
\right\}.
\label{eq:block-probability-schur-precision}
\end{equation}
Lemma~\ref{lem:approximate-schur-mom} then gives
\begin{equation}
\Pr\left[
\left|
\widehat p_{\mathrm{sym}}^{(k_i)}
-
p_{\mathrm{sym}}^{(k_i)}
\right|
>
\varepsilon_i^{(p)}
p_{\mathrm{sym}}^{(k_i)}
\right]
\leq
\delta_i^{(p)}.
\label{eq:block-probability-estimation-error}
\end{equation}
By the union bound, with probability at least $1-\delta$,
\begin{equation}
\left|
\widehat p_{\mathrm{sym}}^{(k_i)}
-
p_{\mathrm{sym}}^{(k_i)}
\right|
\leq
\varepsilon_i^{(p)}
p_{\mathrm{sym}}^{(k_i)}
\label{eq:block-simultaneous-probability-error}
\end{equation}
holds simultaneously for all $1\leq i\leq\lambda$. We condition on
this event for the remainder of the proof.

For each $i$, choose
\begin{equation}
m_i
:=
\left\lceil
\frac{
5k_i^2(k_i-1)
}{
\left(\varepsilon_i^{(\sigma)}\right)^2
}
\right\rceil.
\label{eq:block-state-preparation-copies}
\end{equation}
By Lemma~\ref{lem:block-encoding-single-term}, there exists a unitary
$U_i$ that is a
$(1,a_{\mathrm{BE}}^i,\varepsilon_i^{(\sigma)})$
block encoding of $\sigma_i$, with
$a_{\mathrm{BE}}^i=\operatorname{poly}(m_i,n)$ and circuit complexity
\[
\operatorname{poly}\left(
m_i,n,\log\frac{1}{\varepsilon_i^{(\sigma)}}
\right).
\]
More precisely, the construction in the proof of that lemma gives an
exact block encoding of a density operator
$\widetilde{\sigma}_i$ satisfying
\begin{equation}
\left\|
\widetilde{\sigma}_i-\sigma_i
\right\|_1
\leq
\varepsilon_i^{(\sigma)}.
\label{eq:block-sigma-error}
\end{equation}

Let
\begin{equation}
s
:=
\max_{1\leq i\leq\lambda}
a_{\mathrm{BE}}^i,
\qquad
r
:=
\left\lceil
\log_2\lambda
\right\rceil.
\label{eq:block-common-ancilla}
\end{equation}
By appending idle ancilla qubits, each $U_i$ can be regarded as an
exact $(1,s,0)$ block encoding of $\widetilde{\sigma}_i$.

Recall that
\begin{equation}
\widehat y_i
=
\binom{\lambda}{i}
(-1)^{i-1}
c_i
\widehat p_{\mathrm{sym}}^{(k_i)}.
\end{equation}
Define
\begin{equation}
\alpha
:=
\sum_{i=1}^{\lambda}
|\widehat y_i|
=
\sum_{i=1}^{\lambda}
\binom{\lambda}{i}
c_i
\widehat p_{\mathrm{sym}}^{(k_i)}.
\label{eq:block-lcu-normalization}
\end{equation}
An exact state-preparation pair for these coefficients is
\begin{align}
P_L\ket{0}^{\otimes r}
&=
\sum_{i=1}^{\lambda}
\sqrt{
\frac{
\binom{\lambda}{i}
c_i
\widehat p_{\mathrm{sym}}^{(k_i)}
}{
\alpha
}
}
\ket{i-1},
\\
P_R\ket{0}^{\otimes r}
&=
\sum_{i=1}^{\lambda}
(-1)^{i-1}
\sqrt{
\frac{
\binom{\lambda}{i}
c_i
\widehat p_{\mathrm{sym}}^{(k_i)}
}{
\alpha
}
}
\ket{i-1}.
\end{align}
Hence, $(P_L,P_R)$ is an $(\alpha,r,0)$-state-preparation pair for
$(\widehat y_1,\ldots,\widehat y_\lambda)$.

Applying Lemma~\ref{lem:lc-BE} gives an exact
$(\alpha,a_{\mathrm{BE}},0)$ block encoding of
\begin{equation}
\sum_{i=1}^{\lambda}
\widehat y_i\widetilde{\sigma}_i,
\end{equation}
where
\begin{equation}
a_{\mathrm{BE}}
=
s+\left\lceil\log_2\lambda\right\rceil.
\label{eq:block-total-ancilla}
\end{equation}

We next compare the encoded operator with
$\chi_{\mathrm{PA}}^{(k)}(\lambda)$. Since
$\widetilde{\sigma}_i$ is a density operator,
$\lVert\widetilde{\sigma}_i\rVert_\infty\leq1$. Therefore,
Eqs.~\eqref{eq:block-simultaneous-probability-error} and
\eqref{eq:block-sigma-error} imply
\begin{align}
&
\left\|
\sum_{i=1}^{\lambda}
\widehat y_i\widetilde{\sigma}_i
-
\chi_{\mathrm{PA}}^{(k)}(\lambda)
\right\|_\infty
\nonumber\\
&\quad\leq
\sum_{i=1}^{\lambda}
\binom{\lambda}{i}c_i
\left[
\left|
\widehat p_{\mathrm{sym}}^{(k_i)}
-
p_{\mathrm{sym}}^{(k_i)}
\right|
\left\|
\widetilde{\sigma}_i
\right\|_\infty
+
p_{\mathrm{sym}}^{(k_i)}
\left\|
\widetilde{\sigma}_i-\sigma_i
\right\|_\infty
\right]
\nonumber\\
&\quad\leq
\sum_{i=1}^{\lambda}
\binom{\lambda}{i}
c_i p_{\mathrm{sym}}^{(k_i)}
\left(
\varepsilon_i^{(p)}
+
\varepsilon_i^{(\sigma)}
\right).
\label{eq:block-total-construction-error}
\end{align}
Here we have also used
$\lVert X\rVert_\infty\leq\lVert X\rVert_1$.

By Lemma~\ref{lem:weighted-symmetric-probability},
\begin{equation}
c_i p_{\mathrm{sym}}^{(k_i)}
\leq
\exp\left[
\frac{k_i(k_i-1)}{2\M}
\right]
\leq
\exp\left[
\frac{k_\lambda(k_\lambda-1)}{2\M}
\right].
\label{eq:block-weight-bound}
\end{equation}
Using Eq.~\eqref{eq:block-individual-parameters} and
$\sum_{i=1}^{\lambda}\binom{\lambda}{i}=2^\lambda-1$, we obtain
\begin{equation}
\left\|
\sum_{i=1}^{\lambda}
\widehat y_i\widetilde{\sigma}_i
-
\chi_{\mathrm{PA}}^{(k)}(\lambda)
\right\|_\infty
\leq
\frac{\epsilon_{\mathrm t}}{2}.
\label{eq:block-polynomial-construction-error}
\end{equation}

Combining this with Eq.~\eqref{eq:block-polynomial-error} and using
$\lVert X\rVert_\infty\leq\lVert X\rVert_1$, we obtain
\begin{equation}
\left\|
\sum_{i=1}^{\lambda}
\widehat y_i\widetilde{\sigma}_i
-
\chi_{\mathrm{Sc}}^{(k)}(\rho)
\right\|_\infty
\leq
\epsilon_{\mathrm t}.
\end{equation}
Hence, the constructed unitary is an
$(\alpha,a_{\mathrm{BE}},\epsilon_{\mathrm t})$ block encoding of
$\chi_{\mathrm{Sc}}^{(k)}(\rho)$.

We now bound $\alpha$. Assuming
\eqref{eq:block-simultaneous-probability-error} holds,
\begin{align}
\alpha
&\leq
\left(
1+\max_{1\leq i\leq\lambda}\varepsilon_i^{(p)}
\right)
\sum_{i=1}^{\lambda}
\binom{\lambda}{i}
c_i p_{\mathrm{sym}}^{(k_i)}
\nonumber\\
&\leq
2(2^\lambda-1)
\exp\left[
\frac{k_\lambda(k_\lambda-1)}{2\M}
\right].
\label{eq:block-alpha-preliminary}
\end{align}
Equation~\eqref{eq:block-error-threshold} implies
\begin{equation}
\log\left(
\frac{2C_{\add}}{\epsilon_{\mathrm t}}
\right)
\leq
\frac{\M}{24(k-1)^2}.
\label{eq:block-log-upper-bound}
\end{equation}
Moreover, Eq.~\eqref{eq:block-lambda-choice} gives
\begin{equation}
\lambda
\leq
2\log\left(
\frac{2C_{\add}}{\epsilon_{\mathrm t}}
\right),
\qquad
k_\lambda
\leq
2(k-1)
\log\left(
\frac{2C_{\add}}{\epsilon_{\mathrm t}}
\right)+1.
\label{eq:block-lambda-upper-bound}
\end{equation}
It follows that
\begin{equation}
\frac{k_\lambda(k_\lambda-1)}{\M}
\leq
\frac16
\log\left(
\frac{2C_{\add}}{\epsilon_{\mathrm t}}
\right)
+
\frac1{12}.
\label{eq:block-exponential-bound}
\end{equation}
Substituting these bounds into
Eq.~\eqref{eq:block-alpha-preliminary} yields
\begin{align}
\alpha
&\leq
2e^{1/24}
\left(
\frac{2C_{\add}}{\epsilon_{\mathrm t}}
\right)^{2\log2+1/12}
\nonumber\\
&=
\mathcal O\left[
\left(
\frac{2C_{\add}}{\epsilon_{\mathrm t}}
\right)^2
\right],
\end{align}
which proves Eq.~\eqref{eq:block-normalization-complexity}.

We next bound the copies required to estimate the coefficients.
Lemma~\ref{lem:approximate-schur-mom} gives
\begin{align}
N_i^{(p)}
&=
k_i^2(k_i-1)
\left\lceil
\frac{24(e+1)}
{\left(\varepsilon_i^{(p)}\right)^2}
\right\rceil
\left[
2\left\lceil
\frac94
\ln\frac{\lambda}{\delta}
-\frac12
\right\rceil
+1
\right].
\end{align}
Consequently,
\begin{equation}
N_{\mathrm{prob}}^{\mathrm{suff}}
:=
\sum_{i=1}^{\lambda}N_i^{(p)}
=
\mathcal O\left(
\frac{
\lambda k_\lambda^3
}{
\left(\varepsilon_i^{(p)}\right)^2
}
\log\frac{\lambda}{\delta}
\right).
\label{eq:block-total-probability-copies}
\end{equation}
Using Eqs.~\eqref{eq:block-individual-parameters},
\eqref{eq:block-lambda-upper-bound}, and
\eqref{eq:block-exponential-bound}, we obtain
\begin{align}
N_{\mathrm{prob}}^{\mathrm{suff}}
=
\mathcal O\Bigg(
&k^3
\log^4\left(
\frac{2C_{\add}}{\epsilon_{\mathrm t}}
\right)
\left[
\log\frac1\delta
+
\log\log\left(
\frac{2C_{\add}}{\epsilon_{\mathrm t}}
\right)
\right]
\left(
\frac{2C_{\add}}{\epsilon_{\mathrm t}}
\right)^5
\Bigg),
\end{align}
which proves
Eq.~\eqref{eq:block-probability-copy-complexity}.

Finally, the largest $m_i$ is attained at $i=\lambda$. From
Eqs.~\eqref{eq:block-state-preparation-copies} and
\eqref{eq:block-individual-parameters},
\begin{equation}
m_\lambda
=
\mathcal O\left(
\frac{
k_\lambda^3(2^\lambda-1)^2
}{
\epsilon_{\mathrm t}^2
}
\exp\left[
\frac{k_\lambda(k_\lambda-1)}{\M}
\right]
\right)
=
\operatorname{poly}\left(
k,\frac1{\epsilon_{\mathrm t}}
\right).
\label{eq:block-largest-state-preparation-size}
\end{equation}
Lemma~\ref{lem:block-encoding-single-term} therefore gives
$s=\operatorname{poly}(m_\lambda,n)$.
Together with Eq.~\eqref{eq:block-total-ancilla}, this proves
\begin{equation}
a_{\mathrm{BE}}
=
\operatorname{poly}\left(
n,k,\frac1{\epsilon_{\mathrm t}}
\right).
\end{equation}

It remains to count the calls to the purification circuit.
For the $i$-th term, the block-encoding unitary $U_i$ consists of
one use of a purification circuit for
$\widetilde{\sigma}_i$, one use of its inverse, and a SWAP operation.
The forward purification circuit contains
$U_\rho^{\otimes m_i}$, whereas its inverse contains
$U_\rho^{\dagger\otimes m_i}$. Hence, one use of $U_i$ requires
$m_i$ uses of $U_\rho$ and $m_i$ uses of $U_\rho^\dagger$.

In the direct implementation of the controlled unitary in the LCU
construction, one controlled copy of $U_i$ is included for each
$1\leq i\leq\lambda$. Therefore, the query costs of the individual
terms add, and one use of the resulting block-encoding unitary
requires
\begin{equation}
Q_{\rho}^{\mathrm{BE}}
=
\sum_{i=1}^{\lambda}m_i
\leq
\lambda m_\lambda
\end{equation}
uses of $U_\rho$ and the same number of uses of $U_\rho^\dagger$.

Using Eq.~\eqref{eq:block-largest-state-preparation-size}, we obtain
\begin{align}
Q_{\rho}^{\mathrm{BE}}
&=
\mathcal O\left(
\frac{
\lambda k_\lambda^3(2^\lambda-1)^2
}{
\epsilon_{\mathrm t}^2
}
\exp\left[
\frac{k_\lambda(k_\lambda-1)}{\M}
\right]
\right)
\nonumber\\
&=
\mathcal O\Bigg(
k^3
\log^4\left(
\frac{2C_{\add}}{\epsilon_{\mathrm t}}
\right)
\left(
\frac{2C_{\add}}{\epsilon_{\mathrm t}}
\right)^5
\Bigg),
\end{align}
which proves
Eq.~\eqref{eq:block-final-purification-query-complexity}.

The controlled unitary $W$ appearing in
Lemma~\ref{lem:lc-BE} applies $U_i$ conditioned on the state
$\ket{i-1}$ of the index register. A direct implementation of $W$ consists of $\lambda$ controlled
block-encoding circuits and therefore has circuit complexity
polynomial in $\lambda$ and in the circuit complexities of the
unitaries $U_i$. The state-preparation circuits
$P_L$ and $P_R$ each have circuit complexity $O(\lambda)$. Moreover,
estimating the coefficients requires the number of Schur-label
measurement circuits given in
Eq.~\eqref{eq:block-total-probability-copies}. Combining the costs of
coefficient estimation, state preparation, and the single-term block
encodings, the total circuit complexity of the construction is
\begin{equation}
\operatorname{poly}\left(
n,k,\frac{1}{\epsilon_{\mathrm t}},
\log\frac{1}{\delta}
\right).
\end{equation}
\end{proof}

\end{document}